\pdfoutput=1
\documentclass[letterpaper]{article} 
\usepackage[preprint]{aaai2027}  
\usepackage[hyphens]{url}  
\usepackage{graphicx} 
\usepackage{natbib}  
\usepackage{caption} 
\usepackage{booktabs}
\usepackage{multirow}
\usepackage{amsmath}
\usepackage{amssymb}
\title{Agogic: Performance-Timed Music Tokens for LLM-Native Text-to-Symbolic-Music Generation}
\author{
    Junhao Chen\textsuperscript{1},
    Mingjin Chen\textsuperscript{2},
    Jingjia Mao\textsuperscript{1},
    Lin Chen\textsuperscript{3},
    Saining Zhang\textsuperscript{4},
    Minglin Chen\textsuperscript{5},
    Ruocheng Wu\textsuperscript{6},
    Liaoyuan Fan\textsuperscript{6},
    Wenyi Li\textsuperscript{7},
    Mingju Gao\textsuperscript{8},
    Henghaofan Zhang\textsuperscript{9}, \\
    Zhihao Li\textsuperscript{10},
    Hao Zhao\textsuperscript{1},
    Yufei Wang\textsuperscript{10,\dag},
    Ruqi Huang\textsuperscript{1,\dag}
}
\affiliations{
    \textsuperscript{1}Tsinghua University \quad
    \textsuperscript{2}The Hong Kong Polytechnic University \\
    \textsuperscript{3}Beijing Technology and Business University \quad
    \textsuperscript{4}Nanyang Technological University \\
    \textsuperscript{5}Sun Yat-sen University \quad
    \textsuperscript{6}The University of Hong Kong \quad
    \textsuperscript{7}University of Chinese Academy of Sciences \\
    \textsuperscript{8}Peking University \quad
    \textsuperscript{9}University of Electronic Science and Technology of China \quad
    \textsuperscript{10}SparcAI Inc. \\[3pt]
    Project page: \clickurl{https://yisuanwang.github.io/Agogic}
}

\newtheorem{proposition}{Proposition}
\newenvironment{proof}{\par\noindent\textit{Proof.}\ }{\hfill$\square$\par\medskip}
\newcommand\blfootnote[1]{%
  \begingroup\renewcommand\thefootnote{}\footnote{#1}%
  \addtocounter{footnote}{-1}\endgroup}
\newcommand{\clickurl}[1]{%
  \pdfstartlink attr{/Border[0 0 0]}%
    user{/Subtype/Link/A<</Type/Action/S/URI/URI(#1)>>}%
  \texttt{#1}\pdfendlink}
\begin{document}
\maketitle

\begin{abstract}
\blfootnote{$^{\dag}$Corresponding authors.}
Building a text-to-music language model begins with a choice usually made by default: how to \emph{tokenize} music. Because that choice is normally entangled with the backbone, data, and training recipe, its effect has never been measured in isolation. We fix all three, pretrained Qwen3.5 (0.8B--27B), data, budget, and decoding, and swap \emph{only} the representation across seven tokenizations, anchoring texture metrics to each representation's model-free expressive ceiling. The resulting ordering is clean and, we think, surprising: \textbf{the representation, not the model's size, is the binding variable for distributional fidelity}. Scaling the backbone $34\times$ (0.8B$\to$27B) barely moves Fr\'echet Music Distance, whereas switching representation halves it. A performance-resolution stream we release, PMT (10\,ms timing, per-note velocity, multi-track texture; 609 symbols), reaches FMD $159$ at 0.8B against $272$--$286$ for beat grids ($1.7$--$1.8\times$ lower at this protocol, up to $2.8\times$ at others, non-overlapping bootstrap confidence intervals), so \emph{a 0.8B performance-resolution model beats a 27B beat grid}. The ordering reappears on a 26M from-scratch backbone and on a second, independent performance-resolution tokenizer, making it a property of the representation \emph{class}, not one lucky vocabulary. It is robust, not a finer-lattice artifact: snapping PMT's generated onsets to a fixed grid at the beat grids' resolution still leaves it $67$--$129$ FMD points ahead of \emph{both} ($n{=}500$). We are precise about scope. The effect is distributional; whether it is \emph{audible} is a separate question, which our automatic probe leaves open and a pre-registered, properly-powered human study is collecting to settle. Native caption adherence is weak but \emph{separable}: a lightweight decode-time constraint doubles instrument-F1 ($.28\!\to\!.60$) and Correct-Key ($.16\!\to\!.35$) at no distributional cost. Alongside the harness and 25+ checkpoints we release two corpora, an 86.6k set aligned across caption/MIDI/ABC/audio and a 6.25M captioned corpus (the largest for music to date, symbolic or audio; only \emph{general}-audio corpora are larger), and an imprinting diagnostic: published text-to-MIDI systems reproduce their training distribution near-invariant to the caption ($72\%$ vs.\ $71\%$ chord-time on two disjoint domains). The field's next representation claim can now be measured, not asserted.
\end{abstract}

\section{Introduction}
Text-to-symbolic-music generation maps a natural-language description (``a lively Irish jig with flowing eighth-note runs in D major'') to a machine-readable score that can be played, edited and re-orchestrated. Unlike text-to-audio synthesis, symbolic output preserves note-level structure, underpinning composition assistance, adaptive scoring and music education; it is hard because the model must turn underspecified perceptual vocabulary into precise pitch, timing, velocity and multi-instrument texture over hundreds of correlated events.

Existing LLM-based systems make one of two representational compromises. Beat-grid tokenizations (REMI and descendants~\citep{huang2020pop,fradet2023miditok}) snap onsets and durations to metrical positions: robust, but quantization discards the micro-timing and per-note dynamics that distinguish a performance from a mechanical score. Text-native systems (ChatMusician, MuPT, NotaGen~\citep{yuan2024chatmusician,qu2025mupt,wang2025notagen}) serialize scores as ABC to exploit the LLM's textual prior, arguing from token compactness (${\sim}5\times$ fewer tokens/song than REMI) that it is the better LLM format; that is compelling for music \emph{understanding} but untested for high-fidelity multi-instrument \emph{generation} under controlled conditions. Our measurements suggest compactness does not convert into generative quality: ABC models produce half the polyphony of token models at every scale despite perfect syntax and are furthest from real music (FMD $406$ vs.\ PMT $159$), because ABC has no native fields for the performance parameters that generation fidelity rewards. Dedicated text-to-MIDI systems (text2midi, MIDI-LLM~\citep{bhandari2025text2midi,wu2025midi}) each commit to one representation and recipe and, as our identical-caption evaluation reveals, can imprint their training distribution onto every output regardless of the caption.

We take a different position: \emph{the representation should preserve the performance}. If music is serialized with its performance parameters intact (fine-grained timing, per-note velocity, full multi-track texture), an LLM can learn music as a foreign language while retaining the fine timing and per-note dynamics a beat grid discards; we validate this retention \emph{distributionally} (the perceptual consequence we probe but leave open). We instantiate it as \textbf{PMT} (Performance-timed Music Tokens), released together with the corpora and harness under the project name \textsc{Agogic} after the musical term for the expressive timing deviation that a beat grid removes. Its core idea is \emph{performance-resolution parametric tokenization}: every note becomes a short sequence of interpretable parameter tokens (pitch, duration, velocity, 10\,ms time-shift), multi-instrument music flows in one flat stream, and a pretrained LLM consumes it through a small vocabulary extension.

\begin{figure*}[t]
\centering
\includegraphics[width=\textwidth]{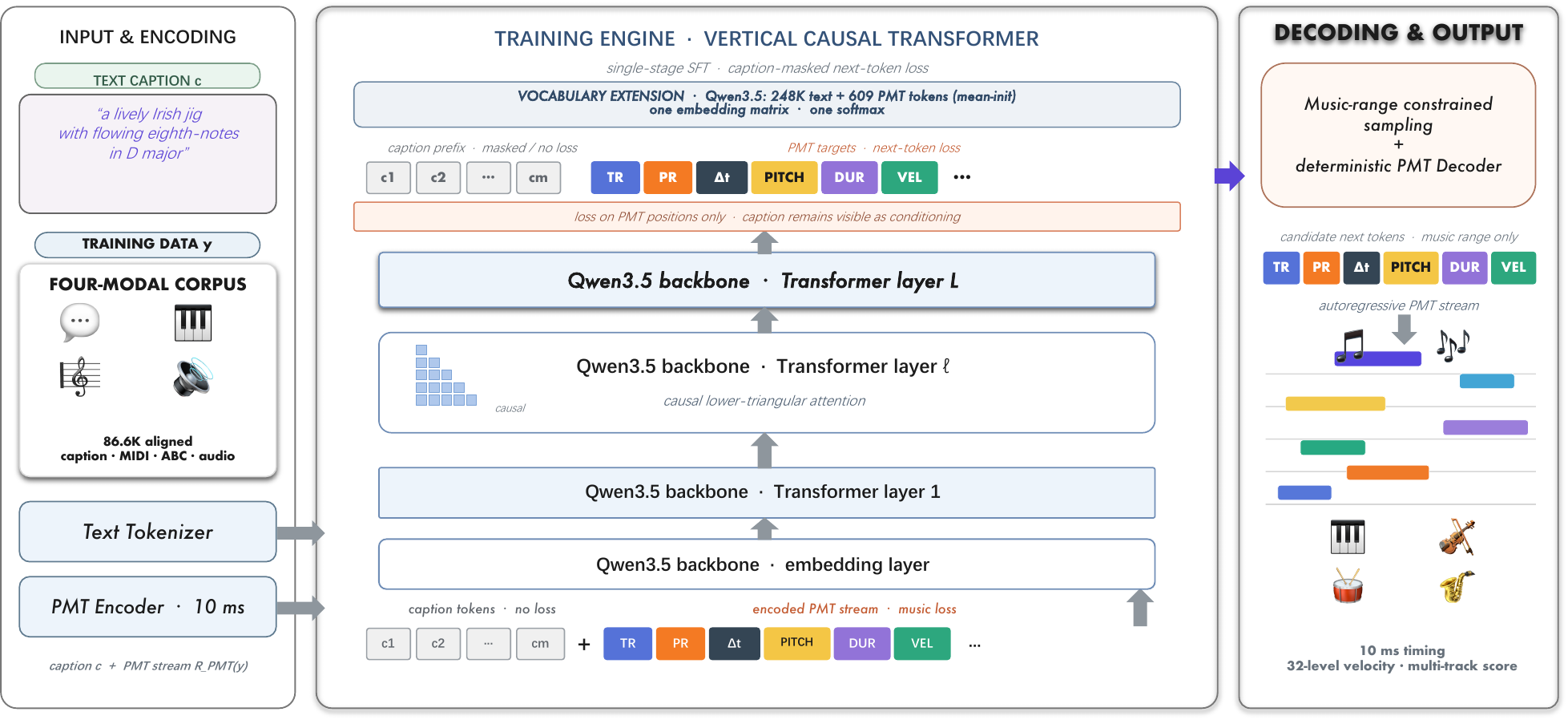}
\caption{\textbf{PMT overview.} \emph{Left:} a caption paired with a piece from the four-modality corpus. \emph{Middle:} the piece as a performance-timed token stream. \emph{Right:} decoding back to MIDI.}
\label{fig:pipeline}
\end{figure*}

PMT comprises three components. The \textbf{performance-timed parametric representation} keeps 10\,ms time-shifts and 32-level velocities (which beat grids quantize away) in a compact 609-symbol vocabulary covering multi-track music in one stream; its round-trip retains 39.1\% chord-time versus 36.7\% raw, among the highest we test; but this chord-time proxy is subset-sensitive, so we rest expressiveness on the robust round-trip onset error ($2.3$\,ms vs.\ $53$\,ms), not chord-time. The \textbf{LLM integration recipe} makes it trainable at scale: mean-initialized vocabulary extension, caption-masked loss, and music-range-constrained decoding, run unchanged from 0.8B to 27B. The \textbf{ceiling-anchored evaluation protocol} proves the design choices rather than asserting them: competing representations (beat-grid TSD, REMI, MIDI-Like, ABC) are trained under the \emph{same} backbone, data, budget, and decoding, and every texture metric is normalized by each representation's own round-trip ceiling, separating what a format can express from what a model learns to use.

Across 25+ trained models, scale barely moves distributional fidelity while the representation moves it decisively. A performance-resolution stream roughly halves the Fr\'echet Music Distance of every beat grid ($152$--$159$ vs.\ $271$--$286$) at every scale from 0.8B to 27B, with non-overlapping bootstrap intervals, and a 0.8B performance-resolution model beats a 27B beat grid: a $34\times$ parameter increase at a matched budget does not overturn the ordering. The ordering reappears on a 26M from-scratch backbone, ruling out a pretraining-prior artifact, and because that model is far better converged than the 27B at the shared budget the effect is not merely large-model undertraining. Our controls then \emph{locate} the effect rather than over-attribute it. Velocity is held at 32 levels across \emph{all} arms, so the swap isolates timing quantization. The raw onset-interval gap (JSD$_{ioi}$ $.03$ vs.\ $.15$) is largely \emph{representability} and closes once every arm is coarse-quantized to a common $60$\,ms grid, but the FMD gap \emph{survives} that coarsening (${\approx}67$ points ahead of both beat grids, $n{=}500$), so it is a robust distributional difference rather than the finer lattice. Per marginal, beat grids in fact match real note density and polyphony better, so we characterize the gap holistically and make no perceptual claim. On identical captions PMT tracks the reference note distribution, whereas MIDI-LLM and text2midi imprint their training data regardless of the caption.

Our contributions are threefold.\\
\textbf{(1) A released tokenizer, PMT.} A $609$-symbol vocabulary that streams multi-track music at ${\approx}4$ tokens/note while preserving $10$\,ms micro-timing and per-note velocity. It extends the performance-event lineage~\citep{oore2020time} with per-note velocity and multi-track program structure, which the closest LLM-vocabulary system MIDI-LLM omits, at $4$ vs.\ PerTok's $6.8$ tokens/note. We claim the controlled swap as novel, not the encoding.\\
\textbf{(2) Two released datasets.} A four-modality aligned corpus of $86{,}598$ real-music pieces, each with an expert-length English caption, MIDI, ABC and rendered audio, to our knowledge the first captioned symbolic corpus aligned per case across all four; and a $6.25$M scaled captioned corpus, the largest for music outside \emph{general}-audio collections.\\
\textbf{(3) A controlled benchmark with released harness.} The first \emph{ceiling-anchored, cross-family} representation comparison for text-conditioned symbolic generation, extending within-family swaps~\citep{fradet2023impact}. It finds that the \emph{representation Pareto-dominates model scale} for distributional fidelity, reproduced on a 26M from-scratch backbone and robust to seeds, and it yields an \emph{imprinting diagnostic} that distance-to-reference metrics expose: published text-to-MIDI systems replay their training distribution near-invariant to the caption (MIDI-LLM emits $72\%$ vs.\ $71\%$ chord-time on two disjoint domains).

The lens is not specific to music: it applies to any LLM extended to a new tokenized modality, echoing the tokenization debate in language modeling~\citep{sennrich2016neural,xue2022byt5,yu2023megabyte} and the same question now being asked of other structured artifacts a model can be taught to emit, from vector paths to character rigs and reconstructed 4D scenes~\citep{Chen_2026_CVPR_LottieGPT,wu2023iconshop,xu2020rignet,chen2026ovow}.

\section{Related Work}

\subsection{Symbolic Music Representations}
Symbolic music is serialized in three families: MIDI-event streams (performance events with fine time shifts~\citep{oore2020time}, bar/position grids in REMI~\citep{huang2020pop}, duration-based TSD standardized by MidiTok~\citep{fradet2023miditok}, and the decoders built over them~\citep{huang2019musictransformer,yu2022museformer,wang2024whole,shih2022theme,liu2022symphony}); grouped attribute tuples~\citep{hsiao2021compound,zeng2021musicbert}, which cost multi-head architectures; and text notations using ABC as plain text~\citep{wu2023tunesformer,yuan2024chatmusician,qu2025mupt,wang2025notagen,deng2024composerx}. Multi-track generation has its own line~\citep{ens2020mmm,dong2023multitrack,von2023figaro,malandro2024composer,xu2020rignet}; pressure to make the encoding itself cheap recurs in other sequence domains~\citep{chen2023towards,xu2026efficient,chen2026exploration}, as does full-song and long-sequence modelling~\citep{chen2025segment,dai2026pushing}, and the same parameterise-or-render choice recurs for character rigs and object skeletons~\citep{sun2025drive,xu2020rignet,Sun_2026_CVPR_Animator}, for reconstructed scenes and 3D assets, where dense fields stand apart from structured, editable output~\citep{poole2022dreamfusion}, for interaction video~\citep{owens2016visually,Chen_2026_CVPR_hvg3d}, for vector graphics~\citep{wu2023iconshop,xing2024svgfusion}, and for generated avatars, where geometry and texture are produced apart~\citep{chen2026ultraman}; PMT instead flattens all tracks into one performance-timed stream consumable by an off-the-shelf single-softmax LLM, descending from the performance-event lineage of \citet{oore2020time}. Learned discrete codes are a third option: MuseTok quantizes bar-wise segments with an RQ-VAE~\citep{huang2026musetok}, as discrete latent vocabularies do for other structured artifacts~\citep{weng2026garmentgpt}, and Pianoroll-Event mixes piano-roll frames with event tokens for encoding efficiency~\citep{qian2026pianoroll}; window-level latents~\citep{yao2026arima}, delay-based token scheduling~\citep{wang2026time} and 2D spectrogram grids~\citep{cheng2026modeling} continue to widen the design space; both target the encoding itself, whereas we hold the encoding fixed per arm and compare families. Prior comparisons stay within one family~\citep{fradet2023impact,wang2024exploring,qian2026beat}, and the editable-parameterisation goal is shared beyond music, from editable illustrations and character rigs to part-level 3D editing and engine-native scene reconstruction~\citep{chen2026Bunraku,xu2020rignet,weng2026feedforward3deditinglearns,poole2022dreamfusion,chen2026enginenativeeditable3dworld} or benchmark representations for \emph{classification and understanding}~\citep{zhang2023symbolic,strepetov2025symurbench} rather than the controlled, cross-family, text-conditioned \emph{generation} we study.

\subsection{Text-Conditioned Symbolic Music Generation}
Text-to-MIDI systems pair a caption encoder with a music decoder: text2midi couples FLAN-T5 with a REMI decoder~\citep{bhandari2025text2midi}; MIDI-LLM extends Llama-3.2's vocabulary with MIDI events~\citep{wu2025midi}; MuseCoco routes free text through musical attributes~\citep{lu2023musecoco}, and attribute- or emotion-conditioned lines control the same axes without free text~\citep{kang2023emogen,huang2024emotion}; the ABC line reaches text conditioning by instruction tuning or metadata prompts~\citep{wu2022exploring,lu2023musecoco}; a further route keeps the artifact in a language the model already writes and lets an interpreter realise it, whether code, vector paths or an edit instruction~\citep{chen2026paircoder,wu2023iconshop,chen2025idea23d,xing2024svgfusion,chen2026paircoderplus,lu2023musecoco,chen2023soulstyler}. Each entangles its representation with its data and recipe, so published numbers cannot isolate either. PMT is an LLM-vocabulary-extension method like MIDI-LLM, which likewise encodes fine ($10$\,ms) arrival time~\citep{thickstun2023anticipatory}, but also preserves per-note velocity (which MIDI-LLM omits) and multi-track program structure and, more importantly, differs in \emph{validation}: we hold backbone, data and budget fixed and swap only the representation, and running MIDI-LLM through our harness on identical captions reveals that it imprints its training distribution ($72\%$ chord-time against a $36.7\%$ reference) where PMT matches within 1pp. Very recent systems continue both lines, symbolic~\citep{su2026amadeus,bhandari2026text2score} and audio~\citep{copet2023simple,liu2024audioldm,evans2025stable,yuan2026yue,xu2026qwen,magentart2_2026} built on neural audio codecs and audio language models~\citep{zeghidour2021soundstream,defossez2022high,dhariwal2020jukebox,kreuk2022audiogen}; video generation is the instructive counter-case, and it separates two things that are easily conflated. Human-centric video models emit pixels and no audio at all~\citep{chen2026dancetogether,owens2016visually,miao2026framessequencestemporallyconsistent}, and the models that do emit sound alongside video produce a joint waveform~\citep{hacohen2026ltx,team2026mova}, as does the video-to-audio line before them~\citep{luo2023diff,gan2020foley,zhang2026foleycrafter}. That is \emph{sound} generation, which is not music generation, and in either case a waveform offers no handle on key, chord, instrument or note: what a model can be asked to control is fixed by what it is asked to emit; a parallel line builds music- and audio-\emph{understanding} models~\citep{gardner2023llark,liu2024music,chu2023qwen,li2024mert,ye2024mmad}, one of which (Qwen3-Omni) supplies our captions. Appendix~\ref{app:related} treats the closest of these in full.

\subsection{Evaluation of Symbolic Music Generation}
Evaluation here is fragmented: MidiCaps supplies the de-facto text--MIDI dataset~\citep{melechovsky2024midicaps} without a protocol; papers mix perplexity, feature histograms~\citep{yang2020evaluation} and rendered-audio scores~\citep{kilgour2018fr}; FMD contributes a symbolic Fr\'echet distance over CLaMP-2~\citep{retkowski2024frechet} and CLaMP-3 a symbolic-native alignment score~\citep{wu2025clamp}; surveys document the absence of standardized benchmarks~\citep{xiong2023comprehensive,ma2024foundation}; benchmarks for other structured outputs face the same problem~\citep{iwbench,fradet2023impact,zhujiu,qian2026beat,llmspark}. Two pitfalls matter for method papers: per-token loss is incomparable across vocabularies, and raw texture statistics ignore that encodings differ in what they can express at all. Our ceiling-anchored protocol addresses both.

\section{Method}

\subsection{Overview}
\textbf{Problem definition.} Given a natural-language description $c$, generate a symbolic score $y$, notes with pitch, onset, duration, velocity and instrument, that is well-formed (decodable and renderable), matches the description, and matches the statistical profile of real music. A representation $R$ serializes $y$ into a token sequence with decoder $R^{-1}$; a text-to-music method is a choice of $R$ plus a procedure for training $p_\theta(R(y)\,|\,c)$. PMT's pipeline (Figure~\ref{fig:pipeline}) has three interlocking components: the \emph{representation} encodes each piece into a flat music-token stream, the \emph{integration recipe} appends that stream to the caption and fine-tunes a pretrained backbone with the caption masked, and the \emph{evaluation protocol} validates every design choice by training competing representations under the identical recipe. The recipe holds everything but the representation fixed, which is what lets the protocol attribute quality differences to the representation.

\begin{figure}[t]
\centering
\includegraphics[width=\columnwidth]{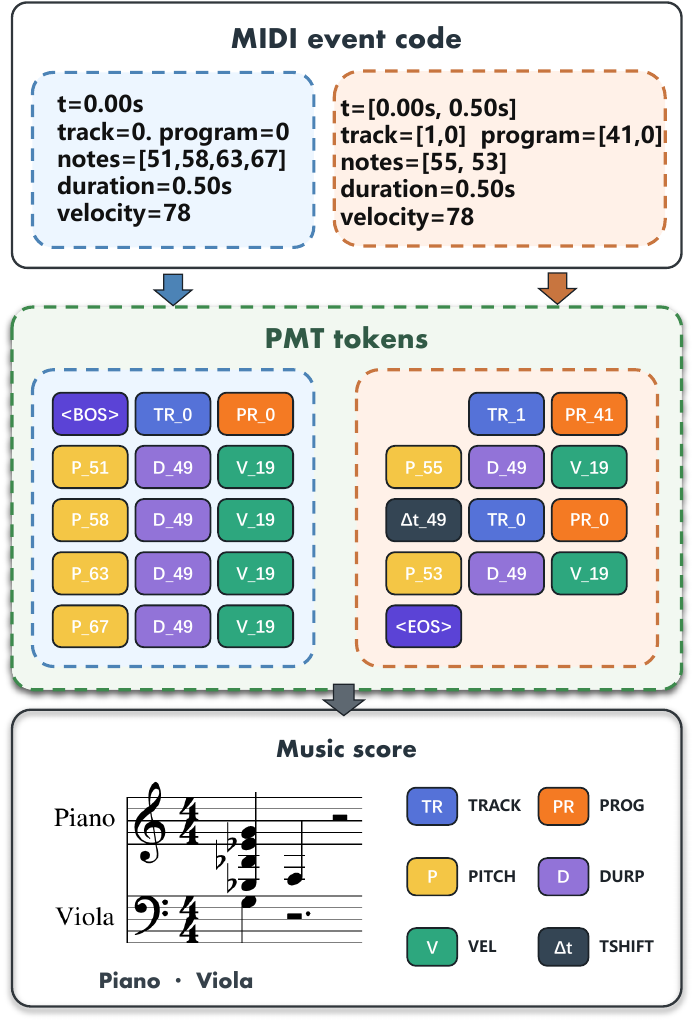}
\caption{\textbf{One excerpt through the encoding.} \emph{Top:} the MIDI events as the file stores them. \emph{Bottom:} the PMT token stream for the same notes.}
\label{fig:tokenizer}
\end{figure}

\subsection{Component 1: Performance-Timed Parametric Representation}
PMT serializes each note as an interpretable parameter tuple; Figure~\ref{fig:tokenizer} walks one excerpt through it, a four-note chord on one track followed by a single onset advance, from the file's own \texttt{note\_on} events to the token stream and back. The vocabulary comprises \texttt{PITCH} (128), \texttt{DUR} (200 bins at 10\,ms, covering 10--2000\,ms), \texttt{VEL} (32 levels), \texttt{TSHIFT} (100 bins at 10\,ms; gaps beyond 1\,s emit repeated shifts), \texttt{TRACK} (16), and \texttt{PROG} (129) plus four special tokens, 609 symbols total. Small named vocabularies of this kind are how other structured artifacts are made LLM-writable~\citep{wu2023iconshop,xing2024svgfusion}, and naming the parts is what makes an output editable after the fact~\citep{poole2022dreamfusion,xu2020rignet}, $\approx$4.0 tokens per note. Writing $[a,b)$ for a family's integer index range and $\mathcal{S}$ for the four special tokens, the music vocabulary is the disjoint union of the six parameter families,
\begin{align}
\mathcal{V}_{\text{music}}=&\ \texttt{TRACK}_{[0,16)}\sqcup\texttt{PROG}_{[0,129)}\sqcup\texttt{PITCH}_{[0,128)}\notag\\
&\sqcup\ \texttt{DUR}_{[1,201)}\sqcup\texttt{VEL}_{[0,32)}\sqcup\texttt{TSHIFT}_{[0,100)}\sqcup\mathcal{S},
\label{eq:vocab}
\end{align}
so $|\mathcal{V}_{\text{music}}|=16{+}129{+}128{+}200{+}32{+}100{+}4=609$. Formally, a piece $y=\{(p_j,o_j,d_j,v_j,c_j)\}_{j=1}^{N}$ (pitch, onset, duration, velocity, channel) is serialized in onset order as
\begin{align}
R_{\text{PMT}}(y) = \big[&\texttt{TSHIFT}_{q_t(\Delta_j)},(\texttt{TRACK}_{c_j}\,\texttt{PROG}_{g_j}),\notag\\[-2pt]
&\texttt{PITCH}_{p_j},\texttt{DUR}_{q_d(d_j)},\texttt{VEL}_{q_v(v_j)}\big]_{j=1}^{N},
\label{eq:pmt}
\end{align}
which is \emph{schematic}: the leading \texttt{TSHIFT} carries the gap to the previous onset and is omitted when $\Delta_j{=}0$, and $(\texttt{TRACK},\texttt{PROG})$ are emitted only when the instrument changes (the exact run-length gating is in Appendix~\ref{app:tokenizer}). With base time unit $\tau{=}10$\,ms and backward inter-onset gap $\Delta_j\!=\!o_j\!-\!o_{j-1}$ (taking $o_0{=}0$), the three quantizers are deterministic and fixed,
\begin{align}
q_t(\Delta)&=\min\!\big(\lfloor \Delta/\tau\rceil,\,99\big),\notag\\[-1pt]
q_d(d)&=\mathrm{clip}\!\big(\lfloor d/\tau\rceil,\,1,\,200\big),\notag\\[-1pt]
q_v(v)&=\lfloor 32v/128\rfloor,
\label{eq:quant}
\end{align}
so time-shift and duration share one 10\,ms lattice and velocity is binned to 32 levels; simultaneous onsets give $\Delta_j{=}0\Rightarrow\texttt{TSHIFT}_0$, so chords are runs of pitches at zero shift, longer gaps are tiled by repeated maximal shift tokens, and \texttt{TRACK}/\texttt{PROG} are emitted only on instrument change, giving $\approx\!4$ tokens/note (full procedure and a worked C-major-triad example in Appendix~\ref{app:tokenizer}). Two design choices define the representation. \emph{Performance resolution}: 10\,ms time-shifts and 32-level velocities preserve the micro-timing and dynamics that beat grids quantize away, with clean model-free evidence in the round-trip \emph{onset error} ($2.3$\,ms median on real music versus $53$\,ms for beat grids). The mechanism is precise: beat-grid tokenizers (REMI, TSD) place each onset on a tempo-relative grid of $R$ steps per beat, encoding it as a bar index plus an in-bar position,
\begin{equation}
\mathrm{pos}(o)=\big\lfloor (o-\mathrm{bar}(o))/\rho\big\rceil,\qquad \rho=\tfrac{60}{B\cdot R}\ \text{s},
\label{eq:beatgrid}
\end{equation}
whose step $\rho$ is \emph{tempo-dependent} ($B$ is beats-per-minute) and bounded by the grid density $R$ (typically $R\!\in\!\{4,8\}$, i.e.\ $\rho\!\approx\!63$ to $125$\,ms at $120$\,BPM, coarser when slower). PMT instead encodes the raw inter-onset interval $\Delta_j$ on the fixed $\tau{=}10$\,ms lattice of Eq.~\eqref{eq:quant}, independent of tempo and metre, so swing and rubato survive rather than snapping to a grid line. That single substitution, continuous performance time for quantized score time, is the one variable the controlled swap of Component~3 isolates. We also report a chord-time ceiling ($39.1\%$ vs.\ REMI/Beat-TSD $35.5\%$, ABC $32.5\%$) but read it with care: because the 10\,ms lattice snaps near-simultaneous onsets to \texttt{TSHIFT}$_0$, PMT's ceiling sits within a couple of points of the raw reference and its sign is subset-dependent, which is why we treat onset error, not the chord-time ceiling, as the primary expressiveness signal (Appendix~\ref{app:tokenizer}). \emph{Flat single stream}: unlike grouped encodings that need multi-head outputs, PMT's stream is consumed by any single-softmax LLM, and it interleaves all instruments so the model attends across the full arrangement. Decoding is deterministic and total, so validity is $0.99$--$1.00$, and because every token names one parameter the representation is disentangled by construction.

The decoder is exact on the quantization lattice: encode$\rightarrow$decode is the identity on quantized pitch, onset, duration, velocity and channel, so every deviation from a source piece is attributable to quantization, never to the serialization. Proposition~\ref{prop:exact} in Appendix~\ref{app:tokenizer} states and proves this, with the error bounds ($\pm$5\,ms onset, $\pm$5\,ms duration, $\pm$2 velocity levels).
\subsection{Component 2: LLM Integration Recipe}
The recipe converts a pretrained text LLM into a text-to-music generator with minimal surgery. (i) \emph{Vocabulary extension}: the 609 music symbols are appended to the backbone tokenizer (Qwen3.5~\citep{yang2025qwen3}, 248k text tokens) as mean-initialized embeddings, no architecture change, one softmax. (ii) \emph{Caption-masked SFT}: each example concatenates caption tokens $c{=}(c_1,\dots,c_m)$ and the music stream $x{=}R_{\text{PMT}}(y){=}(x_1,\dots,x_T)$ into one $1024$-token block $s$. With a per-position mask $\ell_i{=}\mathbb{1}[i{>}m]$, the backbone minimizes the \emph{caption-masked} autoregressive negative log-likelihood over caption--score pairs $\mathcal{D}$,
\begin{equation}
\mathcal{L}(\theta)=-\,\mathbb{E}_{(c,y)\sim\mathcal{D}}\;\frac{1}{|\{i:\ell_i{=}1\}|}\sum_{i=1}^{m+T}\ell_i\,\log p_\theta\big(s_i \mid s_{<i}\big),
\label{eq:loss}
\end{equation}
so gradients flow only from the $T$ music positions. Causal attention over the shared block means the music tokens still attend the full caption, so conditioning is preserved though the caption is not a prediction target. (iii) \emph{Music-constrained decoding}: at inference, logits are renormalized over the music range only,
\begin{equation}
\tilde p_\theta(v \mid c,x_{<t}) \;\propto\; p_\theta(v \mid c,x_{<t})\;\mathbb{1}\!\left[v\in\mathcal{V}_{\text{music}}\right],
\label{eq:decode}
\end{equation}
sampling until \texttt{EOS}, after which the deterministic inverse returns a guaranteed-playable score at any temperature. (iv) \emph{Scale portability}: the identical recipe runs from 0.8B to 27B, and one encoded dataset serves every scale. The recipe is deliberately minimal (single stage, no auxiliary losses, no RL), so measured quality is attributable to the representation.

\subsection{Component 3: Ceiling-Anchored Evaluation Protocol}
The protocol answers ``is each design choice right?'' with controls. \emph{Representation swap}: seven tokenizations grouped by family, \emph{performance-resolution} (PMT and PerTok~\citep{lenz2024pertok}), \emph{beat-grid} (Beat-TSD, identical to PMT except timing; REMI; MIDI-Like, all via MidiTok with an identical base configuration \emph{including 32-level velocity}, so the swap isolates timing quantization), \emph{grouped} (a Structured attribute-tuple control), and \emph{text} (ABC-as-text, whose web-pretraining prior \emph{favors} it), all trained under the same frozen splits, steps, batch, learning rate and decoding. \emph{Ceiling anchoring}: every texture metric $m$ is reported raw, as distance to the real reference, and as \emph{attainment}, where for held-out real music $Y_{\text{test}}$ and encoder/decoder $R,R^{-1}$,
\begin{equation}
A_R(m)=\frac{m(\hat{Y})}{\mathrm{ceil}_R(m)},\qquad \mathrm{ceil}_R(m)=m\!\big(R^{-1}(R(Y_{\text{test}}))\big),
\label{eq:attain}
\end{equation}
so the ceiling is what survives a \emph{model-free} round-trip, the most of $m$ that $R$ can physically carry (a beat grid cannot round-trip sub-beat chords). Attainment separates expressiveness, a property of the \emph{format}, from learnability, the fraction a \emph{model} realizes, the only fair basis for comparing formats of different capacity. \emph{Guarded metrics}: FMD over CLaMP-2 for distributional fidelity, CLaMP-3 for caption alignment with all systems routed through one MIDI$\to$MTF path, and per-token loss used strictly within a representation, since across vocabularies it is meaningless.

\section{Experiments}

\textbf{Dataset.} We build and release a four-modality corpus of 86{,}598 real-music pieces, each aligned across an English caption, sliced MIDI, ABC and rendered audio, from four sources (The Session folk, GiantMIDI classical piano~\citep{kong2022giantmidi}, a genre-balanced Lakh~\citep{raffel2016learning} subset, and jazz). Splits are frozen and source-stratified (84{,}576 train / 2{,}022 test, seed 42) and PMT-encoded lengths are heavy-tailed (median 835 tokens), motivating the 1024-token block that covers $58\%$ of pieces unclipped. Table~\ref{tab:datasets} places the corpus among existing collections: the large raw-MIDI datasets supply scale but no language supervision, multi-track audio--MIDI sets pair synthesized audio but no captions, and the one captioned dataset (MidiCaps) is single-modality, whereas ours is the only corpus aligning \emph{four} modalities per piece. Because the largest raw collections ship un-captioned, we apply the same pipeline to their deduplicated union and build a \emph{scaled} captioned corpus of \textbf{6.25M} MIDI--caption pairs, the largest captioned symbolic dataset, $37\times$ MidiCaps, deduplicated against the frozen test split. The controlled experiments use the aligned 86.6k set for its per-case alignment; \S\ref{app:curation} gives per-source counts, the caption pipeline and its quality-control filter, and redistribution terms.

\textbf{Implementation details.} All backbones are from the \textbf{Qwen3.5} family (0.8B/2B/4B/27B and a 35B-A3B MoE), fine-tuned in a single SFT stage. Music vocabularies extend the tokenizer with mean-initialized embeddings; captions are truncated to 640 tokens and masked with $-100$. Training: 10k steps ($\approx$1.9 epochs), effective batch 16, lr $8{\times}10^{-5}$ cosine, 3\% warmup, bf16, block 1024. Evaluation: temperature 0.95, top-k 60, up to 900 music tokens, music-range-constrained logits, $n{=}100$ frozen-test captions per cell, up to 3 seeds (per-scale splits, vocabulary sizes and compute in Appendix~\ref{app:training}). Two cautions on reading FMD. It is comparable only \emph{within} one protocol: the same PMT-0.8B scores $159$ against the 500-piece stratified reference and $114$ against the 475-piece deduplicated one, and the PMT-vs-beat-grid \emph{ratio} ranges $1.6$--$2.8\times$ across protocols, so the ``${\approx}2\times$'' headline is a representative midpoint. And it carries sampling noise: a real-vs-real floor is $43$ at $n{=}100$, while the PMT-vs-beat-grid gap ($159$ vs.\ $272$) clears that floor. We run no per-comparison significance tests and coverage thins at scale, so the \emph{large} gaps are the claim and finer orderings are trends.

\subsection{Main Results: PMT vs.\ Alternative Representations}
\begin{table}[t]
\centering
\small
\setlength{\tabcolsep}{3.5pt}
\begin{tabular}{llcccc}
\toprule
Size & Repr. & FMD$\downarrow$ & JSD$_{ioi}\downarrow$ & Chord\%$\rightarrow$ & Att.$\uparrow$ \\
\midrule
\multicolumn{6}{l}{\emph{Reference: FMD 0, JSD$_{ioi}$ 0, chord-time 36.7\%}}\\
\midrule
0.8B & \textbf{PMT} & \textbf{159$\pm$8} & \textbf{0.032$\pm$0.002} & \textbf{36.2$\pm$5.6} & \textbf{0.92} \\
0.8B & Beat-TSD & 286$\pm$1 & 0.152$\pm$0.008 & 30.8$\pm$0.9 & 0.87 \\
0.8B & REMI & 285$\pm$1 & 0.152$\pm$0.007 & 30.2$\pm$1.7 & 0.85 \\
0.8B & MIDI-Like & 285$\pm$2 & 0.160$\pm$0.006 & 26.6$\pm$2.0 & 0.75 \\
0.8B & PerTok & 188$\pm$7 & \textbf{0.011$\pm$0.001} & 30.2$\pm$2.6 & 0.86 \\
0.8B & Structured & 287 & 0.121 & 30.3 & 0.85 \\
0.8B & ABC & 406 & 0.079 & 16.2 & 0.50 \\
\midrule
27B & \textbf{PMT} & \textbf{156} & \textbf{0.048} & \textbf{35.8$\pm$0.5} & \textbf{0.92} \\
27B & REMI & 272 & 0.144$\pm$0.005 & 28.8$\pm$0.4 & 0.81 \\
27B & MIDI-Like & 285 & 0.148 & 27.9 & 0.78 \\
\bottomrule
\end{tabular}
\caption{Controlled comparison on 100 frozen-test captions per cell (mean$\pm$std over seeds; single value $=$ one seed); identical backbone, data, budget, decoding, validity 0.99--1.00. \textbf{FMD}$\downarrow$ $=$ CLaMP-2 Fr\'echet Music Distance; \textbf{JSD$_{ioi}$}$\downarrow$ $=$ inter-onset-interval divergence; \textbf{Chord\%}$\rightarrow$/\textbf{Att.}$\uparrow$ $=$ chord-time closer-to-reference and its ceiling-normalized attainment, a coarser proxy that is mixture-confounded per source (Appendix~\ref{app:persource}).}
\label{tab:main}
\end{table}

Table~\ref{tab:main} shows only the two extremes of the grid; the 2B, 4B and 9B blocks behave
identically and are in Appendix~\ref{app:extgrid}. It orders the representations PMT $>$ Beat-TSD $\approx$ PerTok $\approx$ REMI $>$ MIDI-Like\footnote{Three similarly-named entities appear in this paper and are distinct: \textbf{MIDI-Like} is a MidiTok \emph{tokenization} (a controlled-swap arm here); \textbf{MIDI-LLM}~\citep{wu2025midi} and \textbf{MIDILM}~\citep{li2026midilm} are two separate published \emph{systems} used as external baselines.} $\gg$ ABC. Performance resolution is not unique to PMT: PerTok, another performance-resolution encoding, joins PMT in the low-FMD/low-JSD$_{ioi}$ corner and even edges it on onset divergence ($.011$ vs.\ $.032$), corroborating that performance resolution ($10$\,ms timing \emph{and} per-note velocity), not the beat grid, is the binding distinction. PMT's edge over PerTok is \emph{completeness}, not timing: PerTok spends $6.8$ tokens/note against PMT's $4.0$ and its stream is single-track, so the arrangement PMT carries in one stream is outside PerTok's reach (Appendix~\ref{app:extgrid}).

\textbf{Chord-time is a coarse, mixture-confounded proxy.} Stress-testing our own headline statistic, a per-source breakdown (Appendix~\ref{app:persource}) shows the aggregate flatters no arm uniformly: PMT \emph{over}-harmonizes the folk majority ($25\%$ vs.\ a $13\%$ reference) and \emph{under}-fills chordal classical ($50\%$ vs.\ $82\%$), so its blended match ($36\%$ vs.\ $36.7\%$) is partly cancellation across a $69\%$-folk mixture. We therefore rest the quality claim on the full-distribution metrics, FMD and the timing/duration divergences, which are not chord-count cancellations. Unlike the chord-time cancellation, that advantage does \emph{not} wash out per source: split into folk and non-folk, PMT leads Beat-TSD on FMD in \emph{both} ($67$ vs.\ $220$ and $337$ vs.\ $492$).

\textbf{Efficiency and ceilings.} PMT pays a moderate token cost ($4.03$ tokens/note against ABC's $2.84$) and converts it: it has the highest round-trip ceiling ($39.1\%$ chord-time vs.\ $35.5\%$ for beat grids) and, at $0.8$B, the highest attainment of it ($.92$), while compactness alone buys nothing, ABC being leanest and worst. The efficiency table, the model-free encode/decode probe and wall-clock generation rates are in Appendix~\ref{app:structural}.
\subsection{Scaling Behavior}

The clean reading of scale is cross-representation (Appendix~\ref{app:extgrid}, Figure~\ref{fig:scaling}): PMT's FMD band ($152$--$159$ across $0.8$B--$27$B) lies entirely below every beat-grid arm's ($272$--$286$), so the efficiency claim is a Pareto statement, not a fragile single point. The curves do not cross at
any scale, so no amount of backbone size moves a beat grid into PMT's distributional regime. Fitting FMD against $\log_{10}$-parameters, every arm's slope is within noise of flat, leaving a near-constant ${\approx}125$-point PMT-vs-beat-grid offset that scale does not erode. We do not extrapolate a crossover from a fit this flat over $1.5$ decades; the defensible statement is that there is \emph{no} closing trend within range, so representation, not parameter count, is the binding variable at every scale we measure.

\textbf{From-scratch validation.} A web-pretrained backbone could favor some representations through prior exposure (it demonstrably favors ABC). Repeating the swap with a 26M decoder-only Transformer trained \emph{from scratch} on MAESTRO~\citep{hawthorne2018enabling}, changing only the tokenization, reproduces the ordering (PMT FMD $125.1\pm11.7$ over four seeds vs.\ REMI's $344.4\pm9.8$) and transfers to POP909~\citep{wang2020pop909}, so two independent backbone families agree and the pretraining-prior confound is closed. The ordering also holds on \emph{multi-instrument} material trained from scratch, the case a solo-piano reproduction leaves open: on a $1{,}208$-piece genre-balanced subset PMT reaches FMD $287.9\pm1.0$ (three seeds) against $395.3\pm1.1$ for Beat-TSD, with non-overlapping intervals (Appendix~\ref{app:ablate}).

\textbf{Ablations and controls.} Eight controls (Appendix~\ref{app:ablate}) delimit the effect rather than only supporting it. Two matter most. Coarsening PMT's timing to a beat-grid resolution on a fixed absolute grid moves its FMD toward the beat-grid arms, and removing PMT's performance fields one at a time hurts under the from-scratch MAESTRO setup (dropping velocity degrades FMD from $106.1$ to $697$), so the performance parameters are load-bearing. The recipe is not the driver: unmasked caption$+$music loss and unconstrained decoding both leave the ordering intact, while the \emph{representation} family does matter, swapping any MIDI-token encoding for ABC halves attainment ($.69$--$.92\!\to\!.46$--$.51$) and cuts polyphony ($4.1$--$5.8\!\to\!3.0$) under identical conditions. Training-seed std is $\leq$2pp for baselines and ${\approx}6$pp for PMT at $0.8$B; repeated evaluation bounds sampling noise at ${\approx}2.5$pp.
\subsection{Comparison with Published Systems}
Run through our harness on identical captions, \emph{both} published systems imprint their training distributions rather than following the caption: MIDI-LLM emits $72\%$ chord-time and text2midi $66\%$ with $17.3$ max-polyphony, near-invariant to what the caption asks for (Appendix~\ref{app:baselines}, which also shows the full generation pipeline side by side).

\textbf{Public-benchmark transfer (MidiCaps).}
We evaluate on 200 captions from MidiCaps, text2midi's training domain, disjoint from our training subset, with identical extractors and the text2midi adherence metrics (Tempo-Bin, TBT $\pm$10\%; Correct-Key).

\begin{table}[t]
\centering
\small
\setlength{\tabcolsep}{3pt}
\begin{tabular}{lcccccc}
\toprule
System & FMD$\downarrow$ & Valid. & TBT & CK & IF1 & Poly \\
\midrule
\textbf{PMT-4B} 86.6k & \textbf{196} & \textbf{1.00} & .44 & .10 & .46 & 10.6 \\
\textbf{PMT-4B} 6.25M & \textbf{211} & \textbf{1.00} & .44 & .17 & .47 & 10.2 \\
MIDI-LLM & 314 & \textbf{1.00} & .44 & .32 & \textbf{.63} & 11.7 \\
MIDILM & 351 & \textbf{1.00} & .70 & .39 & .54 & 16.8 \\
Amadeus & 380 & \textbf{1.00} & \textbf{.90} & \textbf{.57} & .62 & 72.1 \\
text2midi & 421 & 0.99 & .44 & --$^{\dagger}$ & .34 & 18.5 \\
ChatMusician & 460 & \textbf{1.00} & .45 & .36 & .29 & \phantom{0}3.8 \\
\midrule
\emph{Real reference} & \emph{0} & -- & -- & -- & -- & \emph{6.1} \\
\bottomrule
\end{tabular}
\caption{MidiCaps public test, 200 captions disjoint from our training data, identical extractors and one FMD protocol for every row. \textbf{FMD}$\downarrow$ $=$ distance to the real note distribution; \textbf{TBT} $=$ tempo within $\pm$10\% of the caption$\uparrow$; \textbf{CK} $=$ correct-key rate$\uparrow$; \textbf{IF1} $=$ instrument-F1 against the caption's named instruments$\uparrow$; \textbf{Poly} $=$ max-polyphony. Both PMT rows are zero-shot. Full table, the PMT-27B and zero-shot-LLM rows and the CLAP/audio comparison: Appendix~\ref{app:baselines}, Table~\ref{tab:midicapsfull}. $^{\dagger}$text2midi's decoded MIDI is unparseable by the music21 key-finder.}
\label{tab:midicaps}
\end{table}

Two readings. First, \emph{every} model, ours included, regresses toward its training mean (PMT over-harmonizes folk, Appendix~\ref{app:persource}); the pathology we diagnose is caption-\emph{invariance}, and it is domain-invariant, MIDI-LLM emits a near-identical distribution ($71\%$ vs.\ $72\%$ chord-time) on two disjoint caption domains and does not improve with scale, whereas PMT shifts with the testbed and transfers zero-shot to parity on theirs (Figure~\ref{fig:imprint}). A ``more-texture-is-better'' reading would reward MIDI-LLM's constant $71$--$72\%$; scoring distance to each domain's own real reference is what exposes it as caption-invariance instead, and MIDI-LLM sits near the MidiCaps reference only because MidiCaps \emph{is} its home distribution, which is why the test is run in both directions. Second, on the shared tempo axis in-domain text2midi ($.44$) buys \emph{no} advantage over zero-shot PMT ($.44$); what separates systems is the explicit-attribute axes, where systems that condition on them directly lead. Single-domain benchmarks favor systems trained on that domain; bidirectional evaluation is the corrective.

\paragraph{Where PMT leads and where it does not.} Table~\ref{tab:midicaps} separates two axes that are usually reported as one. On distance to the real note distribution, the axis our thesis targets, PMT is decisively closest: same-protocol FMD $196$ for the 86.6k-trained PMT-4B and $211$ for the 6.25M one-epoch model, against $314$--$460$ for every published system, $1.5$--$2.3\times$ farther. On caption-attribute adherence PMT trails: MIDILM~\citep{li2026midilm} and Amadeus~\citep{su2026amadeus} lead tempo and key (Amadeus TBT $.90$/CK $.57$ vs.\ PMT $.44$/$.17$), and zero-shot general LLMs lead key adherence by copying the stated key (GPT-5.5 CK $.66$) while being playable for only $58\%$ of captions. No competitor is both adherent \emph{and} close-to-real: the adherence leaders reach it by over-texturing, at $16.8$ and $72.1$ max-polyphony against a $6.1$ reference. Cross-modal comparison against audio generators does not cleanly rank quality on this task, caption-matched real music scores \emph{lowest} on CLAP ($31.5$), so we report but do not rank by it (Appendix~\ref{app:baselines}).

\textbf{Robustness to piece complexity.} Public text-to-MIDI benchmarks skew short, so we also build a \emph{tiered} benchmark of 240 real pieces stratified by PMT-token length into four levels, each with its own real reference. PMT-4B and MIDILM both stay $100\%$ valid at every tier, but PMT is closer to the per-tier real distribution at \emph{every} level ($298/348/357/396$ vs.\ $422/378/377/436$) while MIDILM over-textures throughout, so the distributional advantage persists across complexity and degrades gracefully with length (Appendix~\ref{app:structural}).

\subsection{Qualitative Results and Perceptual Scope}
PMT generations are caption-conditioned in texture: folk captions yield single-line dance tunes, while classical, jazz and pop captions yield sustained chordal writing with up to $13$--$17$ simultaneous notes. Appendix~\ref{app:staff} shows the same caption realized by each representation, in piano roll and engraved: the beat-grid arms collapse to sparser tracks and ABC to a two-staff keyboard reduction under identical backbone and budget, so the representation, not the backbone, bounds the notated texture. An automatic listener (Qwen3-Omni-30B) is a \emph{weak} proxy whose signal is material-dependent: it prefers micro-timing on expressive solo piano ($60\%$, $54/90$, $p{=}0.07$) but not on the folk-heavy generated contrast ($42\%$), so we make \emph{no} perceptual claim from it. A pre-registered human study on expressive material is collecting; we report no perceptual result from it until it reaches its target sample, so the timing claim here stays distributional (Appendix~\ref{app:humanstudy}).
We probe audibility with an automatic listener (Qwen3-Omni-30B) and read it as a \emph{weak} proxy that scopes where a human study should look rather than delivering a verdict. The signal is \emph{material-dependent}: on a model-free control isolating micro-timing on expressive solo piano the judge prefers the micro-timing-preserving version $60\%$ of the time ($54/90$, directional, $p{=}0.07$), while on the folk-heavy generated contrast it does \emph{not} prefer PMT ($42\%$), consistent with a source that is already grid-quantized. We therefore make \emph{no} perceptual claim from this probe. A human listening study on the material the probe points at is underway; the instrument, the pre-registered analysis and the response template are in Appendix~\ref{app:humanstudy}, and full listener details in Appendix~\ref{app:perceptual}.
\subsection{Limitations}
Five limitations bound our claims; supporting numbers are in Appendix~\ref{app:limitations}. \emph{(1) Near-duplicate splits}: the folk source has multiple community ``settings'' of one tune and our split is setting-level, so near-duplicates can straddle train/test and inflate held-out likelihood; we exclude $>$4B points from NLL analyses, and generation-side metrics survive a title-level dedup control. \emph{(2) Audibility untested}: the human study is still collecting, and our automatic proxy on the representation contrast is inconclusive, so we assert no perceptual difference in either direction. \emph{(3) Seed coverage at scale}: 27B rows are two-seed, 35B-A3B single-seed. \emph{(4) Caption adherence}: PMT's key adherence is weak (CK $.10$ vs.\ MIDI-LLM's $.32$), but this is an absolute-key \emph{conditioning} gap rather than tonal incompetence, and a decode-time attribute constraint more than doubles \emph{both} instrument-F1 ($.28\!\to\!.60$) and Correct-Key ($.16\!\to\!.35$) at unchanged validity and no distributional cost, without retraining; steering a generator with language rather than retraining it is a recurring pattern~\citep{lu2023musecoco,wu2022exploring}. Our captions are also machine-generated by a single captioner and audited for prompt leakage but not factual accuracy; describing audiovisual content in language is an active task in its own right~\citep{gardner2023llark,li2024mert,chu2023qwen}, and the accuracy of that description bounds any corpus built this way. \emph{(5) Corpus scale in the controlled study}: the experiments use the 86.6k four-modality set, trading raw count for per-case alignment; as a first step to scale, a PMT-4B trained one epoch on the 6.25M corpus cuts held-out FMD $2.0\times$ ($291.7\!\to\!144.8$) and lifts MidiCaps key adherence (CK $.10\!\to\!.17$) at a $15$-point FMD cost there and nearly doubles key adherence. We report that one-epoch model, not a converged one.

\section{Conclusion}
We asked whether, for an LLM extended to symbolic music, the token representation or the model's size sets the ceiling on distributional fidelity, and answered with a controlled swap: fix the backbone, data, budget and decoding, and change only the representation. The representation wins. A performance-resolution stream we release (PMT) roughly halves the Fr\'echet Music Distance of every beat grid, is not overtaken by a $34\times$ parameter increase, and reproduces on a 26M from-scratch backbone and on a second performance-resolution tokenizer, so the effect belongs to the representation \emph{class}, not one vocabulary. We are careful about scope: the gap is distributional, robust to onset coarse-quantization but shown neither audible by our automatic probe nor better-conditioned, and PMT trails dedicated systems on caption adherence, a decode-time conditioning problem we already partly close at no distributional cost. The lesson outlasts music: how a new modality is serialized can dominate parameter count, and a controlled swap establishes this before any model is scaled, wherever a generative model is asked for a structured artifact rather than a dense signal~\citep{xu2020rignet,wu2023iconshop}. We release the datasets, harness and 25+ checkpoints so the next representation claim can be \emph{measured}, not asserted.

\bibliography{mrgb_refs}

\clearpage
\appendix
\setcounter{secnumdepth}{2}

\clearpage
\section*{Contents of the Appendix}
\vspace{1pt}
{\setlength{\parindent}{0pt}\setlength{\parskip}{0pt}
\newcommand{\tocrow}[4]{\noindent\hspace*{#1}%
  \parbox[t]{\dimexpr\columnwidth-#1-1.7em\relax}{\raggedright #2 #3\par}%
  \hfill\makebox[1.6em][r]{#2 \pageref{#4}}\par}
\newcommand{\tocA}[2]{\vspace{3.2pt}\tocrow{0pt}{\small\bfseries}{\ref{#1}\hspace{0.4em}#2}{#1}\nobreak}
\newcommand{\tocB}[2]{\tocrow{0.7em}{\footnotesize}{#2}{#1}}
\newcommand{\tocC}[2]{\tocrow{1.7em}{\scriptsize}{#2}{#1}}

\tocA{app:related}{Extended Related Work}
\tocC{apx:contributions-in-full}{Contributions in full}
\tocB{apx:symbolic-music-representations}{Symbolic Music Representations}
\tocB{apx:text-conditioned-symbolic-music-generation}{Text-Conditioned Symbolic Music Generation}
\tocB{apx:evaluation-of-symbolic-music-generation}{Evaluation of Symbolic Music Generation}
\tocA{app:tokenizer}{Tokenizer: Encoding Procedure and Design Rationale}
\tocC{apx:why-the-chord-time-ceiling-is}{Why the chord-time ceiling is the secondary signal}
\tocC{apx:the-integration-recipe-in-full}{The integration recipe in full}
\tocC{apx:exactness-of-the-decoder}{Exactness of the decoder}
\tocC{apx:encoding-details-behind-component-1}{Encoding details behind Component 1}
\tocC{apx:the-ceiling-anchored-protocol-in-full}{The ceiling-anchored protocol in full}
\tocB{apx:component-3-ceiling-anchored-evaluation-protocol}{Component 3: Ceiling-Anchored Evaluation Protocol}
\tocC{apx:decoder-and-formal-guarantees}{Decoder and formal guarantees}
\tocC{apx:encoding-procedure}{Encoding procedure}
\tocC{apx:worked-example}{Worked example}
\tocC{apx:design-rationale}{Design rationale}
\tocB{app:worked}{A Worked Example: One Caption, Its Notes, Its Tokens}
\tocB{app:paired}{Paired Data Samples: MIDI, Tokens, Score}
\tocC{apx:ex-A051}{Example A05-1}
\tocC{apx:ex-C101}{Example C10-1}
\tocC{apx:ex-A131}{Example A13-1}
\tocC{apx:ex-B018}{Example B018}
\tocA{app:staff}{Engraved-Score Comparison}
\tocB{apx:qualitative-results}{Qualitative Results}
\tocA{app:extmetrics}{Extended Distributional Metrics}
\tocC{apx:model-free-fmd-decomposition}{Model-free FMD decomposition}
\tocA{app:declen}{Decode Budget and the Length Bias of FMD}
\tocC{apx:the-budget-was-binding}{The budget was binding}
\tocC{apx:longer-is-worse-under-fmd-and}{Longer is worse under FMD, and only because it is longer}
\tocC{apx:consequences}{Consequences}
\tocA{app:ablate}{Design-Parameter Ablations}
\tocC{apx:interpretability-disentangled-editable-tokens}{Interpretability: disentangled, editable tokens}
\tocC{apx:from-scratch-validation-in-full}{From-scratch validation in full}
\tocB{apx:ablation-studies}{Ablation Studies}
\tocC{apx:from-scratch-backbone-full-numbers}{From-scratch backbone (full numbers)}
\tocA{app:persource}{Per-Source Texture Breakdown}
\tocC{apx:chord-time-as-a-proxy-the}{Chord-time as a proxy: the full caveat}
\tocC{apx:per-source-fmd-the-advantage-is}{Per-source FMD: the advantage is not a folk artifact}
\tocA{app:perceptual}{Perceptual Evaluation Details}
\tocC{apx:qualitative-results-and-the-automatic-listener}{Qualitative results and the automatic listener in full}
\tocB{apx:qualitative-results-2}{Qualitative Results}
\tocB{sec:perceptual}{Automatic MLLM Listener}
\tocB{app:humanstudy}{Human listening study: instrument and pre-registered protocol}
\tocC{apx:design}{Design}
\tocC{apx:pre-registered-analysis}{Pre-registered analysis}
\tocA{app:baselines}{Extended Baseline Comparisons}
\tocC{apx:midicaps-public-test-all-rows-and}{MidiCaps public test: all rows and columns}
\tocC{apx:newest-systems-and-cross-modal-baselines}{Newest systems and cross-modal baselines in full}
\tocC{apx:newest-dedicated-systems-2025-26}{Newest dedicated systems (2025--26)}
\tocC{apx:cross-modal-and-zero-shot-baselines}{Cross-modal and zero-shot baselines (this appendix)}
\tocC{apx:reading-table-row-by-row}{Reading Table~\ref{tab:midicaps} row by row}
\tocC{apx:cross-modal-audio-generators}{Cross-modal (audio generators)}
\tocC{apx:zero-shot-general-llms}{Zero-shot general LLMs}
\tocC{apx:audio-domain-view-of-over-texturing}{Audio-domain view of over-texturing}
\tocA{app:structural}{Additional Structural Metrics}
\tocB{apx:efficiency-and-ceilings}{Efficiency and Ceilings}
\tocB{apx:comprehensive-generative-metrics-structure-diversity-novelty}{Comprehensive generative metrics: structure, diversity, novelty, harmony}
\tocA{app:extgrid}{Extended Representation Grid}
\tocC{apx:the-controlled-grid-at-every-scale}{The controlled grid at every scale}
\tocC{apx:reading-the-controlled-grid-in-full}{Reading the controlled grid in full}
\tocC{apx:scaling-behaviour-in-full}{Scaling behaviour in full}
\tocA{app:training}{Training Details}
\tocC{apx:implementation-details-in-full}{Implementation details in full}
\tocC{apx:per-scale-sharding-and-cost}{Per-scale sharding and cost}
\tocC{apx:trainability-check}{Trainability check}
\tocC{apx:statistical-rigor}{Statistical rigor}
\tocA{app:limitations}{Limitations Details}
\tocC{apx:limitations-in-full}{Limitations in full}
\tocB{apx:limitations}{Limitations}
\tocC{apx:near-duplicate-split-control}{Near-duplicate split control}
\tocC{apx:key-adherence-breakdown}{Key-adherence breakdown}
\tocC{apx:embedder-robustness-the-same-gap-under}{Embedder robustness: the same gap under CLaMP-3}
\tocC{apx:reference-robustness-the-same-gap-against}{Reference robustness: the same gap against a source-stratified reference}
\tocC{apx:6-25m-corpus-scaling-trajectory}{6.25M-corpus scaling trajectory}
\tocA{app:curation}{Scalable Audio-to-Symbolic Data Curation}
\tocC{apx:corpus-construction-in-full}{Corpus construction in full}
\tocC{apx:aligned-86-6k-corpus}{Aligned 86.6k corpus}
\tocC{apx:acquisition-all-sources}{Acquisition, all sources}
\tocC{apx:streaming-transcription}{Streaming transcription}
\tocC{apx:quality-curation}{Quality curation}
\tocC{apx:deduplication-and-captioning}{Deduplication and captioning}
\tocC{apx:captioning-at-scale}{Captioning at scale}
\tocC{apx:caption-quality-control}{Caption quality control}
\tocC{apx:positioning-among-caption-datasets}{Positioning among caption datasets}
\tocC{apx:per-source-acquisition-curation-funnel}{Per-source acquisition--curation funnel}
\tocC{apx:corpus-statistics-duration-density-token-length}{Corpus statistics: duration, density, token length}
\tocA{app:gallery}{Generation Gallery: Caption, Piano Roll and Engraved Score}

\par}
\clearpage

\section{Extended Related Work}
\label{app:related}

\paragraph{Contributions in full.}\label{apx:contributions-in-full} The main text states the three contributions compactly; this is the full form, with the downstream direction we see for the measurement lens.

Our contributions are threefold. \textbf{(1) A new music tokenizer, PMT}, which we release: it streams multi-track music as a compact $609$-symbol vocabulary at ${\approx}4$ tokens/note while preserving $10$\,ms micro-timing and per-note velocity. PMT sits in the performance-event lineage~\citep{oore2020time}; its delta is per-note velocity and multi-track program structure, which the closest LLM-vocabulary system MIDI-LLM omits, at higher token efficiency than the comparable PerTok ($4$ vs.\ $6.8$ tokens/note). We do not claim the flat performance-event encoding as novel, but we are the first to subject it to a controlled swap against every alternative. \textbf{(2) Two new datasets}: a four-modality aligned corpus of $86{,}598$ real-music pieces, each with an expert-length English caption, MIDI, ABC, and rendered audio under per-source licenses (to our knowledge the first captioned symbolic corpus aligned per case across all four modalities), and a $6.25$M scaled captioned corpus, the largest captioned music dataset to date (of any domain; exceeded only by \emph{general}-audio corpora, \S\ref{app:curation}). \textbf{(3) A new controlled benchmark} (with released harness): the first \emph{ceiling-anchored, cross-family} representation comparison (MIDI-event, grouped, and ABC encodings) for text-conditioned symbolic generation, extending within-family swaps~\citep{fradet2023impact}. Its headline finding is that the \emph{representation Pareto-dominates model scale} for distributional fidelity, a 0.8B performance-resolution model beats a 27B beat grid, reproduced on a 26M from-scratch backbone (including a multi-instrument control) and robust to seeds; alongside an \emph{imprinting diagnostic}, published text-to-MIDI systems replay their training distribution near-invariant to the caption (MIDI-LLM emits $72\%$ vs.\ $71\%$ chord-time on two disjoint domains), which distance-to-reference metrics expose. The measurement lens is not specific to music: it applies to any LLM extended to a new tokenized modality, echoing the broader tokenization debate in language modeling (subword~\citep{sennrich2016neural} vs.\ byte-level and tokenizer-free models~\citep{xue2022byt5,yu2023megabyte}), and running parallel to the effort to give LLMs an editable, parametric output format for other structured artifacts~\citep{fradet2023miditok,wu2022exploring}.
This appendix carries the full related-work discussion; the main text keeps the citations and the
distinctions the argument turns on.

\subsection{Symbolic Music Representations}\label{apx:symbolic-music-representations}
Symbolic music is serialized in three families. MIDI-event tokenizations unroll performances into event streams: Oore et al.'s performance events with fine time shifts~\citep{oore2020time}, bar/position grids in REMI~\citep{huang2020pop}, and duration-based TSD, standardized by MidiTok~\citep{fradet2023miditok}. Transformer decoders over such streams scale structure from the Music Transformer~\citep{huang2019musictransformer} to long-sequence attention~\citep{yu2022museformer} and cascaded-diffusion~\citep{wang2024whole} successors. Grouped encodings pack each note into an attribute tuple (Compound Word~\citep{hsiao2021compound}, Octuple~\citep{zeng2021musicbert}) at the cost of multi-head architectures. Multi-track generation has its own line, MMM~\citep{ens2020mmm}, the Multitrack Music Transformer~\citep{dong2023multitrack}, and FIGARO~\citep{von2023figaro}; PMT instead flattens all tracks into one performance-timed stream consumable by an off-the-shelf single-softmax LLM. Text notations (TunesFormer, ChatMusician, MuPT, NotaGen~\citep{wu2023tunesformer,yuan2024chatmusician,qu2025mupt,wang2025notagen}) use ABC as plain text. Prior comparisons stay within one family (time encodings of MIDI tokens~\citep{fradet2023impact}; patching within ABC~\citep{wang2024exploring}; uniform temporal-step tokens~\citep{qian2026beat}) or benchmark representations for \emph{classification and understanding}~\citep{zhang2023symbolic,strepetov2025symurbench} rather than the controlled, cross-family, text-conditioned \emph{generation} we study. PMT descends from the performance-event lineage of \citet{oore2020time}, cast as a compact parametric vocabulary covering multi-track music in one stream and integrated with a pretrained LLM; unlike prior work, we do not merely adopt a representation but \emph{prove} the design choices against REMI, beat-grid TSD, MIDI-Like, and ABC under a fixed backbone and budget.

\subsection{Text-Conditioned Symbolic Music Generation}\label{apx:text-conditioned-symbolic-music-generation}
Text-to-MIDI systems pair a caption encoder with a music decoder: text2midi couples FLAN-T5 with a REMI decoder~\citep{bhandari2025text2midi}; MIDI-LLM extends Llama-3.2's vocabulary with MIDI events~\citep{wu2025midi}; MuseCoco routes free text through musical attributes~\citep{lu2023musecoco}. The ABC line reaches text conditioning via instruction tuning (ChatMusician) or metadata prompts (NotaGen), with the earliest pretrained-checkpoint study by \citet{wu2022exploring}. Each system entangles its representation with its data and recipe, so published numbers cannot isolate either. PMT is an LLM-vocabulary-extension method like MIDI-LLM, which likewise encodes fine (10\,ms) arrival-time via AMT-style tokens~\citep{thickstun2023anticipatory}, but \emph{also} preserves per-note velocity (which MIDI-LLM omits) and multi-track program structure and, more importantly, differs in \emph{validation}: we hold backbone, data, and budget fixed and swap the representation, running MIDI-LLM through our harness on identical captions, which reveals that it imprints its training distribution (72\% chord-time against a 36.7\% reference) where PMT matches within 1pp. \emph{Very recent systems (2025--26)} continue both lines: symbolic ones, Amadeus~\citep{su2026amadeus} (autoregressive notes plus discrete-diffusion attributes) and Text2Score~\citep{bhandari2026text2score} (plan-then-execute bar-level ABC); audio ones, MusicGen~\citep{copet2023simple}, AudioLDM\,2~\citep{liu2024audioldm}, Stable Audio Open~\citep{evans2025stable}, built on neural audio codecs~\citep{zeghidour2021soundstream,defossez2022high} and audio language models~\citep{dhariwal2020jukebox,kreuk2022audiogen}, and full-song YuE~\citep{yuan2026yue}, the concurrent Qwen-Music~\citep{xu2026qwen}, and real-time on-device models (Magenta RealTime~2~\citep{magentart2_2026}). A parallel line builds \emph{music-understanding} models rather than generators, LLark~\citep{gardner2023llark}, Music Understanding LLaMA~\citep{liu2024music}, and audio-language backbones~\citep{chu2023qwen,li2024mert}, one of which (Qwen3-Omni) supplies our captions. Qwen-Music also extends a Qwen backbone, but its ``Music-Tokenizer'' is a $25$\,Hz neural \emph{audio} codec producing sung waveforms, whereas PMT's is a \emph{symbolic} performance representation producing an editable score; they share a backbone but target opposite deliverables (a fixed recording vs.\ a re-renderable score) on disjoint metric families, so they are complementary rather than directly comparable. We compare against representative open systems from both sides under one protocol (\S Experiments): symbolic systems share our MIDI metrics directly, while audio generators are compared cross-modally by rendering our MIDI to audio (FAD/CLAP), though FAD structurally favors native-audio timbre, so we read caption-alignment as the fairer axis. Our contribution is orthogonal to any single competitor: a \emph{representation} result, isolated by controlled swaps, that any of these systems could adopt.

\subsection{Evaluation of Symbolic Music Generation}\label{apx:evaluation-of-symbolic-music-generation}
Evaluation here is fragmented: MidiCaps provides the de-facto text--MIDI dataset~\citep{melechovsky2024midicaps} without a protocol; papers mix perplexity, feature histograms~\citep{yang2020evaluation}, and rendered-audio scores such as Fr\'echet Audio Distance~\citep{kilgour2018fr}; FMD contributes a symbolic Fr\'echet distance over CLaMP-2 embeddings~\citep{retkowski2024frechet} and CLaMP-3 a symbolic-native text-alignment score~\citep{wu2025clamp}; surveys document the absence of standardized benchmarks~\citep{xiong2023comprehensive,ma2024foundation}. Two pitfalls matter for method papers: per-token loss is incomparable across vocabularies, and raw texture statistics ignore that encodings differ in what they can express at all. Our ceiling-anchored protocol addresses both, loss is used only within a representation and texture is normalized by each representation's round-trip ceiling, which is what makes our cross-representation validation trustworthy rather than self-serving.

\section{Tokenizer: Encoding Procedure and Design Rationale}
\label{app:tokenizer}

\paragraph{Why the chord-time ceiling is the secondary signal.}\label{apx:why-the-chord-time-ceiling-is} The main text states this in short form; the full argument, with the subset-dependence that motivates it, is:

We also report a chord-time ceiling (39.1\% vs.\ REMI/Beat-TSD 35.5\%, MIDI-Like 35.7\%, ABC 32.5\%) but read it with care: because the 10\,ms lattice snaps near-simultaneous onsets to \texttt{TSHIFT}$_0$, PMT's ceiling sits within a couple of points of the raw reference and its sign is subset-dependent (39.1\% vs.\ the 36.7\% reference on the 500-piece set of Table~\ref{tab:eff}, but 35.5\% vs.\ the 37.1\% reference on the 300-piece subset of Table~\ref{tab:rt-shift}). Because PMT round-trips chord-time near-losslessly, its ceiling $\approx$ the reference, so its attainment denominator is essentially the reference and attainment can drift slightly above $1$ under sampling noise, a further reason we treat onset error, not the chord-time ceiling, as the primary expressiveness signal.

\paragraph{The integration recipe in full.}\label{apx:the-integration-recipe-in-full} Expanded form of Component~2, with the conditioning rationale spelled out.

The recipe converts a pretrained text LLM into a text-to-music generator with minimal surgery. (i) \emph{Vocabulary extension}: the 609 music symbols are appended to the backbone tokenizer (Qwen3.5~\citep{yang2025qwen3}, 248k text tokens) as mean-initialized embeddings, no architecture change, one softmax. (ii) \emph{Caption-masked SFT}: with $\mathcal{V}{=}\mathcal{V}_{\text{text}}\cup\mathcal{V}_{\text{music}}$ ($|\mathcal{V}_{\text{music}}|{=}609$), each example concatenates caption tokens $c{=}(c_1,\dots,c_m)$ and the music stream $x{=}R_{\text{PMT}}(y){=}(x_1,\dots,x_T)$ of Eq.~\eqref{eq:pmt} into one $1024$-token block $s{=}(c_1,\dots,c_m,x_1,\dots,x_T)$. With a per-position label mask $\ell_i{=}\mathbb{1}[i{>}m]$ ($0$ on the caption prefix, $1$ on the music stream), the backbone $p_\theta$ minimizes the \emph{caption-masked} autoregressive negative log-likelihood over caption--score pairs $\mathcal{D}$,
\begin{equation}
\mathcal{L}(\theta)=-\,\mathbb{E}_{(c,y)\sim\mathcal{D}}\;\frac{1}{|\{i:\ell_i{=}1\}|}\sum_{i=1}^{m+T}\ell_i\,\log p_\theta\big(s_i \mid s_{<i}\big),
\end{equation}
so gradients flow only from the $T$ music positions (caption positions carry the ignore label $-100$). The backbone's language understanding is reused as a fixed conditioning prefix while the loss shapes only the music continuation $p_\theta(x\mid c)$; causal self-attention over the shared block means the music tokens still attend the full caption, so conditioning is preserved though the caption is not a prediction target. (iii) \emph{Music-constrained decoding}: at inference a caption is completed autoregressively with logits renormalized over the music range only,
\begin{equation}
\tilde p_\theta(v \mid c,x_{<t}) \;\propto\; p_\theta(v \mid c,x_{<t})\;\mathbb{1}\!\left[v\in\mathcal{V}_{\text{music}}\right],
\end{equation}
sampling $x_t\sim\tilde p_\theta$ until \texttt{EOS}; the deterministic inverse $y{=}R_{\text{PMT}}^{-1}(x)$ of Eq.~\eqref{eq:pmt} then returns a guaranteed-playable score at any temperature (validity $0.99$--$1.00$). (iv) \emph{Scale portability}: the identical recipe runs from 0.8B to 27B (bf16, gradient checkpointing, effective batch 16, DeepSpeed ZeRO-3 beyond 4B), and one encoded dataset serves every scale. The recipe is deliberately minimal (single stage, no auxiliary losses, no RL), so measured quality is attributable to the representation and the method is drop-in reproducible.

\paragraph{Exactness of the decoder.}\label{apx:exactness-of-the-decoder}
\begin{proposition}[Round-trip identity on the lattice]
\label{prop:exact}
Let $y=\{(p_j,o_j,d_j,v_j,c_j)\}_{j=1}^{N}$ be a piece and let $\hat y = R_{\mathrm{PMT}}^{-1}(R_{\mathrm{PMT}}(y))$. Write $\bar o_j,\bar d_j,\bar v_j$ for the lattice-quantized onset, duration and velocity implied by Eq.~\eqref{eq:quant}. Then $\hat y$ has exactly $N$ notes and, for every $j$, $\hat p_j=p_j$, $\hat c_j=c_j$, $\hat o_j=\bar o_j$, $\hat d_j=\bar d_j$ and $\hat v_j=\bar v_j$. In particular $R_{\mathrm{PMT}}^{-1}\circ R_{\mathrm{PMT}}$ is the identity on pieces whose attributes already lie on the lattice, and otherwise deviates only by quantization: with $\tau=10$\,ms, every inter-onset gap and every duration is within $\tau/2$ of its true value, and every velocity is within half a bin, $|\hat v_j-v_j|\le 2$ MIDI velocity units. Onsets are accumulated gaps, so their per-note deviation is bounded by $j\tau/2$ rather than $\tau/2$; measured on real music the median is $2.3$\,ms.
\end{proposition}

\begin{proof}
Encoding emits, per note in onset order, one \texttt{TSHIFT} carrying $q_t(\Delta_j)$, the pair $(\texttt{TRACK}_{c_j},\texttt{PROG}_{g_j})$ when the instrument changes, then $\texttt{PITCH}_{p_j}$, $\texttt{DUR}_{q_d(d_j)}$ and $\texttt{VEL}_{q_v(v_j)}$. The six families occupy disjoint index ranges (Eq.~\eqref{eq:vocab}), so each emitted token is attributable to exactly one attribute and the stream parses unambiguously into $N$ five-tuples: \texttt{PITCH} starts a note, and the two tokens following it are by construction its duration and velocity.

Pitch and channel are copied, hence recovered exactly. Onsets: decoding accumulates $\hat o_j=\sum_{i\le j}\tau\,q_t(\Delta_i)$, and since $q_t$ is applied to the backward gaps of a common origin $o_0=0$, this telescopes to the quantized onset $\bar o_j$; gaps exceeding what one shift spans are tiled by maximal shifts plus a remainder, which changes the summands but not the sum. Duration and velocity are read back as the bin values $\tau\,q_d(d_j)$ and $q_v(v_j)$, i.e.\ $\bar d_j$ and $\bar v_j$. Instrument identity persists between change points by the run-length gating, so $\hat c_j=c_j$ for every note, not only at change points.

The bounds follow from the quantizers alone: $q_t$ and $q_d$ round to the nearest multiple of $\tau$, giving at most $\tau/2$ per gap and per duration, which the onset sum accumulates over $j$ gaps; and $q_v(v)=\lfloor 32v/128\rfloor$ has bin width $4$ MIDI velocity units, i.e.\ at most $2$ units once decoded to the bin center. No step of the serialization introduces error beyond these, which is the property the controlled comparison needs: any deviation a round-trip shows is quantization, never the encoding.
\end{proof}

\paragraph{Encoding details behind Component 1.}\label{apx:encoding-details-behind-component-1} The run-length gating, the lattice ranges, and why the chord-time ceiling is subset-dependent.

so time-shift and duration share one 10\,ms lattice (durations $10$--$2000$\,ms, shifts up to $\sim\!1$\,s) and velocity is binned to 32 levels; simultaneous onsets give $\Delta_j{=}0\Rightarrow\texttt{TSHIFT}_0$, so chords are runs of pitches at zero shift, longer gaps are tiled by repeated maximal shift tokens plus a remainder, and \texttt{TRACK}/\texttt{PROG} are emitted only on instrument change (the run-length gating giving $\approx\!4$ tokens/note; full procedure and a worked C-major-triad example in this appendix). Two design choices define the representation. \emph{Performance resolution}: 10\,ms time-shifts and 32-level velocities preserve the micro-timing and dynamics that beat grids quantize away, with clean model-free evidence in the round-trip \emph{onset error} (2\,ms median on real music versus 53\,ms for beat grids; Fig.~\ref{fig:tokenizer_compare}). The mechanism is precise: beat-grid tokenizers (REMI, TSD) place each onset on a tempo-relative grid of $R$ steps per beat, encoding it as a bar index plus an in-bar position,

\paragraph{The ceiling-anchored protocol in full.}\label{apx:the-ceiling-anchored-protocol-in-full} Expanded form of Component~3, including the full rationale for attainment as the comparison basis.

\subsection{Component 3: Ceiling-Anchored Evaluation Protocol}\label{apx:component-3-ceiling-anchored-evaluation-protocol}
The protocol answers ``is each design choice right?'' with controls. \emph{Representation swap}: we compare seven tokenizations grouped by family, \emph{performance-resolution} (PMT and PerTok~\citep{lenz2024pertok}, both keeping sub-grid onset timing), \emph{beat-grid} (Beat-TSD, identical to PMT except timing; REMI; MIDI-Like, all via MidiTok, multi-track single-stream, identical base configuration \emph{including 32-level velocity}, so the swap isolates timing quantization, not velocity granularity), \emph{grouped} (a Structured attribute-tuple control), and \emph{text} (ABC-as-text, the MuPT setting, whose web-pretraining prior \emph{favors} it). All arms train under the same frozen splits, steps, effective batch, learning rate, and decoding, with multiple seeds. \emph{Ceiling anchoring}: every texture metric $m$ (chord-time fraction, max polyphony) is reported three ways, raw value $m(\hat{Y})$ on generations $\hat{Y}$, distance to the real-data reference, and \emph{attainment}, where for held-out real music $Y_{\text{test}}$ and encoder/decoder $R,R^{-1}$,
\begin{equation}
A_R(m)=\frac{m(\hat{Y})}{\mathrm{ceil}_R(m)},\qquad \mathrm{ceil}_R(m)=m\!\big(R^{-1}(R(Y_{\text{test}}))\big),
\end{equation}
where the ceiling $\mathrm{ceil}_R(m)$ is the metric surviving a \emph{model-free} encode--decode round-trip, the most of attribute $m$ that $R$ can physically carry (e.g.\ a beat grid cannot round-trip sub-beat chords). Attainment separates expressiveness (the ceiling, a property of the \emph{format}) from learnability (the fraction a \emph{model} realizes), the only fair basis for comparing formats of different capacity, and it lets a smaller-vocabulary format that fully uses its ceiling outscore a richer one that under-uses its own. \emph{Guarded metrics}: distributional fidelity by FMD over CLaMP-2 embeddings; caption alignment by CLaMP-3 with all systems routed through the same MIDI$\to$MTF path; per-token loss used strictly within a representation (across vocabularies it is meaningless, PMT's NLL is the \emph{highest} while its generations are the best, because its prediction target is the richest).
\paragraph{Decoder and formal guarantees.}\label{apx:decoder-and-formal-guarantees} The decoder $R_{\text{PMT}}^{-1}$ inverts Eq.~\eqref{eq:pmt} in closed form. Reading the stream as per-note tuples with the running \texttt{TRACK}/\texttt{PROG} carried forward, and writing $t_j,\delta_j,u_j$ for the \texttt{TSHIFT}/\texttt{DUR}/\texttt{VEL} bin indices emitted for note $j$, it reconstructs
\begin{equation}
\begin{aligned}
&\hat{o}_j=\textstyle\sum_{i\le j}\tau\,t_i,\quad \hat{d}_j=\tau\,\delta_j,\quad \hat{v}_j=4u_j{+}2,\\
&\hat{p}_j=p_j,\quad \hat{c}_j=c_j,
\end{aligned}
\label{eq:decode_inv}
\end{equation}
i.e.\ absolute onset is the running sum of time-shifts, duration and velocity are bin centers, and pitch/track pass through unchanged. Two properties follow. \textbf{(P1) Bounded distortion}: each quantizer rounds to the nearest lattice point, so every note's round-trip error is bounded a priori ($|\hat{d}_j-d_j|\le\tau/2$, per-shift timing error $\le\tau/2=5$\,ms, $|\hat{v}_j-v_j|\le 2$ velocity units); accumulated onset error stays small in practice ($2.3$\,ms median on real music vs.\ $53$\,ms for beat grids, Fig.~\ref{fig:tokenizer_compare}), and every music-range string is decodable, so decodability is $1.00$ by construction (empirical validity $0.99$--$1.00$). \textbf{(P2) Attribute locality}: since Eq.~\eqref{eq:decode_inv} expresses each output attribute through a \emph{single} token family, the Jacobian $\partial R_{\text{PMT}}^{-1}/\partial(\text{token family})$ is block-diagonal across $(p,o,d,v,c)$, so replacing one family's tokens moves that attribute and provably leaves the others exactly fixed. This is the formal basis for the single-family edits in \S Interpretability (shifting \texttt{VEL} by $+8$ bins changes velocity by $+26.1$ with $0$ pitch change), which entangled position/duration encodings (REMI) cannot guarantee.

This appendix gives the full encoding procedure summarized by Eq.~\eqref{eq:pmt} and the design choices behind it.

\paragraph{Encoding procedure.}\label{apx:encoding-procedure} Given a MIDI file whose instruments $\{I_k\}$ each carry a program $g_k$ (or a drum flag) and notes $(p,o,d,v)$ (pitch, onset, duration, velocity), the encoder $R_{\text{PMT}}$ produces the flat stream of Eq.~\eqref{eq:pmt} in five deterministic steps.
\begin{enumerate}
\item \emph{Flatten and order.} Collect all notes across instruments into one list, each tagged with its track index and program, and sort by onset $o$, breaking ties by ascending pitch so that chord tones follow a canonical low-to-high order.
\item \emph{Inter-onset differencing.} For note $j$ compute the gap $\Delta_j{=}o_j{-}o_{j-1}$ to its predecessor (with $o_0{=}0$, so $\Delta_1{=}o_1$). Timing is thus stored \emph{relatively}, making the stream shift-invariant and encoding a chord ($\Delta{=}0$) natively as a run of pitches at zero shift.
\item \emph{Quantize.} Map $(\Delta_j,d_j,v_j)$ to bin indices through the fixed quantizers of Eq.~\eqref{eq:quant}, giving $\texttt{TSHIFT}_{q_t(\Delta_j)}$, $\texttt{DUR}_{q_d(d_j)}$, $\texttt{VEL}_{q_v(v_j)}$, with time-shift and duration sharing one $\tau{=}10$\,ms lattice and velocity on a $32$-level lattice.
\item \emph{Long-gap tiling.} A single shift spans at most $100\tau{=}1000$\,ms, so a longer gap is emitted as $\lfloor\Delta_j/(100\tau)\rfloor$ maximal $\texttt{TSHIFT}_{99}$ tokens followed by one remainder shift, and arbitrarily long rests stay exactly representable.
\item \emph{Run-length instrument gating.} Emit a $\texttt{TRACK}_{c}\,\texttt{PROG}_{g}$ pair only when the active instrument changes (drum tracks map to $\texttt{PROG}_{128}$); otherwise emit just the per-note tuple $[\texttt{TSHIFT},\texttt{PITCH},\texttt{DUR},\texttt{VEL}]$. Gating the program tokens by change, rather than repeating them per note, is what yields the ${\approx}4$ tokens/note in Table~\ref{tab:eff}.
\end{enumerate}
Decoding (Eq.~\eqref{eq:decode_inv}) inverts each step in closed form: the running sum of \texttt{TSHIFT} recovers absolute onsets, \texttt{DUR}/\texttt{VEL} are read as bin centers, and the carried \texttt{TRACK}/\texttt{PROG} routes each note to its instrument.

\paragraph{Worked example.}\label{apx:worked-example} A real caption, its notes and the exact token stream the shipped tokenizer emits for them are given in Appendix~\ref{app:worked}, together with the round trip.

\begin{table}[h]
\centering\small\setlength{\tabcolsep}{4pt}
\begin{tabular}{lrl}
\toprule
Family & \#Sym. & Encodes \\
\midrule
\texttt{TRACK}  & 16  & track index (multi-instrument routing) \\
\texttt{PROG}   & 129 & GM program; $128{=}$drums \\
\texttt{PITCH}  & 128 & MIDI pitch $0$--$127$ \\
\texttt{DUR}    & 200 & note duration, $10$--$2000$\,ms @ $10$\,ms \\
\texttt{VEL}    & 32  & velocity level (dynamics) \\
\texttt{TSHIFT} & 100 & inter-onset gap, $10$--$1000$\,ms @ $10$\,ms \\
special         & 4   & \textsc{bos}/\textsc{eos}/\textsc{pad}/\textsc{bar} \\
\midrule
Total & \textbf{609} & flat stream, ${\approx}4$ tokens/note \\
\bottomrule
\end{tabular}
\caption{PMT vocabulary layout ($\mathcal{V}_{\text{music}}$, Eq.~\eqref{eq:vocab}).}
\label{tab:vocab}
\end{table}

\paragraph{Design rationale.}\label{apx:design-rationale} Four choices define the representation. \emph{(a) Flat single stream.} Unlike grouped encodings (Compound Word, Octuple) that pack a note's attributes into one super-token and therefore need a multi-head output, PMT emits each parameter as its own token, so any off-the-shelf LLM consumes it with a single softmax and no architecture change. \emph{(b) Onset-order interleaving.} All instruments share one time-ordered stream rather than being serialized track-by-track, so self-attention sees the full vertical texture at each moment, which is what lets the model coordinate multi-instrument arrangement (the capacity the single-track ablation removes, \S Ablation Studies). \emph{(c) $10$\,ms performance lattice.} Storing the raw inter-onset interval on a fixed absolute lattice, rather than a tempo-relative beat grid (Eq.~\eqref{eq:beatgrid}), preserves swing, rubato, and micro-timing; this is the single variable the controlled swap isolates and the source of PMT's onset-timing advantage. \emph{(d) Explicit velocity.} A dedicated $32$-level \texttt{VEL} token keeps per-note dynamics that beat grids discard; ablating it costs $40$ FMD points (\S Ablations). The vocabulary is deliberately small ($609$ symbols): it adds under $0.3\%$ to the $248$k-token backbone embedding, so the pretrained text distribution is barely perturbed and the music tokens are learned as a compact ``foreign vocabulary'' on top of it.

\subsection{A Worked Example: One Caption, Its Notes, Its Tokens}
\label{app:worked}
The vocabulary layout is in Table~\ref{tab:vocab} and the encoding rules in this appendix; this is one
generation carried through all three, so both can be checked rather than taken on trust. Every number
below is emitted by the shipped tokenizer on the shipped file.

\paragraph{The caption.} The text conditioning the generation, verbatim:

\begin{quote}
\emph{A bright and yearning film-score cue for piano, strings and soft synth pad, with a rising chord progression and a lyrical melody that climbs into a warm climax, uplifting and cinematic.}
\end{quote}

\paragraph{The MIDI.} The first 6 notes of the generated file, in onset order, exactly as
\texttt{pretty\_midi} reads them. Onset and duration are seconds, velocity is the MIDI 0--127 scale,
and the program is the General MIDI instrument the model chose for that track.

\par\medskip\noindent\begin{minipage}{\columnwidth}
\centering\small\setlength{\tabcolsep}{3pt}
\begin{tabular}{rrrrrl}
\toprule
onset (s) & dur (s) & pitch & vel & trk & GM program \\
\midrule
0.00 & 0.50 & 51 & 78 & 0 & 0 (Acoustic Grand Piano) \\
0.00 & 0.50 & 58 & 78 & 0 & 0 (Acoustic Grand Piano) \\
0.00 & 0.50 & 63 & 78 & 0 & 0 (Acoustic Grand Piano) \\
0.00 & 0.50 & 67 & 78 & 0 & 0 (Acoustic Grand Piano) \\
0.00 & 0.50 & 55 & 78 & 1 & 41 (Viola) \\
0.50 & 0.50 & 53 & 78 & 0 & 0 (Acoustic Grand Piano) \\
\bottomrule
\end{tabular}
\captionof{table}{\textbf{The notes.} The first 6 notes of the generated file as
\texttt{pretty\_midi} reads them; \texttt{trk} is the index its \texttt{TRACK} token carries.}
\label{tab:wenotes}
\end{minipage}\par\medskip

\paragraph{The MIDI, as the file stores it.} The table above is already an abstraction: a MIDI file
holds \emph{delta-time events}, and note duration is not stored at all but implied by a later
\texttt{note\_on} of velocity zero. This is the same six notes at that level, which is what the
tokenizer reads:

\par\medskip\noindent\begin{minipage}{\columnwidth}
\centering\scriptsize
\begin{tabular}{@{}l@{}}
\texttt{ticks\_per\_beat=220, format 1, 3 tracks} \\
\texttt{track 0:} \\
\texttt{  dt=  0  set\_tempo  500000 us/beat} \\
\texttt{  dt=  0  time\_signature  4/4} \\
\texttt{  dt=  1  end\_of\_track} \\
\texttt{track 1:   (= trk 0 in the table above)} \\
\texttt{  dt=  0  program\_change  ch=0 program=0} \\
\texttt{  dt=  0  note\_on  ch=0 note=51 vel=78} \\
\texttt{  dt=  0  note\_on  ch=0 note=58 vel=78} \\
\texttt{  dt=  0  note\_on  ch=0 note=63 vel=78} \\
\texttt{  dt=  0  note\_on  ch=0 note=67 vel=78} \\
\texttt{  dt=220  note\_on  ch=0 note=51 vel=0} \\
\texttt{  dt=  0  note\_on  ch=0 note=53 vel=78} \\
\texttt{  dt=  0  note\_on  ch=0 note=58 vel=0} \\
\texttt{  dt=  0  note\_on  ch=0 note=63 vel=0} \\
\texttt{  dt=  0  note\_on  ch=0 note=67 vel=0} \\
\texttt{  dt=220  note\_on  ch=0 note=53 vel=0} \\
\texttt{  dt=  1  end\_of\_track} \\
\texttt{track 2:   (= trk 1 in the table above)} \\
\texttt{  dt=  0  program\_change  ch=1 program=41} \\
\texttt{  dt=  0  note\_on  ch=1 note=55 vel=78} \\
\texttt{  dt=220  note\_on  ch=1 note=55 vel=0} \\
\texttt{  dt=  1  end\_of\_track} \\
\end{tabular}
\captionof{table}{\textbf{The same notes as MIDI events.} Every event the file holds, meta track
included. Note-off is a \texttt{note\_on} of velocity zero, and \texttt{dt} is a delta time in
ticks, so seconds follow only from \texttt{ticks\_per\_beat} together with \texttt{set\_tempo}.}
\label{tab:weevents}
\end{minipage}\par\medskip

\noindent Two properties of this format drive the design. Timing is in \emph{ticks} against a
\texttt{ticks\_per\_beat} grid plus a tempo map, so an absolute onset in seconds is only recoverable by
integrating tempo changes; and note-off is a separate event, so duration must be reconstructed by pairing.
PMT stores an absolute $10$\,ms shift and an explicit duration instead, which is what makes the stream
tempo-free and invertible without pairing.

\paragraph{The tokens.} Encoding those 6 notes gives 27 tokens
(4.50 per note including the two sentinels):

\begin{center}\small
\parbox{0.95\columnwidth}{\raggedright \texttt{<BOS>} \texttt{TRACK\_0} \texttt{PROG\_0} \texttt{PITCH\_51} \texttt{DURP\_49} \texttt{VEL\_19} \texttt{PITCH\_58} \texttt{DURP\_49} \texttt{VEL\_19} \texttt{PITCH\_63} \texttt{DURP\_49} \texttt{VEL\_19} \texttt{PITCH\_67} \texttt{DURP\_49} \texttt{VEL\_19} \texttt{TRACK\_1} \texttt{PROG\_41} \texttt{PITCH\_55} \texttt{DURP\_49} \texttt{VEL\_19} \texttt{TSHIFT\_49} \texttt{TRACK\_0} \texttt{PROG\_0} \texttt{PITCH\_53} \texttt{DURP\_49} \texttt{VEL\_19} \texttt{<EOS>}}
\end{center}

\noindent Reading it. \texttt{<BOS>} opens and \texttt{<EOS>} closes the stream. Notes are emitted in
(onset, track, pitch) order, and each contributes
\texttt{PITCH}\,$\to$\,\texttt{DURP}\,$\to$\,\texttt{VEL}, preceded by a \texttt{TSHIFT} when the
onset advances and by a \texttt{(TRACK, PROG)} pair when the instrument changes. Both are \emph{omitted
otherwise}, which is the economy of the scheme: the four notes of the opening chord share one onset, so
three of them cost three tokens each and no timing token at all, and a passage that stays on one instrument
pays for its program once.

\emph{Indices count units, starting at one.} \texttt{TSHIFT}$_k$ is a gap of $(k{+}1)\tau$ and
\texttt{DURP}$_k$ a duration of $(k{+}1)\tau$ with $\tau{=}10$\,ms, since a zero-length shift or note
is never emitted; the $0.50$\,s gap and the $0.50$\,s durations above therefore appear as
\texttt{TSHIFT\_49} and \texttt{DURP\_49}. A gap longer than one token spans ($100\tau{=}1$\,s) is
tiled by repeated maximal shifts. \texttt{VEL}$_k$ is one of 32 bins over the MIDI $0$--$127$
range, so velocity $78$ lands in bin $19$.

The four special symbols are \texttt{<PAD>}, \texttt{<BOS>}, \texttt{<EOS>} and \texttt{<BAR>}.
\texttt{<BAR>} belongs to the beat-grid mode and is unused in the performance-timed mode reported
throughout the paper, which is why no bar marker appears here. Total music vocabulary: 609
symbols.

\paragraph{Four encoder rules, located in real generations.} The example above shows a chord and one
onset advance. These four cases were found by scanning generated files for a window that actually
exhibits each behaviour and then encoding that window, so the strings below are emitted by the encoder:

\par\medskip\noindent\begin{minipage}{\columnwidth}
\centering\scriptsize
\begin{tabular}{@{}p{0.30\columnwidth}p{0.64\columnwidth}@{}}
\toprule
\textbf{Long gap, tiled} \newline \emph{the 1.52\,s gap needs two tokens: one maximal \texttt{TSHIFT\_99} ($1$\,s) plus the remainder} & \texttt{TRACK\_0} \texttt{PROG\_0} \texttt{PITCH\_67} \texttt{DURP\_41} \texttt{VEL\_17} \texttt{TSHIFT\_99} \texttt{TSHIFT\_51} \texttt{PITCH\_55} \texttt{DURP\_199} \texttt{VEL\_11} \texttt{PITCH\_67} \texttt{DURP\_199} \texttt{VEL\_13} \texttt{TSHIFT\_0} \texttt{PITCH\_60} \texttt{DURP\_199} \texttt{VEL\_11} \\[4pt]
\textbf{Drum track} \newline \emph{a drum track is carried by \texttt{PROG\_128}, outside the 0--127 GM range} & \texttt{TRACK\_0} \texttt{PROG\_48} \texttt{PITCH\_83} \texttt{DURP\_10} \texttt{VEL\_21} \texttt{TSHIFT\_0} \texttt{TRACK\_1} \texttt{PROG\_128} \texttt{PITCH\_39} \texttt{DURP\_1} \texttt{VEL\_26} \texttt{TSHIFT\_28} \texttt{TRACK\_2} \texttt{PROG\_0} \texttt{PITCH\_31} \texttt{DURP\_49} \texttt{VEL\_25} \texttt{PITCH\_57} \texttt{DURP\_17} \texttt{VEL\_20} \\[4pt]
\textbf{Instrument change} \newline \emph{each change of instrument re-emits \texttt{(TRACK, PROG)}; notes in between do not} & \texttt{TRACK\_0} \texttt{PROG\_0} \texttt{PITCH\_64} \texttt{DURP\_13} \texttt{VEL\_12} \texttt{TRACK\_1} \texttt{PROG\_40} \texttt{PITCH\_69} \texttt{DURP\_13} \texttt{VEL\_12} \texttt{PITCH\_73} \texttt{DURP\_13} \texttt{VEL\_12} \texttt{PITCH\_76} \texttt{DURP\_13} \texttt{VEL\_12} \\[4pt]
\textbf{Velocity range} \newline \emph{velocities 62, 66, 70 land in distinct \texttt{VEL} bins, so dynamics survive the encoding} & \texttt{TRACK\_0} \texttt{PROG\_0} \texttt{PITCH\_67} \texttt{DURP\_50} \texttt{VEL\_17} \texttt{TSHIFT\_0} \texttt{PITCH\_63} \texttt{DURP\_49} \texttt{VEL\_16} \texttt{PITCH\_72} \texttt{DURP\_49} \texttt{VEL\_17} \texttt{TSHIFT\_0} \texttt{PITCH\_48} \texttt{DURP\_48} \texttt{VEL\_15} \\[4pt]
\bottomrule
\end{tabular}
\captionof{table}{\textbf{Four encoder rules, each located in a real generation.} The token
strings are emitted by the shipped encoder on the window named in the left column.}
\label{tab:werules}
\end{minipage}\par\medskip

\paragraph{The round trip.} Decoding those tokens returns 6 notes with pitch and velocity
exact, maximum onset error 0.0\,ms and maximum duration error 0.0\,ms, both within the
$\tau/2{=}5$\,ms bound of Proposition~\ref{prop:exact}.

\subsection{Paired Data Samples: MIDI, Tokens, Score}
\label{app:paired}
Four generations shown as the three artefacts that must agree: the events a MIDI file actually stores, the token stream the shipped tokenizer emits for them, and the same notes engraved. Excerpts are short so the stream fits on the page and can be checked symbol by symbol; the score is engraved from the identical excerpt, so the three views cannot drift apart. Section conventions (index bases, omission rules, special symbols) are stated once in Appendix~\ref{app:worked} and not repeated per example.

\paragraph{Example A05-1.}\label{apx:ex-A051} Chosen for a chord and one onset advance: three of the four chord tones cost no timing token. Instruments: track 0: Acoustic Grand Piano (\texttt{PROG\_0}), track 1: Viola (\texttt{PROG\_41}).

\emph{Caption.}
\begin{quote}
\emph{A bright and yearning film-score cue for piano, strings and soft synth pad, with a rising chord progression and a lyrical melody that climbs into a warm climax, uplifting and cinematic.}
\end{quote}

\par\medskip\noindent\begin{minipage}{\columnwidth}
\centering\small\setlength{\tabcolsep}{3pt}
\begin{tabular}{rrrrr}
\toprule
onset (s) & dur (s) & pitch & vel & trk \\
\midrule
0.00 & 0.50 & 51 & 78 & 0 \\
0.00 & 0.50 & 58 & 78 & 0 \\
0.00 & 0.50 & 63 & 78 & 0 \\
0.00 & 0.50 & 67 & 78 & 0 \\
0.00 & 0.50 & 55 & 78 & 1 \\
0.50 & 0.50 & 53 & 78 & 0 \\
\bottomrule
\end{tabular}
\captionof{table}{\textbf{Example A05-1, the notes.} In the tokenizer's emission order.
Onset and duration in seconds, velocity on the MIDI 0--127 scale, \texttt{trk} the source track,
which is the index its \texttt{TRACK} token carries.}
\label{tab:notesA051}
\end{minipage}\par\medskip

\par\medskip\noindent\begin{minipage}{\columnwidth}
\centering\scriptsize
\begin{tabular}{@{}l@{}}
\texttt{ticks\_per\_beat=220, format 1, 3 tracks} \\
\texttt{track 0:} \\
\texttt{  dt=   0  set\_tempo  500000 us/beat} \\
\texttt{  dt=   0  time\_signature  4/4} \\
\texttt{track 1:   (= trk 0 in the table above)} \\
\texttt{  dt=   0  program\_change  ch=0 program=0} \\
\texttt{  dt=   0  note\_on  ch=0 note=51 vel=78} \\
\texttt{  dt=   0  note\_on  ch=0 note=58 vel=78} \\
\texttt{  dt=   0  note\_on  ch=0 note=63 vel=78} \\
\texttt{  dt=   0  note\_on  ch=0 note=67 vel=78} \\
\texttt{  dt= 220  note\_on  ch=0 note=51 vel=0} \\
\texttt{  dt=   0  note\_on  ch=0 note=53 vel=78} \\
\texttt{  dt=   0  note\_on  ch=0 note=58 vel=0} \\
\texttt{  dt=   0  note\_on  ch=0 note=63 vel=0} \\
\texttt{  dt=   0  note\_on  ch=0 note=67 vel=0} \\
\texttt{  dt= 220  note\_on  ch=0 note=53 vel=0} \\
\texttt{track 2:   (= trk 1 in the table above)} \\
\texttt{  dt=   0  program\_change  ch=1 program=41} \\
\texttt{  dt=   0  note\_on  ch=1 note=55 vel=78} \\
\texttt{  dt= 220  note\_on  ch=1 note=55 vel=0} \\
\end{tabular}
\captionof{table}{\textbf{Example A05-1, the same notes as MIDI events.} Every event the
file holds, meta track included. Delta times are in ticks: at \texttt{ticks\_per\_beat}${=}$220 and \texttt{set\_tempo}${=}$500000\,$\mu$s/beat, a $dt$ of $220$ is $220/220\times0.5=0.5000$\,s, the 0.50\,s of the table above.}
\label{tab:eventsA051}
\end{minipage}\par\medskip

\emph{PMT tokens} (27 for 6 notes,
4.50 per note including the two sentinels):
\begin{center}\small
\parbox{0.95\columnwidth}{\raggedright \texttt{<BOS>} \texttt{TRACK\_0} \texttt{PROG\_0} \texttt{PITCH\_51} \texttt{DURP\_49} \texttt{VEL\_19} \texttt{PITCH\_58} \texttt{DURP\_49} \texttt{VEL\_19} \texttt{PITCH\_63} \texttt{DURP\_49} \texttt{VEL\_19} \texttt{PITCH\_67} \texttt{DURP\_49} \texttt{VEL\_19} \texttt{TRACK\_1} \texttt{PROG\_41} \texttt{PITCH\_55} \texttt{DURP\_49} \texttt{VEL\_19} \texttt{TSHIFT\_49} \texttt{TRACK\_0} \texttt{PROG\_0} \texttt{PITCH\_53} \texttt{DURP\_49} \texttt{VEL\_19} \texttt{<EOS>}}
\end{center}

\emph{Engraved from the same excerpt.}
\begin{center}
\includegraphics[width=0.95\columnwidth]{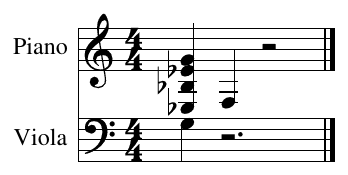}
\end{center}

\noindent\emph{Round trip:} 6 notes recovered, pitch and velocity
exact, maximum onset error 0.0\,ms and duration error
0.0\,ms, inside the $\tau/2{=}5$\,ms bound of Proposition~\ref{prop:exact}.

\paragraph{Example C10-1.}\label{apx:ex-C101} Chosen for a single-instrument passage, where the \texttt{(TRACK, PROG)} pair is paid once. Instruments: track 0: Acoustic Grand Piano (\texttt{PROG\_0}).

\emph{Caption.}
\begin{quote}
\emph{A hymn arrangement for piano, organ and string ensemble in four-part harmony, all voices sounding together throughout, dignified and warm.}
\end{quote}

\par\medskip\noindent\begin{minipage}{\columnwidth}
\centering\small\setlength{\tabcolsep}{3pt}
\begin{tabular}{rrrrr}
\toprule
onset (s) & dur (s) & pitch & vel & trk \\
\midrule
0.00 & 0.75 & 53 & 78 & 0 \\
0.00 & 0.01 & 55 & 127 & 0 \\
0.00 & 0.75 & 57 & 78 & 0 \\
0.00 & 0.75 & 60 & 78 & 0 \\
0.00 & 0.75 & 65 & 78 & 0 \\
0.75 & 0.75 & 46 & 78 & 0 \\
0.75 & 0.75 & 58 & 78 & 0 \\
0.75 & 0.38 & 62 & 78 & 0 \\
\bottomrule
\end{tabular}
\captionof{table}{\textbf{Example C10-1, the notes.} In the tokenizer's emission order.
Onset and duration in seconds, velocity on the MIDI 0--127 scale, \texttt{trk} the source track,
which is the index its \texttt{TRACK} token carries.}
\label{tab:notesC101}
\end{minipage}\par\medskip

\par\medskip\noindent\begin{minipage}{\columnwidth}
\centering\scriptsize
\begin{tabular}{@{}l@{}}
\texttt{ticks\_per\_beat=220, format 1, 2 tracks} \\
\texttt{track 0:} \\
\texttt{  dt=   0  set\_tempo  500000 us/beat} \\
\texttt{  dt=   0  time\_signature  4/4} \\
\texttt{track 1:   (= trk 0 in the table above)} \\
\texttt{  dt=   0  program\_change  ch=0 program=0} \\
\texttt{  dt=   0  note\_on  ch=0 note=53 vel=78} \\
\texttt{  dt=   0  note\_on  ch=0 note=55 vel=127} \\
\texttt{  dt=   0  note\_on  ch=0 note=57 vel=78} \\
\texttt{  dt=   0  note\_on  ch=0 note=60 vel=78} \\
\texttt{  dt=   0  note\_on  ch=0 note=65 vel=78} \\
\texttt{  dt=   4  note\_on  ch=0 note=55 vel=0} \\
\texttt{  dt= 326  note\_on  ch=0 note=46 vel=78} \\
\texttt{  dt=   0  note\_on  ch=0 note=53 vel=0} \\
\texttt{  dt=   0  note\_on  ch=0 note=57 vel=0} \\
\texttt{  dt=   0  note\_on  ch=0 note=58 vel=78} \\
\texttt{  dt=   0  note\_on  ch=0 note=60 vel=0} \\
\texttt{  dt=   0  note\_on  ch=0 note=62 vel=78} \\
\texttt{  dt=   0  note\_on  ch=0 note=65 vel=0} \\
\texttt{  dt= 167  note\_on  ch=0 note=62 vel=0} \\
\texttt{  dt= 163  note\_on  ch=0 note=46 vel=0} \\
\texttt{  dt=   0  note\_on  ch=0 note=58 vel=0} \\
\end{tabular}
\captionof{table}{\textbf{Example C10-1, the same notes as MIDI events.} Every event the
file holds, meta track included. Delta times are in ticks: at \texttt{ticks\_per\_beat}${=}$220 and \texttt{set\_tempo}${=}$500000\,$\mu$s/beat, a $dt$ of $4$ is $4/220\times0.5=0.0091$\,s, the 0.01\,s of the table above.}
\label{tab:eventsC101}
\end{minipage}\par\medskip

\emph{PMT tokens} (29 for 8 notes,
3.62 per note including the two sentinels):
\begin{center}\small
\parbox{0.95\columnwidth}{\raggedright \texttt{<BOS>} \texttt{TRACK\_0} \texttt{PROG\_0} \texttt{PITCH\_53} \texttt{DURP\_74} \texttt{VEL\_19} \texttt{PITCH\_55} \texttt{DURP\_0} \texttt{VEL\_31} \texttt{PITCH\_57} \texttt{DURP\_74} \texttt{VEL\_19} \texttt{PITCH\_60} \texttt{DURP\_74} \texttt{VEL\_19} \texttt{PITCH\_65} \texttt{DURP\_74} \texttt{VEL\_19} \texttt{TSHIFT\_74} \texttt{PITCH\_46} \texttt{DURP\_74} \texttt{VEL\_19} \texttt{PITCH\_58} \texttt{DURP\_74} \texttt{VEL\_19} \texttt{PITCH\_62} \texttt{DURP\_37} \texttt{VEL\_19} \texttt{<EOS>}}
\end{center}

\emph{Engraved from the same excerpt.}
\begin{center}
\includegraphics[width=0.95\columnwidth]{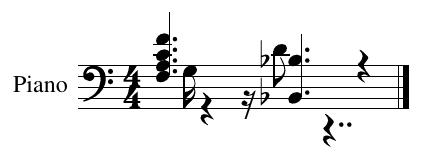}
\end{center}

\noindent\emph{Round trip:} 8 notes recovered, pitch and velocity
exact, maximum onset error 0.0\,ms and duration error
0.9\,ms, inside the $\tau/2{=}5$\,ms bound of Proposition~\ref{prop:exact}.

\paragraph{Example A13-1.}\label{apx:ex-A131} Chosen for a dense multi-instrument opening, where the pair is re-emitted at every change. Instruments: track 0: Acoustic Bass (\texttt{PROG\_32}), track 1: Pizzicato Strings (\texttt{PROG\_45}).

\emph{Caption.}
\begin{quote}
\emph{A chamber arrangement for piano, cello and violin with an interlocking arpeggio accompaniment and a long-breathed melody, intimate and cinematic.}
\end{quote}

\par\medskip\noindent\begin{minipage}{\columnwidth}
\centering\small\setlength{\tabcolsep}{3pt}
\begin{tabular}{rrrrr}
\toprule
onset (s) & dur (s) & pitch & vel & trk \\
\midrule
0.00 & 0.29 & 37 & 50 & 0 \\
0.00 & 0.29 & 25 & 50 & 1 \\
0.00 & 0.29 & 37 & 50 & 1 \\
0.00 & 0.29 & 49 & 50 & 1 \\
0.29 & 0.29 & 36 & 50 & 0 \\
0.29 & 0.29 & 24 & 50 & 1 \\
0.29 & 0.29 & 36 & 50 & 1 \\
0.29 & 0.29 & 48 & 50 & 1 \\
\bottomrule
\end{tabular}
\captionof{table}{\textbf{Example A13-1, the notes.} In the tokenizer's emission order.
Onset and duration in seconds, velocity on the MIDI 0--127 scale, \texttt{trk} the source track,
which is the index its \texttt{TRACK} token carries.}
\label{tab:notesA131}
\end{minipage}\par\medskip

\par\medskip\noindent\begin{minipage}{\columnwidth}
\centering\scriptsize
\begin{tabular}{@{}l@{}}
\texttt{ticks\_per\_beat=220, format 1, 3 tracks} \\
\texttt{track 0:} \\
\texttt{  dt=   0  set\_tempo  500000 us/beat} \\
\texttt{  dt=   0  time\_signature  4/4} \\
\texttt{track 1:   (= trk 0 in the table above)} \\
\texttt{  dt=   0  program\_change  ch=0 program=32} \\
\texttt{  dt=   0  note\_on  ch=0 note=37 vel=50} \\
\texttt{  dt= 128  note\_on  ch=0 note=36 vel=50} \\
\texttt{  dt=   0  note\_on  ch=0 note=37 vel=0} \\
\texttt{  dt= 127  note\_on  ch=0 note=36 vel=0} \\
\texttt{track 2:   (= trk 1 in the table above)} \\
\texttt{  dt=   0  program\_change  ch=1 program=45} \\
\texttt{  dt=   0  note\_on  ch=1 note=25 vel=50} \\
\texttt{  dt=   0  note\_on  ch=1 note=37 vel=50} \\
\texttt{  dt=   0  note\_on  ch=1 note=49 vel=50} \\
\texttt{  dt= 128  note\_on  ch=1 note=24 vel=50} \\
\texttt{  dt=   0  note\_on  ch=1 note=25 vel=0} \\
\texttt{  dt=   0  note\_on  ch=1 note=36 vel=50} \\
\texttt{  dt=   0  note\_on  ch=1 note=37 vel=0} \\
\texttt{  dt=   0  note\_on  ch=1 note=48 vel=50} \\
\texttt{  dt=   0  note\_on  ch=1 note=49 vel=0} \\
\texttt{  dt= 127  note\_on  ch=1 note=24 vel=0} \\
\texttt{  dt=   0  note\_on  ch=1 note=36 vel=0} \\
\texttt{  dt=   0  note\_on  ch=1 note=48 vel=0} \\
\end{tabular}
\captionof{table}{\textbf{Example A13-1, the same notes as MIDI events.} Every event the
file holds, meta track included. Delta times are in ticks: at \texttt{ticks\_per\_beat}${=}$220 and \texttt{set\_tempo}${=}$500000\,$\mu$s/beat, a $dt$ of $128$ is $128/220\times0.5=0.2909$\,s, the 0.29\,s of the table above.}
\label{tab:eventsA131}
\end{minipage}\par\medskip

\emph{PMT tokens} (35 for 8 notes,
4.38 per note including the two sentinels):
\begin{center}\small
\parbox{0.95\columnwidth}{\raggedright \texttt{<BOS>} \texttt{TRACK\_0} \texttt{PROG\_32} \texttt{PITCH\_37} \texttt{DURP\_28} \texttt{VEL\_12} \texttt{TRACK\_1} \texttt{PROG\_45} \texttt{PITCH\_25} \texttt{DURP\_28} \texttt{VEL\_12} \texttt{PITCH\_37} \texttt{DURP\_28} \texttt{VEL\_12} \texttt{PITCH\_49} \texttt{DURP\_28} \texttt{VEL\_12} \texttt{TSHIFT\_28} \texttt{TRACK\_0} \texttt{PROG\_32} \texttt{PITCH\_36} \texttt{DURP\_28} \texttt{VEL\_12} \texttt{TRACK\_1} \texttt{PROG\_45} \texttt{PITCH\_24} \texttt{DURP\_28} \texttt{VEL\_12} \texttt{PITCH\_36} \texttt{DURP\_28} \texttt{VEL\_12} \texttt{PITCH\_48} \texttt{DURP\_28} \texttt{VEL\_12} \texttt{<EOS>}}
\end{center}

\emph{Engraved from the same excerpt.}
\begin{center}
\includegraphics[width=0.95\columnwidth]{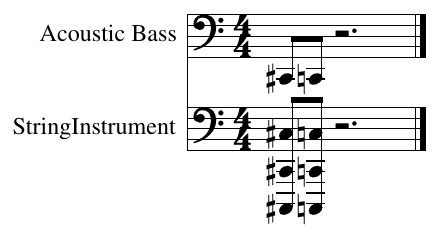}
\end{center}

\noindent\emph{Round trip:} 8 notes recovered, pitch and velocity
exact, maximum onset error 0.9\,ms and duration error
1.4\,ms, inside the $\tau/2{=}5$\,ms bound of Proposition~\ref{prop:exact}.

\paragraph{Example B018.}\label{apx:ex-B018} Chosen for a zero-shot benchmark caption, with the widest velocity spread of the four. Instruments: track 0: Acoustic Grand Piano (\texttt{PROG\_0}).

\emph{Caption.}
\begin{quote}
\emph{A slow-paced pop song with a touch of rock, this Christmas-themed composition exudes happiness and love through its melodic lines. Performed by a combination of piano, violin, string ensemble, drums, and clean electric guitar, it's set in the key of A minor with a 4/4 time signature. The chord progression of E7, Am, and Dm repeats frequently throughout the song, creating a sense of familiarity and comfort.}
\end{quote}

\par\medskip\noindent\begin{minipage}{\columnwidth}
\centering\small\setlength{\tabcolsep}{3pt}
\begin{tabular}{rrrrr}
\toprule
onset (s) & dur (s) & pitch & vel & trk \\
\midrule
2.40 & 0.90 & 53 & 127 & 0 \\
2.40 & 0.90 & 65 & 127 & 0 \\
2.40 & 0.90 & 69 & 127 & 0 \\
2.40 & 0.90 & 72 & 127 & 0 \\
3.30 & 0.90 & 55 & 127 & 0 \\
3.30 & 0.90 & 67 & 127 & 0 \\
3.30 & 0.90 & 71 & 127 & 0 \\
3.30 & 0.90 & 74 & 127 & 0 \\
\bottomrule
\end{tabular}
\captionof{table}{\textbf{Example B018, the notes.} In the tokenizer's emission order.
Onset and duration in seconds, velocity on the MIDI 0--127 scale, \texttt{trk} the source track,
which is the index its \texttt{TRACK} token carries.}
\label{tab:notesB018}
\end{minipage}\par\medskip

\par\medskip\noindent\begin{minipage}{\columnwidth}
\centering\scriptsize
\begin{tabular}{@{}l@{}}
\texttt{ticks\_per\_beat=220, format 1, 2 tracks} \\
\texttt{track 0:} \\
\texttt{  dt=   0  set\_tempo  500000 us/beat} \\
\texttt{  dt=   0  time\_signature  4/4} \\
\texttt{track 1:   (= trk 0 in the table above)} \\
\texttt{  dt=   0  program\_change  ch=0 program=0} \\
\texttt{  dt=1056  note\_on  ch=0 note=53 vel=127} \\
\texttt{  dt=   0  note\_on  ch=0 note=65 vel=127} \\
\texttt{  dt=   0  note\_on  ch=0 note=69 vel=127} \\
\texttt{  dt=   0  note\_on  ch=0 note=72 vel=127} \\
\texttt{  dt= 396  note\_on  ch=0 note=53 vel=0} \\
\texttt{  dt=   0  note\_on  ch=0 note=55 vel=127} \\
\texttt{  dt=   0  note\_on  ch=0 note=65 vel=0} \\
\texttt{  dt=   0  note\_on  ch=0 note=67 vel=127} \\
\texttt{  dt=   0  note\_on  ch=0 note=69 vel=0} \\
\texttt{  dt=   0  note\_on  ch=0 note=71 vel=127} \\
\texttt{  dt=   0  note\_on  ch=0 note=72 vel=0} \\
\texttt{  dt=   0  note\_on  ch=0 note=74 vel=127} \\
\texttt{  dt= 396  note\_on  ch=0 note=55 vel=0} \\
\texttt{  dt=   0  note\_on  ch=0 note=67 vel=0} \\
\texttt{  dt=   0  note\_on  ch=0 note=71 vel=0} \\
\texttt{  dt=   0  note\_on  ch=0 note=74 vel=0} \\
\end{tabular}
\captionof{table}{\textbf{Example B018, the same notes as MIDI events.} Every event the
file holds, meta track included. Delta times are in ticks: at \texttt{ticks\_per\_beat}${=}$220 and \texttt{set\_tempo}${=}$500000\,$\mu$s/beat, a $dt$ of $1056$ is $1056/220\times0.5=2.4000$\,s, the 2.40\,s of the table above.}
\label{tab:eventsB018}
\end{minipage}\par\medskip

\emph{PMT tokens} (32 for 8 notes,
4.00 per note including the two sentinels):
\begin{center}\small
\parbox{0.95\columnwidth}{\raggedright \texttt{<BOS>} \texttt{TSHIFT\_99} \texttt{TSHIFT\_99} \texttt{TSHIFT\_39} \texttt{TRACK\_0} \texttt{PROG\_0} \texttt{PITCH\_53} \texttt{DURP\_89} \texttt{VEL\_31} \texttt{PITCH\_65} \texttt{DURP\_89} \texttt{VEL\_31} \texttt{PITCH\_69} \texttt{DURP\_89} \texttt{VEL\_31} \texttt{PITCH\_72} \texttt{DURP\_89} \texttt{VEL\_31} \texttt{TSHIFT\_89} \texttt{PITCH\_55} \texttt{DURP\_89} \texttt{VEL\_31} \texttt{PITCH\_67} \texttt{DURP\_89} \texttt{VEL\_31} \texttt{PITCH\_71} \texttt{DURP\_89} \texttt{VEL\_31} \texttt{PITCH\_74} \texttt{DURP\_89} \texttt{VEL\_31} \texttt{<EOS>}}
\end{center}

\emph{Engraved from the same excerpt.}
\begin{center}
\includegraphics[width=0.95\columnwidth]{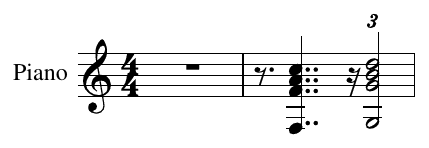}
\end{center}

\noindent\emph{Round trip:} 8 notes recovered, pitch and velocity
exact, maximum onset error 0.0\,ms and duration error
0.0\,ms, inside the $\tau/2{=}5$\,ms bound of Proposition~\ref{prop:exact}.

\section{Engraved-Score Comparison}
\label{app:staff}
\subsection{Qualitative Results}\label{apx:qualitative-results}
PMT generations exhibit caption-conditioned texture: folk captions yield single-line dance tunes, while classical, jazz, and pop captions yield sustained chordal writing with up to 13--17 simultaneous notes. On the multi-instrument sources, generations qualify as multi-instrument-with-chords ($\geq$2 instruments, $\geq$30\% chord-time) with coherent 5--9-instrument arrangements (Figure~\ref{fig:qual}), reaching 98--100\% chord-time on chord-requesting captions and 40+-second forms. Figure~\ref{fig:qual} makes the representation effect concrete: on one orchestral-soundtrack caption PMT renders an 8-instrument arrangement approaching the 11-instrument reference, while the beat-grid tokenizers collapse to 5 sparser tracks under identical backbone and budget. Every trained cell scores 6/6 valid on the sanity suite, and all examples ship with audio. The same contrast is legible as engraved notation in this appendix (Figure~\ref{fig:staff_compare}): MIDI-native systems notate full multi-instrument scores whereas the ABC baseline collapses to a two-staff keyboard reduction, so representation, not backbone, bounds the notated texture.

\begin{figure*}[tp]
\centering
\includegraphics[width=0.86\textwidth]{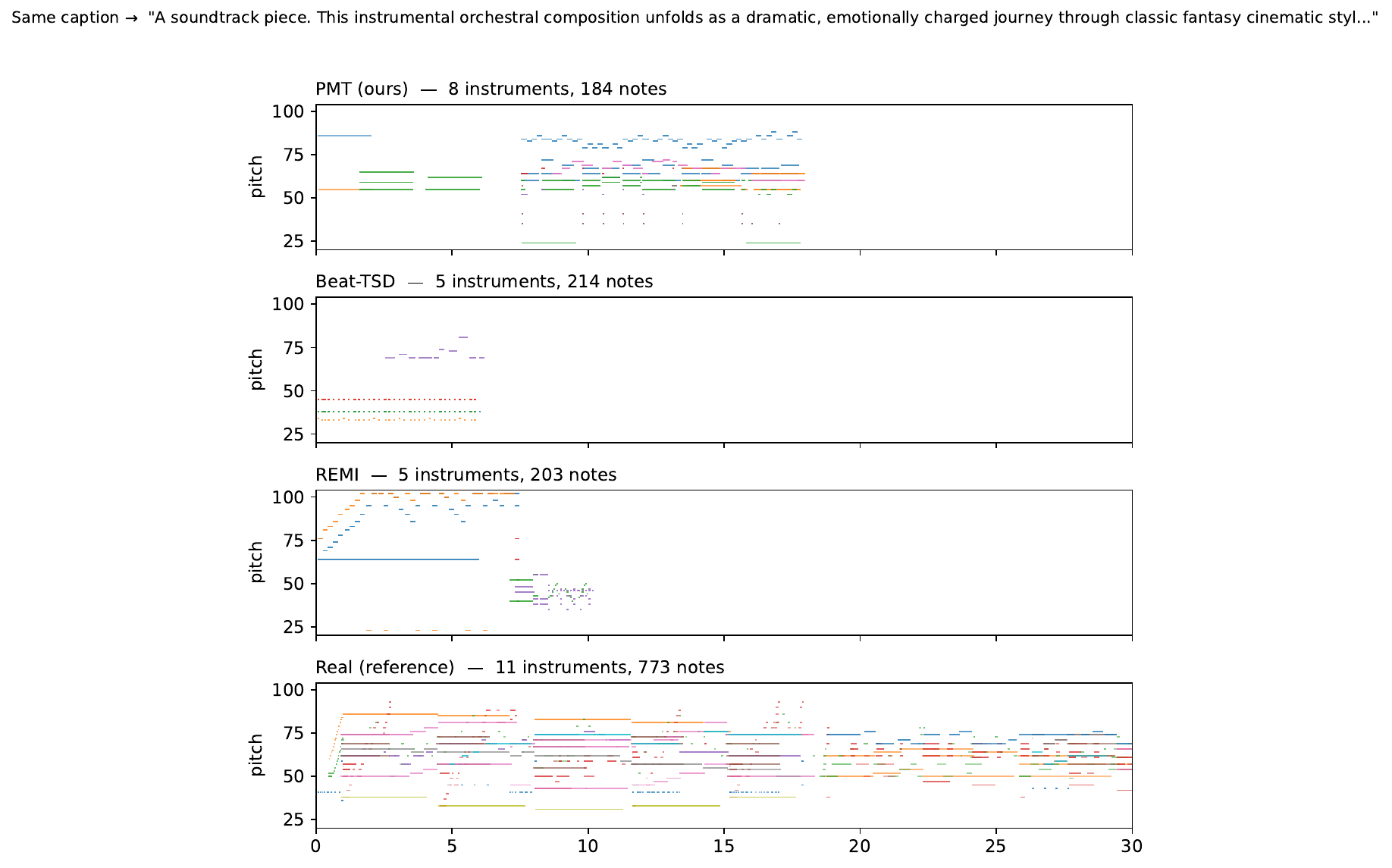}
\caption{\textbf{Same caption, different representation} (0.8B, identical backbone/data/budget). Piano rolls (colour $=$ instrument, opacity $=$ velocity) for one orchestral-soundtrack caption. PMT (top) generates an 8-instrument arrangement approaching the 11-instrument real reference (bottom), while beat-grid Beat-TSD and REMI collapse to 5 sparser tracks filling only the first seconds. Texture is set by the representation, not the backbone.}
\label{fig:qual}
\end{figure*}

Figures~\ref{fig:staff_compare}--\ref{fig:staff_rock} complement the instrument-colored piano rolls of Figures~\ref{fig:pipeline_compare} and~\ref{fig:qual} with a Western common-practice-notation view of the same-caption comparison, engraved with \textsc{verovio}, across three genres (pop, classical, rock). Each figure stacks, for one caption, the human reference, PMT (ours), the three \emph{newest} dedicated MIDI systems, MIDILM, Amadeus, MIDI-LLM (2025), and the ABC-based ChatMusician. They make the representation-level texture gap legible bar by bar and confirm that PMT's generations are well-formed enough to typeset as conventional multi-part scores.

\begin{figure*}[tp]
\centering
\includegraphics[width=\textwidth,height=0.84\textheight,keepaspectratio]{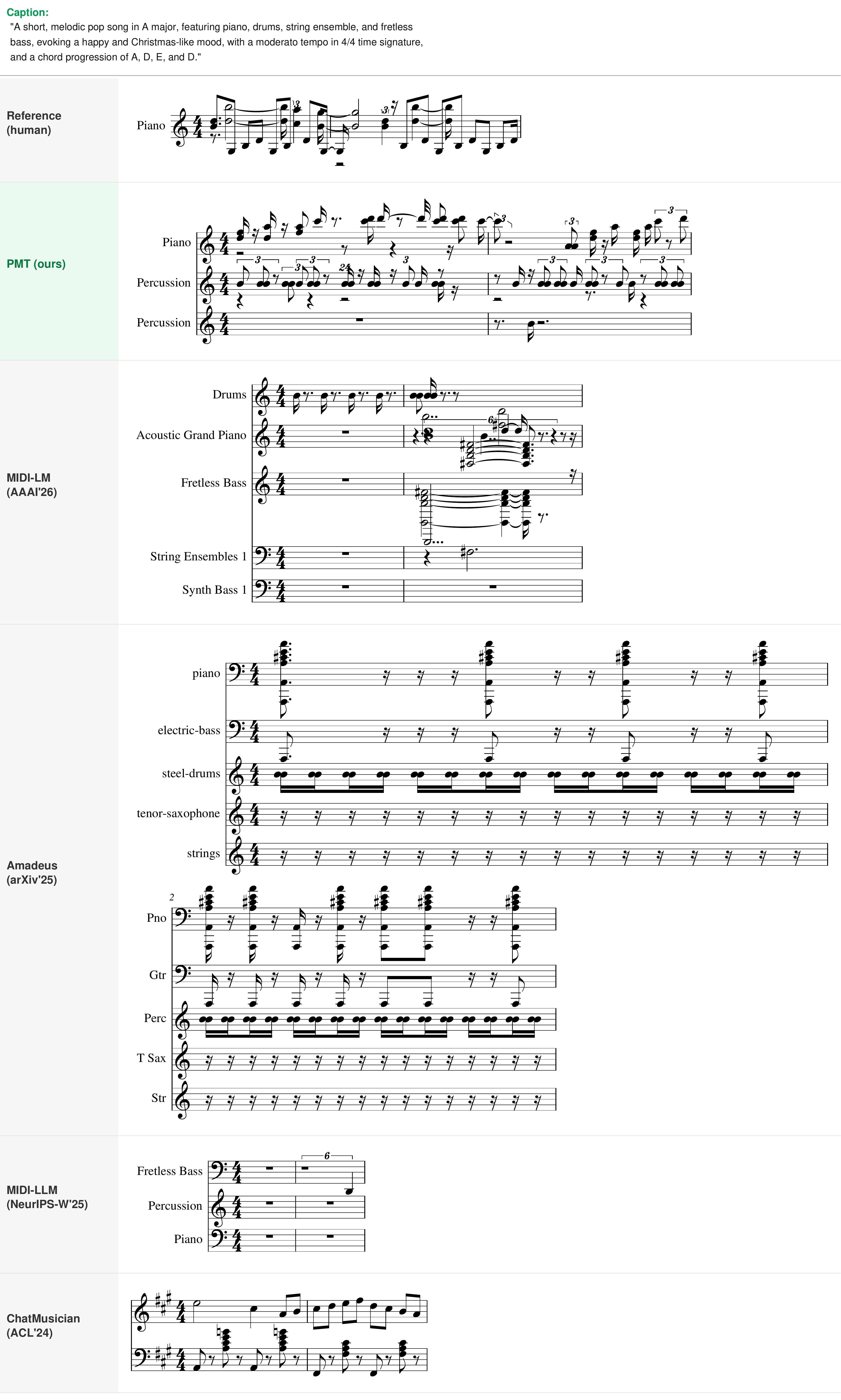}
\caption{\textbf{Same caption, engraved as sheet music (pop).} One MidiCaps caption drives six systems; each generation is notated from its own MIDI under one engraver.}
\label{fig:staff_compare}
\end{figure*}

\begin{figure*}[tp]
\centering
\includegraphics[width=\textwidth,height=0.84\textheight,keepaspectratio]{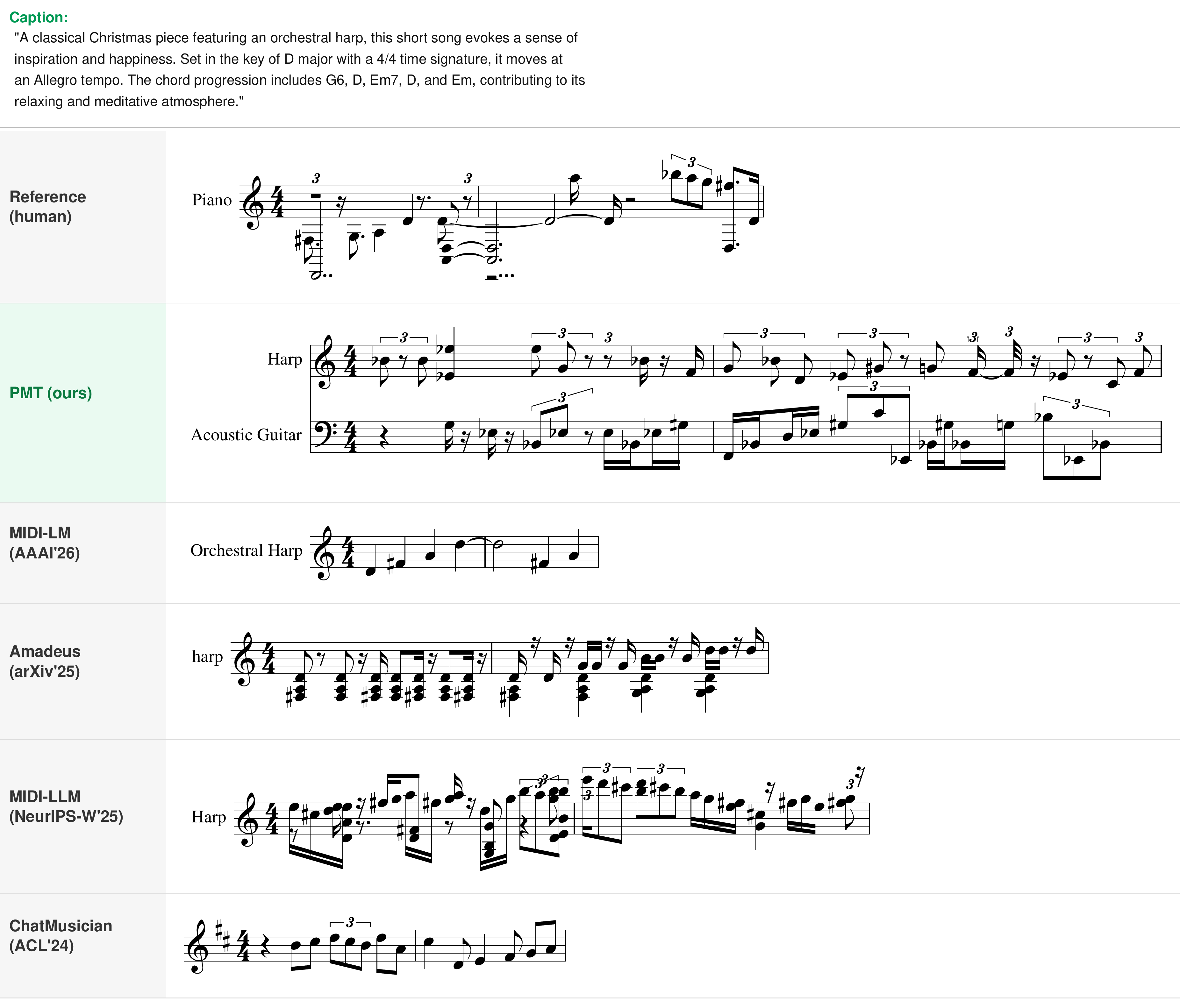}
\caption{\textbf{Instrument-specific caption (classical, ``orchestral harp'').} When the caption names a specific instrument, every MIDI-native system instantiates a harp part, MIDILM labels its staff \emph{Orchestral Harp}, Amadeus \emph{harp}, MIDI-LLM \emph{Harp}, and PMT writes a harp line with an accompanying guitar, whereas ChatMusician emits a single ABC melody line. The MIDI token vocabularies carry explicit instrument-program tokens; the ABC representation does not.}
\label{fig:staff_harp}
\end{figure*}

\begin{figure*}[tp]
\centering
\includegraphics[width=\textwidth,height=0.84\textheight,keepaspectratio]{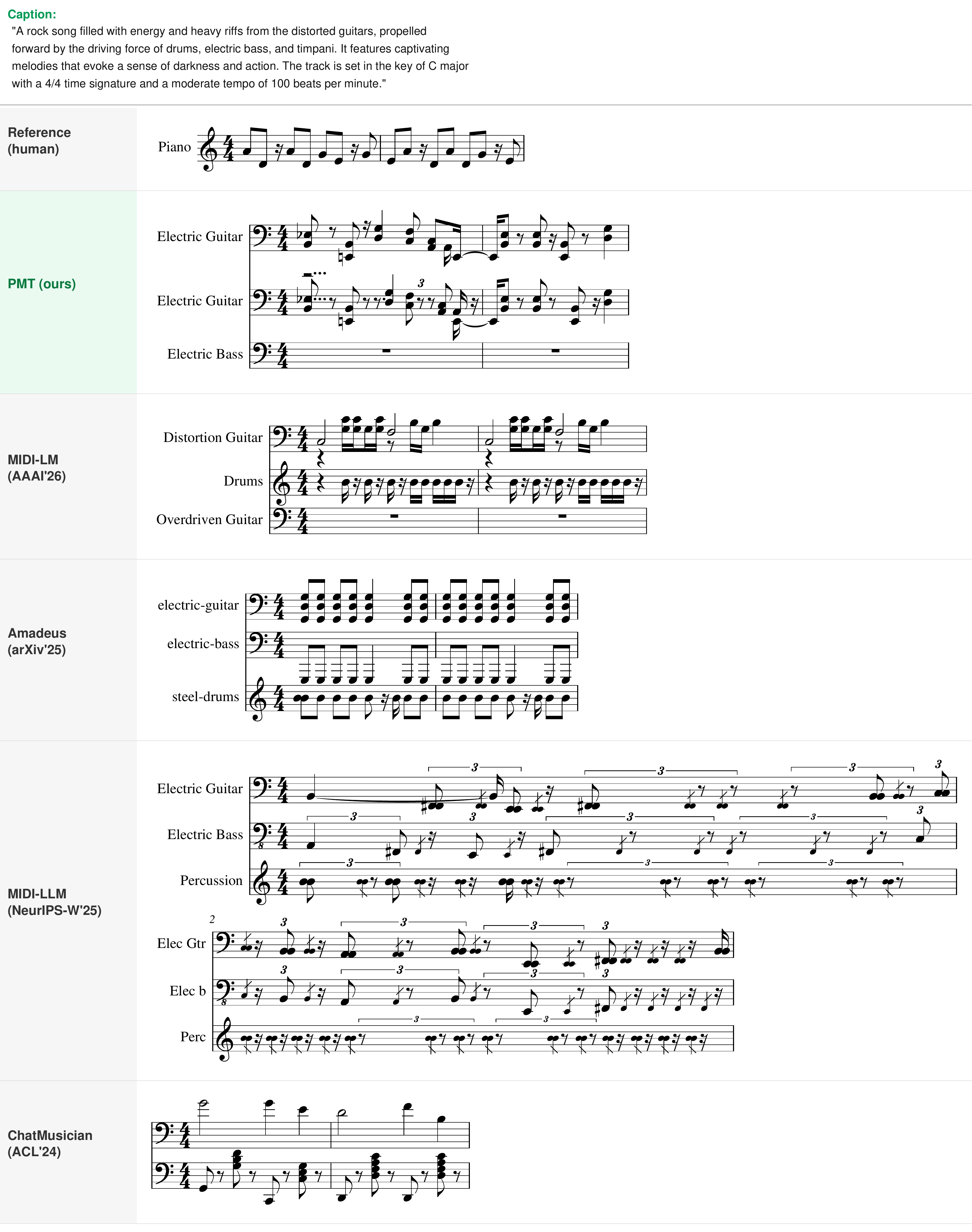}
\caption{\textbf{Rock caption (``distorted guitars, electric bass, drums'').} The MIDI-native systems realize the requested rock instrumentation, MIDILM's \emph{Distortion Guitar}/\emph{Overdriven Guitar}/\emph{Drums}, Amadeus's electric-guitar/bass/steel-drums, and PMT's dual electric guitars with bass, while ChatMusician again reduces to a two-staff keyboard part. Across all three genres the pattern is identical: the token representation, not the backbone, sets the achievable notated texture.}
\label{fig:staff_rock}
\end{figure*}

\section{Extended Distributional Metrics}
\label{app:extmetrics}
Tables~\ref{tab:fid} and~\ref{tab:muspy} give the full distributional-fidelity results summarized in the main text.

\begin{table}[t]
\centering
\small
\setlength{\tabcolsep}{4.5pt}
\begin{tabular}{llccc}
\toprule
Repr. & Size & FMD$\downarrow$ & CLaMP-3$\uparrow$ & Groove \\
\midrule
Real data & -- & 0 & -- & 0.984 \\
\midrule
\textbf{PMT} & 0.8B & \textbf{159$\pm$8} & .146 & \textbf{0.98} \\
\textbf{PMT} & 2B & \textbf{152$\pm$9} & .146 & \textbf{0.98} \\
\textbf{PMT} & 4B & \textbf{158$\pm$4} & .144 & \textbf{0.98} \\
\textbf{PMT} & 27B & \textbf{156} & .144 & -- \\
PerTok & 0.8B & 188$\pm$7 & .137 & -- \\
PerTok & 2B & 185 & .137 & -- \\
Beat-TSD & 0.8B & 286$\pm$1 & .151 & 0.60$^{*}$ \\
Beat-TSD & 2B & 280$\pm$3 & .157 & 0.61$^{*}$ \\
REMI & 0.8B & 285$\pm$1 & .151 & 0.64$^{*}$ \\
REMI & 2B--4B & 281$\pm$3 & .151--.154 & 0.66$^{*}$ \\
REMI & 27B & 272 & .149 & -- \\
MIDI-Like & 0.8B--4B & 280--285 & .152--.154 & 0.57$^{*}$ \\
MIDI-Like & 27B & 285 & .144 & -- \\
ABC & 0.8B--4B & 394--407 & .121--.124 & 1.00$^{*}$ \\
ABC & 9B & 388 & .117 & -- \\
\bottomrule
\end{tabular}
\caption{Distributional fidelity and alignment by representation and model size (mean$\pm$std over seeds). \textbf{FMD} $=$ Fr\'echet Music Distance over CLaMP-2 embeddings against a fixed $500$-piece held-out reference, corrected multi-track pipeline (lower is better); \textbf{CLaMP-3} $=$ caption--music alignment (higher is better); \textbf{Groove} $=$ groove consistency (closer to Real data is better). Values marked $^{*}$ traverse different decode paths
(miditok/abc2midi tick grids) and are reported for completeness, not ranked.}
\label{tab:fid}
\end{table}

Two things about this reference deserve stating. It is drawn from the folk source, the largest single source in the frozen split and a monophonic one (max-polyphony $2.0$, chord-time $9.0\%$), so it is narrower than the corpus, while the \emph{captions} driving every cell are stratified across all eight sources. Swapping in a source-stratified $500$-piece reference (chord-time $35.0\%$, max-polyphony $6.1$) leaves the ordering unchanged and \emph{widens} the PMT-vs-beat-grid gap from $1.8\times$ to $2.4\times$ (Appendix~\ref{app:limitations}), so the values reported here are the conservative ones. Read down the table, PMT's FMD is stable at 152--159 from 0.8B to 27B while beat-grid FMD stays at 271--286 even at 27B, with PerTok intermediate (185--188), completing the performance-resolution gradient. ABC's Groove of 1.00 reflects quantization to exact beat positions rather than musical quality.

\begin{table}[t]
\centering
\small
\setlength{\tabcolsep}{3pt}
\begin{tabular}{lcccccc}
\toprule
 & \multicolumn{3}{c}{\emph{structural} ($\to$ref)} & \multicolumn{3}{c}{\emph{timing/dynamics} $\downarrow$}\\
\cmidrule(lr){2-4}\cmidrule(lr){5-7}
Repr. & PCE & Scale & JSD$_{pc}$ & JSD$_{dur}$ & JSD$_{ioi}$ & Groove$_{\to}$ \\
\midrule
Real data & 2.611 & 0.992 & 0 & 0 & 0 & 0.984 \\
\textbf{PMT} & 2.593 & 0.973 & .012 & \textbf{.065} & \textbf{.036} & \textbf{0.980} \\
REMI & 2.597 & 0.975 & \textbf{.010} & .135 & .153 & 0.643$^{*}$ \\
Beat-TSD & 2.606 & 0.977 & .013 & .136 & .153 & 0.604$^{*}$ \\
MIDI-Like & 2.579 & 0.977 & .017 & .158 & .153 & 0.560$^{*}$ \\
PerTok & 2.570 & 0.974 & .015 & .070 & \textbf{.015} & 0.994$^{*}$ \\
ABC & 2.545 & 0.944 & .027 & .188 & .080 & 0.998$^{*}$ \\
\bottomrule
\end{tabular}
\caption{Public MusPy/MGEval suite, aggregated over all seeds and scales (0.8B--4B). Arrows give the direction of better; $\rightarrow$ means closer to the real reference.}
\label{tab:muspy}
\end{table}

\paragraph{Model-free FMD decomposition.}\label{apx:model-free-fmd-decomposition} Encoding then decoding real music with \emph{no model} yields a round-trip FMD of only $6.2$ for PMT and $14.4$ for Beat-TSD ($n{=}200$, whole untruncated pieces), both far below the \emph{trained} gap ($159$ vs.\ $\sim\!280$). The format-ceiling difference is thus ${\approx}8$ FMD points, under $10\%$ of the ${\approx}120$-point trained gap, so PMT's distributional advantage is a \emph{learnability} result the model realizes under supervision rather than a re-measured format ceiling; because the round-trip runs on whole pieces it also rules out truncation as the source.

\section{Decode Budget and the Length Bias of FMD}
\label{app:declen}
Every generation reported in this paper uses a decode budget of $900$ new tokens, which at PMT's ${\approx}4.0$ tokens per note is ${\approx}225$ notes. That budget is a property of the evaluation protocol, not of the representation, and it interacts with the reference set in a way worth stating explicitly.

\paragraph{The budget was binding.}\label{apx:the-budget-was-binding} Regenerating the same $100$ public-test captions with the budget raised to $7000$ tokens changes the output substantially: the median piece grows from $196$ to $704$ notes and from $18.2$ to $44.3$ seconds, chord-time fraction rises from $61.6\%$ to $67.8\%$, and median maximum polyphony from $8$ to $11$. Note count increases for $90/100$ captions and duration for $81/100$. The comparison is paired within caption, so the caption distribution cannot account for it. The $900$-token budget was therefore truncating pieces the model had learned to write, and the effect is largest for models trained with a long context window: at an $8192$-token training block, $900$ tokens is roughly $11\%$ of the sequence length the model was optimized over.

\paragraph{Longer is worse under FMD, and only because it is longer.}\label{apx:longer-is-worse-under-fmd-and} On those same $100$ captions, FMD degrades from $231.2$ at $900$ tokens to $299.6$ at $7000$ tokens, a $+68.4$ regression. This is a metric artifact rather than a musical one: $188$ of the $200$ reference pieces contain exactly $256$ notes, because the reference set is itself clipped, so distance to that reference structurally rewards short generations. Truncating each $7000$-token generation to its own reference's note count, which removes the length difference and leaves only the question of whether the music is closer to the reference, gives FMD $229.3$, i.e.\ $-1.9$ against the $900$-token arm and far inside the real-versus-real noise floor of ${\approx}43$ at this sample size. The longer generations are thus not distributionally worse music; they are penalized for being complete.

\paragraph{Consequences.}\label{apx:consequences} Two things follow. First, the comparisons in this paper remain internally valid: every arm decodes under the identical budget against the identical reference, so the representation ranking is unaffected, and the model-free round trip in Appendix~\ref{app:extmetrics} already rules out truncation as the source of the representation gap. Second, and more importantly for anyone using this benchmark, FMD against a clipped reference cannot be read as a proxy for musical completeness, and the two objectives pull in opposite directions: the decode budget that minimizes FMD is not the one that produces whole pieces. We report this rather than tuning the budget per arm, and we recommend that future work on this benchmark fix the budget in advance and report generation length alongside FMD.

\section{Design-Parameter Ablations}
\label{app:ablate}

\paragraph{Interpretability: disentangled, editable tokens.}\label{apx:interpretability-disentangled-editable-tokens} Model-free edit precision and the paired-caption check at the model level.

\textbf{Interpretability: disentangled, editable tokens.} Because PMT tokens are named parameters, single-family edits move only their own attribute with zero leakage (over 120 real pieces, velocity $+8$ bins $\to$ $+26.1$ realized, pitch $+5$ $\to$ exactly $+5.00$ semitones; full model-free matrix in this appendix, Table~\ref{tab:disentangle}), precision REMI-style position/duration entanglement cannot offer. At the model level, paired captions differing in one attribute (tempo, dynamics, or mode; 60 samples each) shift the realized attribute in the intended direction.

\paragraph{From-scratch validation in full.}\label{apx:from-scratch-validation-in-full} The MAESTRO reproduction, the POP909 transfer, and the multi-instrument control.

\textbf{From-scratch validation.} A web-pretrained backbone could favor some representations through prior exposure (it demonstrably favors ABC). Repeating the swap with a 26M decoder-only Transformer trained \emph{from scratch} on MAESTRO~\citep{hawthorne2018enabling} (changing only the tokenization) reproduces the ordering, PMT FMD 125.1$\pm$11.7 over four seeds vs.\ REMI's 344.4$\pm$9.8, and the win transfers to POP909~\citep{wang2020pop909} (full numbers in this appendix); with the Qwen3.5 grid, two independent backbone families agree, closing the pretraining-prior confound. The ordering also holds on \emph{multi-instrument} material trained from scratch, the case a solo-piano reproduction leaves open: on a 1{,}208-piece genre-balanced subset (median 7 instruments/piece, held-out 202-piece reference), PMT reaches FMD 287.9$\pm$1.0 (three seeds) versus 395.3$\pm$1.1 for Beat-TSD and 386.5$\pm$3.0 for REMI, a narrower but fully non-overlapping margin, so the from-scratch advantage is not specific to single-line piano.
\subsection{Ablation Studies}\label{apx:ablation-studies}
\begin{table}[t]
\centering\small
\setlength{\tabcolsep}{3.5pt}
\begin{tabular}{lccccc}
\toprule
\multicolumn{6}{c}{\emph{(a) JSD$_{ioi}$ with all onsets on a common grid}}\\
grid & PMT & Beat-TSD & REMI & PerTok & real floor\\
\midrule
none   & .030 & .160 & .145 & .012 & .008\\
10\,ms & .013 & .098 & .094 & .013 & .006\\
30\,ms & .006 & .059 & .062 & .008 & .004\\
60\,ms & \textbf{.003} & \textbf{.006} & .012 & .007 & .002\\
125\,ms& .003 & .002 & .007 & .005 & .001\\
\midrule
\multicolumn{6}{c}{\emph{(b) per-marginal JSD to real (0.8B); 95\% CI}}\\
        & IOI & dens. & poly & pitch & IOI 95\% CI\\
\midrule
PMT      & \textbf{.030} & .089 & .045 & \textbf{.002} & [.027,\,.051]\\
Beat-TSD & .159 & \textbf{.040} & \textbf{.021} & .004 & [.144,\,.184]\\
\midrule
\multicolumn{6}{c}{\emph{(c) FMD after snapping onsets to a common grid}}\\
grid & PMT & Beat-TSD & REMI & PMT lead & \\
\midrule
none    & \textbf{145} & 274 & 272 & \textbf{+127} & \\
60\,ms  & \textbf{205} & 274 & 272 & \textbf{+67} & \\
125\,ms & \textbf{147} & 280 & 276 & \textbf{+129} & \\
\bottomrule
\end{tabular}
\caption{\textbf{Controls delimiting the timing/FMD advantage} (headline reference and binning of Table~\ref{tab:main}; $n{=}400$ unless noted, $n{=}500$ in \emph{(c)}).}
\label{tab:controls}
\end{table}

\begin{figure*}[tp]
\centering
\includegraphics[width=\textwidth]{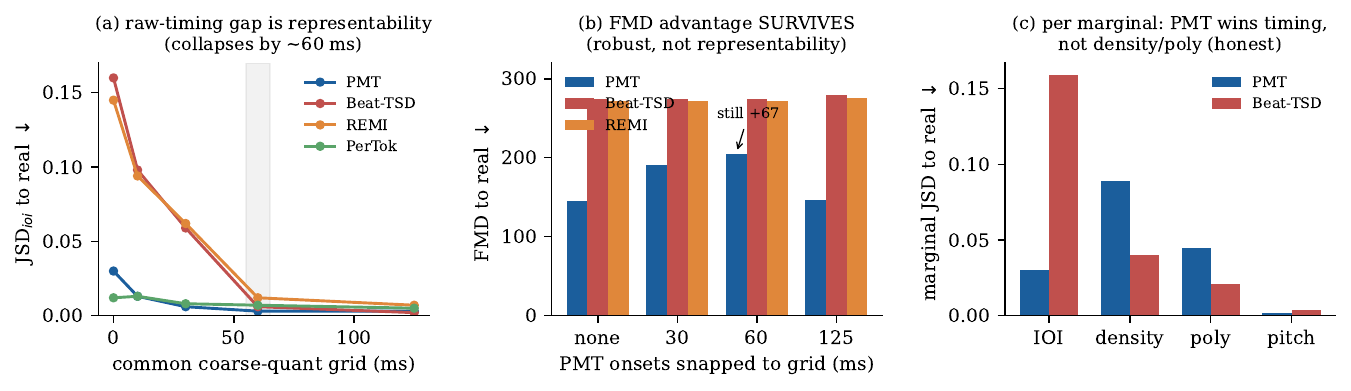}
\caption{\textbf{Controls on the timing/FMD advantage} (visualizing Table~\ref{tab:controls}).}
\label{fig:controls}
\end{figure*}

\textbf{What the timing advantage is, and is not.} Two controls (Table~\ref{tab:controls}, visualized in Figure~\ref{fig:controls}) delimit the effect, and we report them because the naive reading over-claims. \emph{(i) The raw onset-timing gap is representability, not learned placement.} Coarse-quantizing every arm's onsets \emph{and} the reference to a common grid shrinks PMT's JSD$_{ioi}$ lead over the beat grids monotonically until it \emph{closes near 60\,ms} (PMT $.003$ vs.\ Beat-TSD $.006$), exactly the beat grids' own resolution ($\rho\!\approx\!63$--$125$\,ms); PerTok, a second performance-microtiming encoding, already matches PMT natively. So the ``${\approx}4\times$ lower timing divergence'' reflects beat grids being structurally unable to place sub-grid onsets, and we do not read it as a quality difference. \emph{(ii) The FMD advantage is robust, not a fine-timing representability artifact.} The skeptic's hypothesis is that PMT wins only because it can place onsets on a finer grid than a beat grid physically allows. We test it directly: snap a trained PMT model's generated onsets onto a common \emph{fixed} (absolute-time) grid at the beat grids' resolution scale, removing the sub-grid timing while keeping pitch, duration, and velocity, and recompute FMD.\footnote{We snap to a fixed absolute grid rather than re-encoding through the beat-grid tokenizers themselves: generated output carries no canonical tempo, so a tempo-relative grid is ill-defined for it, and forcing it through the tokenizers' tempo re-estimation injects ${\approx}200$\,ms of onset distortion (FMD then reflects that re-gridding artifact, $386$ at $n{=}500$, not the removed sub-grid timing). The fixed grid isolates exactly the sub-grid timing at the beat grids' resolution.} The advantage \emph{survives at every grid} ($n{=}500$): at a $60$\,ms grid PMT scores $205$ versus $274$/$272$ for Beat-TSD/REMI (${+}67$), and at a $125$\,ms grid $147$ versus $280$/$276$ (${+}129$). The curve is non-monotone in the grid step because a fixed absolute grid interacts with each piece's natural subdivisions ($125$\,ms is beat-aligned at common tempi and $60$\,ms is not, so the $60$\,ms snap distorts more), but PMT leads at \emph{every} step, so the FMD gap is a genuine distributional difference, not the finer lattice. Two further controls agree: the model-free round-trip floor is only $8$ points (PMT $6.2$ vs.\ Beat-TSD $14.4$), and beat-quantizing \emph{real} music raises its FMD to just $14.4$, so CLaMP-2 does not merely penalize a beat-quantized surface. \emph{(iii) The gap is not an artifact of the embedder.} Because a single embedding space carries the headline number, we recompute the Fr\'echet distance over \textbf{CLaMP-3} embeddings, a different model family (text/audio/symbolic contrastive) that reads MIDI through a lossless tick-level serialization, on the \emph{same} generations against the \emph{same} reference, so only the embedder changes. Both the ordering and the magnitude replicate: PMT $110.4$ and PerTok $141.0$ against $206.3$--$216.6$ for the three beat grids, i.e.\ the beat grids sit $1.9\times$ above PMT under CLaMP-3 versus $1.8\times$ under CLaMP-2 (Appendix~\ref{app:limitations}). The class-level structure, both performance-resolution encodings ahead of every beat grid, therefore does not depend on CLaMP-2. We are equally careful about what we do \emph{not} claim: the raw onset-interval-histogram statistic (JSD$_{ioi}$) \emph{is} representability-dominated, collapsing to parity under the same common coarse-quantization (Table~\ref{tab:controls}); and the FMD gap does not reduce cleanly to a single crude marginal (beat grids are in fact closer on note density and polyphony), so we characterize it as a holistic distributional difference rather than over-attributing it to one feature. It is robust across scale, seeds, references, and a from-scratch backbone, but we make no perceptual claim (\S\ref{sec:perceptual}). PerTok~\citep{lenz2024pertok} corroborates that the binding property is performance-microtiming in general, not PMT's specific lattice: it matches PMT on JSD$_{ioi}$ at a $1.7\times$ token cost, and PMT vs.\ PerTok is within noise (overlapping CIs).

\textbf{Performance parameters are load-bearing.} Under the from-scratch MAESTRO setup, removing PMT's performance fields one at a time hurts: dropping velocity degrades FMD from 106.1 (the single-run from-scratch base, vs.\ the four-seed 125.1$\pm$11.7 for the same setup) to 145.7 (MGEval feature-overlap area, OA, $0.697 \to 0.450$) and coarsening micro-timing degrades FMD to 140.7, so the parameters PMT preserves, exactly those beat grids quantize away, carry measurable generative signal.

\textbf{Recipe: constrained decoding.} Ablating music-constrained decoding (Eq.~\eqref{eq:decode}), sampling over the full extended vocabulary with \emph{no} mask, still yields $100\%$ valid, decodable MIDI on 100 held-out captions: the vocabulary extension and caption-masked SFT teach the backbone to stay within the music range on its own, so the constraint is a \emph{robustness guarantee} (validity holds at any temperature and for undertrained checkpoints), not a crutch.

\textbf{Recipe: caption-masking is a convention, not the fidelity driver.} Training an otherwise-identical PMT-0.8B on the \emph{un}-masked caption$+$music loss (matched budget, two seeds, MidiCaps-200) adds only an auxiliary caption language-modeling term, yet reaches \emph{lower} FMD than the masked recipe ($181.6{\pm}3.5$ vs.\ $231.0{\pm}0.4$; both $0.99$ valid): the auxiliary objective mildly regularizes the small model, so our masked numbers are if anything \emph{conservative}. We keep masking for inference-alignment and a clean conditioning/generation split; whether the gain survives at scale is open. As with the constrained-decoding null, the advantage lives in the representation, not the recipe.

\textbf{Representation family (token vs.\ text).} Swapping any MIDI-token encoding for ABC halves attainment (.69--.92 $\to$ .46--.51) and cuts polyphony (4.1--5.8 $\to$ 3.0) under identical everything, a conservative estimate since ABC received every advantage (prior, compactness, more music per step). The gap follows from what ABC \emph{is}: a \emph{notation} for human sight-reading, recording metrical note values on a written grid with \emph{no} native fields for the sub-beat micro-timing and per-note velocity PMT preserves, and polyphony only through the cumbersome multi-voice (\texttt{V:}) syntax LLMs rarely emit, so generations collapse toward single-line tunes. This is why ABC's model-free round-trip ceiling is the lowest we measure (32.5\% chord-time vs.\ PMT's 39.1\%; Table~\ref{tab:eff}), before any training. ABC's celebrated compactness is a symptom of the same limitation: it stores fewer tokens because it stores less music.
\textbf{Multi-track coverage (natural experiment).} An earlier revision encoded only the first track of multi-instrument pieces for the baselines; their ceilings collapsed (35.5\% $\to$ 22.3\%) and raw scores followed while attainment held at $\sim$80--90\%, so the ceiling, the representation, is the binding constraint, and ceiling-anchored reporting caught in hours what raw scores would have misattributed to model quality.

\textbf{Multi-track encoding (direct ablation).} Training a PMT-0.8B on a cache with all instruments merged into one track (dropping \texttt{TRACK}/\texttt{PROG}) stays $1.00$ valid but collapses generated arrangements to a single instrument (mean $1.0$ vs.\ $1.8$); on multi-instrument-reference captions ($\geq$3 instruments, $n{=}16$) the multi-track model generates $5.25$ instruments (near the $6.4$ reference) versus $1.0$ for the ablation, while on single-instrument ($n{=}82$) both produce ${\sim}1$. Aggregate FMD is marginally \emph{lower} for the single-track variant ($149$ vs.\ $162$) precisely because $82\%$ of this test is single-instrument folk, so a folk-dominated FMD cannot credit the multi-instrument capacity \texttt{TRACK}/\texttt{PROG} provides, which is why we anchor texture claims to per-representation ceilings rather than raw FMD.
\textbf{Seeds and sampling noise.} Training-seed std is $\leq$2pp for baselines and ${\approx}6$pp for PMT at 0.8B, and repeated evaluation bounds sampling noise at $\approx$2.5pp ($n{=}100$); the PMT-vs-ABC gap ($>$14pp) is far outside noise. Re-evaluating the three 0.8B arms at $n{=}500$ leaves the timing separation intact (JSD$_{ioi}$ $.023$ for PMT vs.\ $.154$/$.164$ for Beat-TSD/REMI), so the $n{=}100$ protocol is not sample-limited on the metric that carries the thesis.

We isolate the two performance parameters PMT preserves, timing resolution and velocity granularity, both model-free (round-trip encode/decode of held-out real music) and with trained 0.8B models under the identical recipe. All trained single-checkpoint ablations here and in \S Ablation Studies compare against the \emph{seed-0} PMT-0.8B baseline (FMD 162, JSD$_{ioi}$ .030, attainment .96, chord-time 37.7\%), the same-seed control for each single-checkpoint variant; the corresponding 3-seed report values (159, .032, .92, 36.2\%) are those of Table~\ref{tab:main}.

\textbf{Timing resolution (model-free round-trip).} Encoding and decoding 300 held-out real pieces at coarsening \texttt{TSHIFT} resolutions preserves chord-time but destroys onset-timing fidelity (Table~\ref{tab:rt-shift}): inter-onset-interval divergence from the raw reference rises $8\times$ from 10\,ms to 20\,ms and $40\times$ by 100\,ms. The 10\,ms grid is the coarsest that keeps the real inter-onset distribution intact.

\begin{table}[h]
\centering
\small
\begin{tabular}{lccc}
\toprule
\texttt{TSHIFT} & Chord\% & JSD$_{ioi}\downarrow$ & vs.\ 10\,ms \\
\midrule
10\,ms (PMT) & 35.5 & \textbf{.015} & 1$\times$ \\
20\,ms & 36.8 & .126 & 8$\times$ \\
50\,ms & 37.1 & .114 & 7$\times$ \\
100\,ms & 37.3 & .620 & 40$\times$ \\
\bottomrule
\end{tabular}
\caption{Round-trip expressiveness vs.\ timing resolution (raw reference chord-time 37.1\% on this 300-piece subset; 36.7\% on the main 500-piece stratified set). Chord-time (simultaneity) is resolution-insensitive, but onset-timing realism (JSD$_{ioi}$, lower is better) collapses as the grid coarsens; the small 20/50\,ms non-monotonicity is quantization aliasing and it recurs in the trained-model ablation (Table~\ref{tab:trained-ablate}). This is the mechanism behind PMT's timing-axis advantage.}
\label{tab:rt-shift}
\end{table}

\textbf{Trained-model ablations.} We also train PMT-0.8B at coarser \texttt{TSHIFT} and evaluate under matched decoding (Table~\ref{tab:trained-ablate}): the round-trip prediction holds on \emph{generated} output, halving the timing resolution to 20\,ms worsens onset-timing divergence $5\times$ (JSD$_{ioi}$ .030$\to$.152), inflates FMD (162$\to$221), and drops attainment (.96$\to$.75). The 10\,ms grid is not an arbitrary choice: it is the operating point at which a trained model still reproduces the real inter-onset distribution, and any coarsening moves the model into the beat-grid regime.

\begin{table}[h]
\centering
\small
\begin{tabular}{lccc}
\toprule
PMT \texttt{TSHIFT} & Chord\% & Att.$\uparrow$ & JSD$_{ioi}\downarrow$ \\
\midrule
\textbf{10\,ms (default)} & \textbf{37.7} & \textbf{0.96} & \textbf{0.030} \\
20\,ms & 29.3 & 0.75 & 0.152 \\
50\,ms & 23.2 & 0.59 & 0.127 \\
100\,ms & 28.7 & 0.73 & 0.613 \\
\bottomrule
\end{tabular}
\caption{Trained-model timing-resolution ablation (PMT-0.8B, matched decoding).}
\label{tab:trained-ablate}
\end{table}

\textbf{Velocity granularity.} Dynamics is the second performance parameter PMT preserves. Coarsening velocity from 32 to 8 bins (Table~\ref{tab:vel-ablate}) collapses velocity-distribution realism ($8\times$: JSD$_{vel}$ .023$\to$.179) and chord fidelity, while leaving onset-timing (JSD$_{ioi}$) untouched; holistic FMD is insensitive to velocity granularity and barely moves (153 vs.\ 162), so the effect is specific to the velocity distribution it governs. Velocity and timing are thus \emph{independently} controllable, and \emph{having} velocity at all is load-bearing (dropping it entirely costs 40 FMD points in the from-scratch ablation, \S Ablations).

\begin{table}[h]
\centering
\small
\begin{tabular}{lcccc}
\toprule
PMT velocity & Chord\% & Att.$\uparrow$ & JSD$_{vel}\downarrow$ & JSD$_{ioi}\downarrow$ \\
\midrule
32 bins (default) & \textbf{37.7} & \textbf{.96} & \textbf{.023} & .030 \\
8 bins & 24.3 & .62 & .179 & .040 \\
\bottomrule
\end{tabular}
\caption{Trained-model velocity-granularity ablation (PMT-0.8B). Fewer velocity bins collapse dynamics realism (JSD$_{vel}$, $8\times$ worse) without affecting timing (JSD$_{ioi}$), the two performance parameters contribute independently.}
\label{tab:vel-ablate}
\end{table}

\textbf{Disentanglement matrix.} Table~\ref{tab:disentangle} makes the interpretability claim of \S Experiments quantitative: a model-free encode--edit--decode over 120 real multi-instrument pieces, editing one token family and measuring the change in \emph{every} attribute. Each edit moves only its own attribute; leakage to all others is \emph{exactly} zero.

\begin{table}[h]
\centering
\small
\setlength{\tabcolsep}{4pt}
\begin{tabular}{lcccc}
\toprule
Single-family edit & $\Delta$vel & $\Delta$pitch & $\Delta$dur & $\Delta$\#notes \\
\midrule
VEL$_k$\,$+8$\,bins & \textbf{+26.1}\,{\scriptsize$\pm$8.4} & 0.00 & 0.00 & 0.00 \\
PITCH$_n$\,$+5$\,st & 0.00 & \textbf{+5.00}\,{\scriptsize$\pm$.00} & 0.00 & 0.00 \\
\bottomrule
\end{tabular}
\caption{\textbf{Disentanglement matrix} (120 real multi-instrument pieces; model-free encode--edit--decode).}
\label{tab:disentangle}
\end{table}

\paragraph{From-scratch backbone (full numbers).}\label{apx:from-scratch-backbone-full-numbers} Repeating the representation swap with a 26M decoder-only Transformer trained from scratch (random init) on MAESTRO~\citep{hawthorne2018enabling}, PMT reaches FMD $125.1\pm11.7$ over four seeds versus REMI's $344.4\pm9.8$ (a $2.75\times$ gap, an order of magnitude beyond seed noise; TSD and MIDI-Like fall in between; MAESTRO held-out reference), and the win transfers unchanged to POP909~\citep{wang2020pop909} (PMT $175.8$ vs.\ REMI $497.1$, $200/200$ vs.\ $187/200$ valid). A second, non-pretrained backbone family thus agrees with the Qwen3.5 grid on the ordering, closing the pretraining-prior confound for the token-family comparison.

\section{Per-Source Texture Breakdown}
\label{app:persource}

\paragraph{Chord-time as a proxy: the full caveat.}\label{apx:chord-time-as-a-proxy-the} Per-source breakdown behind the main text's warning, including the per-genre error comparison.

\textbf{Chord-time is a coarse, mixture-confounded proxy.} Stress-testing our own headline statistic, a per-source breakdown (this appendix) shows the aggregate flatters no arm uniformly: PMT \emph{over}-harmonizes the folk majority (25\% vs.\ a 13\% reference) and \emph{under}-fills chordal classical (50\% vs.\ 82\%), so its blended match (36\% vs.\ 36.7\%) is partly cancellation across a 69\%-folk mixture, and REMI tracks each genre's chord count more closely (per-genre unweighted per-genre mean absolute error 4.8 vs.\ PMT's 14.0). We therefore rest the quality claim on the full-distribution metrics, FMD and the timing/duration divergences (Tables~\ref{tab:fid},~\ref{tab:muspy}), which are not chord-count cancellations and locate PMT's advantage on the timing axis. Unlike the chord-time cancellation, this advantage does \emph{not} wash out per source: split into folk and non-folk, PMT leads Beat-TSD on FMD in \emph{both} ($67$ vs.\ $220$ and $337$ vs.\ $492$; this appendix), so the headline lead is not a folk-mixture artifact.
Table~\ref{tab:persource} decomposes 0.8B chord-time by data source. No arm matches every genre's target: PMT regresses toward a middle density (over-harmonizing folk, under-filling classical), while REMI/Beat-TSD track chord \emph{count} per genre more closely. This is why we do not rest the quality claim on chord-time alone, the distributional (FMD) and timing (JSD$_{ioi}$/JSD$_{dur}$) metrics, which the main text uses as primary evidence, are computed over the full note distribution and are not subject to this mixture cancellation.

\begin{table}[h]
\centering
\small
\begin{tabular}{lcccc}
\toprule
Arm (0.8B) & Folk & Classical & Balanced & Jazz \\
 & \emph{(13)} & \emph{(82)} & \emph{(84)} & \emph{(66)} \\
\midrule
\textbf{PMT} & 25 & 50 & 73 & 65 \\
Beat-TSD & 13 & 73 & 66 & 60 \\
REMI & 11 & 83 & 75 & 73 \\
PerTok & 16 & 82 & 67 & 45 \\
ABC & 2 & 53 & 45 & 50 \\
\bottomrule
\end{tabular}
\caption{Per-source chord-time \% (0.8B); per-genre reference target in \emph{italics}. PMT's aggregate match masks per-genre over/undershoot; REMI is more faithful on chord count per genre. The paper's primary quality evidence is therefore FMD and timing divergence, not chord-time.}
\label{tab:persource}
\end{table}

\paragraph{Per-source FMD: the advantage is not a folk artifact.}\label{apx:per-source-fmd-the-advantage-is} Because the frozen test is folk-heavy (69\% single-line), a natural worry is that PMT's distributional lead reflects only that single-line folk melodies suit performance-event tokenization. We test this directly by recomputing FMD \emph{within} the folk split ($n{=}69$, single-line) and the non-folk split ($n{=}31$, multi-instrument classical/balanced/jazz), each against its same-source real reference. PMT leads Beat-TSD in \emph{both} splits: FMD $67.3$ vs.\ $219.7$ on folk and $336.8$ vs.\ $491.9$ on non-folk (absolute FMD is higher for both arms on the non-folk split because its $31$-piece multi-instrument reference distribution is wider and smaller-sample). The advantage therefore holds on multi-instrument material, not just folk, and corroborates the cross-domain MidiCaps result, where the 86.6k-trained PMT-4B scores FMD $196$ against $314$--$421$ for the published systems on an entirely non-folk (pop/classical) benchmark (\S Public-Benchmark Transfer). The single-track ablation's \emph{lower} aggregate FMD (\S Ablation Studies) is thus a folk-mixture effect on the blended test, not a sign that the multi-track encoding hurts distributional fidelity: split by source, the full multi-track PMT is closest to real music in each genre.

\section{Perceptual Evaluation Details}
\label{app:perceptual}

\paragraph{Qualitative results and the automatic listener in full.}\label{apx:qualitative-results-and-the-automatic-listener} The two subsections the main text merges, each at full length.

\subsection{Qualitative Results}\label{apx:qualitative-results-2}
PMT generations are caption-conditioned in texture: folk captions yield single-line dance tunes, while classical, jazz and pop captions yield sustained chordal writing with up to $13$--$17$ simultaneous notes and coherent $5$--$9$-instrument arrangements, reaching $98$--$100\%$ chord-time on chord-requesting captions. Appendix~\ref{app:staff} shows the same caption realized by each representation, in piano roll and as engraved notation: the beat-grid arms collapse to sparser tracks and the ABC baseline to a two-staff keyboard reduction under identical backbone and budget, so the representation, not the backbone, bounds the notated texture.
\subsection{Automatic MLLM Listener}
\label{sec:perceptual}
We probe audibility with an automatic listener (Qwen3-Omni-30B) that rates and forced-choice-compares rendered audio, and read it as a \emph{weak} proxy, short clips judged by a model rather than a human ear, that scopes where a human study should look rather than delivering a verdict. The signal is \emph{material-dependent}. On a model-free control that isolates micro-timing on expressive solo piano (MAESTRO performances round-tripped two ways, micro-timing preserved vs.\ beat-quantized, identical soundfont), the judge prefers the micro-timing-preserving version $60\%$ of the time ($54/90$ on that control, three personas each $60\%$; directional, $p{=}0.07$): the timing PMT keeps is audible where it is musically active. Where it is not, the signal vanishes as expected, on the generated PMT-vs-Beat-TSD contrast the judge does \emph{not} prefer PMT ($42\%$ overall, $n{=}420$; that set is folk-heavy and the effect is entirely folk-driven, $41\%$ on folk against chance on expressive multi-instrument material, $46\%$, $n{=}126$, n.s.), consistent with a source that is already grid-quantized. Absolute naturalness places PMT-0.8B on par with real performances ($3.43$ vs.\ $3.42$) but compresses to the mean. We therefore make \emph{no} perceptual claim from this proxy; it says only that the definitive test is a properly-powered \emph{human} study on \emph{expressive} material, which is ongoing (full protocols, persona breakdown, and real-clip anchor in this appendix).

The automatic MLLM listener study of \S Perceptual in full. \emph{Absolute rating}: 1--5 naturalness from three listener personas, 360 ratings with the real ground truth as a hidden anchor; PMT-0.8B $3.43$ vs.\ real $3.42$, REMI-27B $3.47$, density-overshooting MIDI-LLM lowest at $3.31$. \emph{Forced choice}: six listener-persona agents $\times$ 70 caption-matched items $\times$ three contrasts ($1{,}260$ head-to-heads; both clips in one prompt, ``which sounds more human-performed?''; two-sided binomial). On PMT vs.\ Beat-TSD ($n{=}420$, $42\%$ overall), PMT is chosen $41\%$ on single-line folk ($n{=}294$, where micro-timing is least active) and $46\%$ on expressive multi-instrument material ($n{=}126$, n.s.). Against the \emph{real} performances, expressive material exposes a ceiling neither reaches (real preferred over PMT $68\%$, over Beat-TSD $79\%$), whereas on flat folk the real renders are themselves dispreferred ($64\%$, a render-quality confound: bare folk melodies synthesize with little dynamic life). At scale the study \emph{corrects} a smaller underpowered pilot on this same generated contrast ($n{=}90$, $58\%$) that had suggested a mild PMT preference, underscoring that short-clip MLLM judgments are weak fine-timing perceivers.

\subsection*{Human listening study: instrument and pre-registered protocol}
\label{app:humanstudy}
The automatic probe above is a proxy, and it told us two useful things: any human study has to run on
\emph{expressive} material, because the effect is at chance on the folk half of our test set, and it has to be
forced-choice, because absolute ratings compress to the mean. We therefore built the study rather than describing
it, and we fix the protocol here \emph{before} collecting data so that the analysis cannot drift toward whatever
the first responses happen to show. \textbf{At submission time the instrument is complete and data collection is
underway; Table~\ref{tab:humanstudy} is the pre-registered reporting template and carries no results yet.}

\begin{figure}[tb]
\centering
\includegraphics[width=\linewidth]{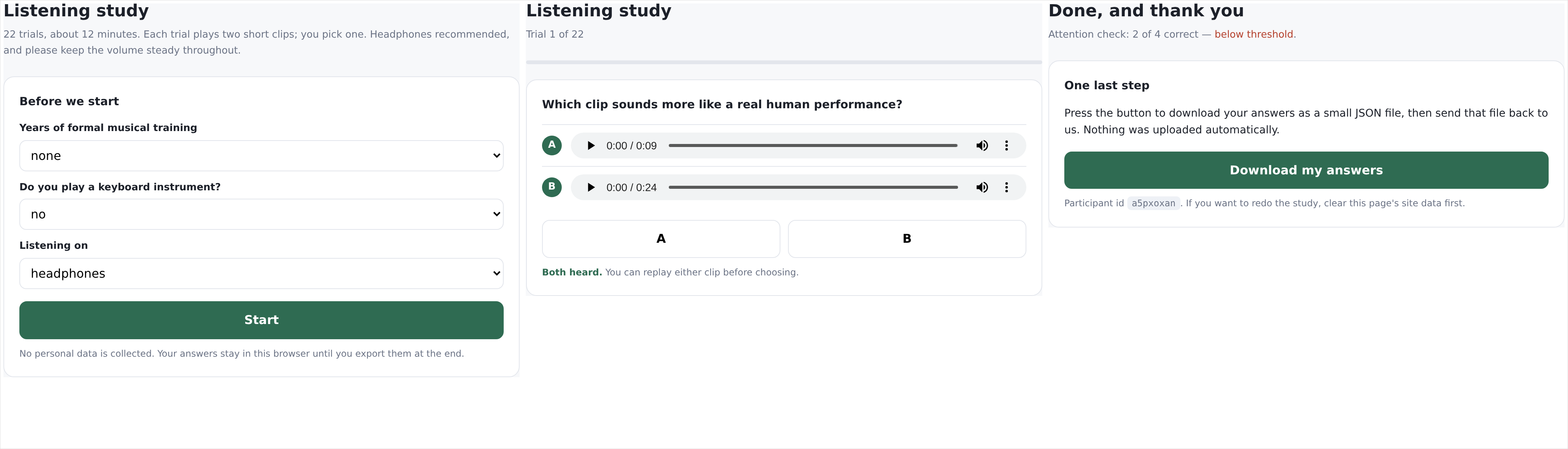}
\caption{\textbf{The listening-study interface} (captures of the deployed page). \emph{Left:} intake. \emph{Middle:} a trial. \emph{Right:} the closing screen.}
\label{fig:studyui}
\end{figure}

\paragraph{Design.}\label{apx:design} $16$ forced-choice trials per session, about $10$ minutes, sampled per session from a pool of $89$
pairs built on the $30$ chord-richest of the shared $100$ controlled-grid captions (chord-time
$96$--$100\%$ on the real reference). Four blocks: \emph{(A)} PMT-0.8B against Beat-TSD-0.8B on the same
caption, identical backbone, data and budget, $8$ trials, which is the representation question the paper
leaves open and therefore carries the most weight; \emph{(B1)} PMT-4B trained on the 6.25M corpus
against MIDI-LLM, $3$ trials; \emph{(B2)} the same model against text2midi, $3$ trials, both
system-level rather than representation-level; \emph{(C)} a real performance against the same
performance with $45\%$ of its pitches randomly displaced, $2$ trials, which has a known answer and
serves as the attention check. Both clips of a pair are cut to a common duration (the shorter of the
two, capped at $20$\,s) and normalised to $-18$\,LUFS, so neither length nor level can identify a
system; pairs whose common duration falls below $8$\,s are dropped rather than padded. Which pairs
appear, which side each system takes, and the order asked are all sampled per session. The question is identical in all three blocks: \emph{which clip sounds
more like a real human performance?}

\paragraph{Pre-registered analysis.}\label{apx:pre-registered-analysis} Block C has a known answer and gates the analysis: a session that does not answer every C trial
correctly is excluded from A, B1 and B2, a rule fixed before collection rather than applied after
inspection. The primary test is a two-sided binomial test on block A against $50\%$, with sessions as
the clustering unit; B1, B2 and C are secondary. Powering for the effect size our proxy suggests
($60/40$) at $\alpha{=}0.05$ and $80\%$ power needs $194$ A-trials, so $25$ sessions at $8$ each, and
allowing a design effect of $1.5$--$2$ for within-session correlation we target $40$--$50$ sessions. We will report
the preference rate with a bootstrap confidence interval, the exclusion count, and the by-covariate split
(trained vs.\ untrained listeners), whichever direction the result takes. A null result on block A would bound
the representation claim to the distributional evidence, which is how \S\ref{sec:perceptual} already frames it.
 \emph{Ethics.} Participation is voluntary and anonymous: a session carries a random code and three covariates (musical training, keyboard experience, playback device), and no name, e-mail or IP is recorded. Raters are told before starting what is recorded and how it is used, and may stop at any point; answers already given are retained and cannot be traced to a person. No compensation is offered and no personal data leaves the study.

\begin{table}[t]
\centering\small
\setlength{\tabcolsep}{4pt}
\begin{tabular}{llcc}
\toprule
Block & Contrast & /sess. & Result \\
\midrule
A & PMT vs.\ Beat-TSD & 8 & collecting \\
  & \emph{matched backbone/data/budget} & & \\
\midrule
B1 & ours vs.\ MIDI-LLM & 3 & collecting \\
B2 & ours vs.\ text2midi & 3 & collecting \\
  & \emph{same caption, blind} & & \\
\midrule
C & real vs.\ pitch-scrambled & 2 & collecting \\
  & \emph{attention check} & & \\
\bottomrule
\end{tabular}
\caption{\textbf{Human listening study: the pre-registered reporting template.} Blind forced choice, side randomized per trial, clips cut to a common length and loudness-normalized so neither duration nor level identifies a system, every answer written when submitted so partial sessions still count. Block A is the representation question the paper leaves open; B1 and B2 are system-level and cannot settle it; C has a known answer and gates the analysis. \emph{The table carries no results:} every block is reported in this form once it reaches its target sample, whichever direction it takes, including a null. The instrument is in Figure~\ref{fig:studyui}.}
\label{tab:humanstudy}
\end{table}

\section{Extended Baseline Comparisons}
\label{app:baselines}

\paragraph{MidiCaps public test: all rows and columns.}\label{apx:midicaps-public-test-all-rows-and} The compact table in the main text drops the tempo-bin and compression columns and the PMT-27B and zero-shot-LLM rows; they are here.

\begin{table*}[t]
\centering
\footnotesize
\setlength{\tabcolsep}{1pt}
\begin{tabular}{l@{\hspace{3pt}}ccccccc}
\toprule
System & Valid. & TB & TBT & CK & IF1 & CR$_z\downarrow$ & Poly \\
\midrule
\textbf{PMT-4B} (0-shot) & \textbf{1.00} & .24 & .44 & .10 & .46 & \textbf{4.8} & 10.6 \\
\textbf{PMT-27B} (0-shot) & \textbf{1.00} & .24 & .44 & .10 & .47 & 5.8 & \phantom{0}9.6 \\
\textbf{PMT-4B} 6.25M (0-shot) & \textbf{1.00} & .24 & .44 & .17 & .47 & 5.4 & 10.2 \\
MIDI-LLM (0-shot) & \textbf{1.00} & .24 & .44 & .32 & \textbf{.63} & 5.6 & 11.7 \\
text2midi (in-dom.) & 0.99 & .26 & .44 & --$^{\dagger}$ & .34 & 5.7 & 18.5 \\
ChatMusician (ABC) & \textbf{1.00} & -- & .45 & .36 & .29 & -- & \phantom{0}3.8 \\
MIDILM & \textbf{1.00} & .46 & .70 & .39 & .54 & 5.6 & 16.8 \\
Amadeus & \textbf{1.00} & .72 & \textbf{.90} & \textbf{.57} & .62 & 10.2 & 72.1 \\
\midrule
\multicolumn{8}{@{}l}{\emph{Zero-shot general LLMs (no music training; prompt$\rightarrow$ABC$\rightarrow$MIDI)}}\\
GPT-5.5 & 0.58 & -- & -- & .66 & .32 & -- & 11.4 \\
Qwen3.5-4B & \textbf{1.00} & -- & -- & .29 & .29 & -- & \phantom{0}2.5 \\
\bottomrule
\end{tabular}
\caption{MidiCaps public test (200 captions disjoint from our training data; identical extractors; adherence metrics are caption-relative, so no reference row applies). \textbf{Val.}\ $=$ validity$\uparrow$; \textbf{TB}/\textbf{TBT} $=$ tempo within $\pm$5\,bpm / $\pm$10\% of the caption tempo$\uparrow$; \textbf{CK} $=$ correct-key rate$\uparrow$; \textbf{IF1} $=$ instrument-F1 vs.\ the caption's named instruments$\uparrow$; \textbf{CR$_z$} $=$ zlib compression ratio$\downarrow$ (lower $=$ less repetitive); \textbf{Poly} $=$ max-polyphony. $^{\dagger}$text2midi's decoded MIDI is unparseable by the music21 key-finder. ``--'' marks metrics not extracted under an identical protocol.}
\label{tab:midicapsfull}
\end{table*}

The first two PMT rows are the 86.6k controlled-grid models; the third is the 6.25M one-epoch model, whose extra data nearly doubles key adherence (CK $.17$ vs.\ $.10$) at the same validity and tempo adherence. \emph{MIDILM} and \emph{Amadeus} lead caption-attribute adherence but over-texture heavily (max-poly $16.8$ and $72.1$ against a $6.1$ reference). \emph{ChatMusician} leads key adherence ($.36$) with near-single-line ABC output (max-poly $3.8$), the polyphony limit our controlled study predicts.

\paragraph{Newest systems and cross-modal baselines in full.}\label{apx:newest-systems-and-cross-modal-baselines} Per-system discussion behind the main text's summary, plus the CLAP comparison.

\paragraph{Newest dedicated systems (2025--26).}\label{apx:newest-dedicated-systems-2025-26} MIDILM~\citep{li2026midilm} (distinct from the MIDI-LLM baseline and the MIDI-Like tokenizer) and Amadeus~\citep{su2026amadeus} beat PMT on caption-attribute adherence (Amadeus TBT .90/CK .57 vs.\ PMT .44/.10) but by \emph{over-texturing} (max-polyphony 16.8 and 72.1 vs.\ a 6.1 reference). On the axis our thesis targets, distance to the real note distribution, PMT is decisively closest: same-protocol FMD \textbf{196} for the 86.6k-trained PMT-4B and \textbf{211} for the 6.25M one-epoch model, vs.\ MIDI-LLM 314, MIDILM 351, Amadeus 380, text2midi 421 ($1.5$--$2.1\times$ farther). Two things are worth separating here. First, no competitor is both adherent \emph{and} close-to-real. Second, the two PMT rows show the corpus buying adherence at a small fidelity cost: scaling $86.6$k$\to$$6.25$M nearly doubles caption-key adherence (CK $.10\!\to\!.17$, CKD $.17\!\to\!.32$) while FMD moves from $196$ to $211$, so the corpus helps where captions are the bottleneck and the representation result stands on the controlled grid either way.

\paragraph{Cross-modal and zero-shot baselines (this appendix).}\label{apx:cross-modal-and-zero-shot-baselines} A CLAP comparison against audio generators (MusicGen, AudioLDM\,2) does not cleanly rank quality (caption-matched real music scores \emph{lowest}, $31.5$), so we report but do not rank by it (why we omit Fr\'echet Audio Distance is in the appendix). Zero-shot general LLMs confirm prompting is not enough: GPT-5.5 leads key adherence ($.66$) by copying the stated key yet is playable for only $58\%$ of captions, and Qwen3.5-4B compiles $100\%$ only via trivial single-line tunes (poly $2.5$), so neither is reliable \emph{and} rich, separating our contribution from ``just prompt a bigger model.'' The reliable distributional axis remains FMD, where PMT leads.

\paragraph{Reading Table~\ref{tab:midicaps} row by row.}\label{apx:reading-table-row-by-row} MidiCaps public test (200 captions disjoint from our training data; identical extractors; adherence metrics are caption-relative, so no reference row applies; MidiCaps source MIDI not redistributed). \textbf{Val.}\ $=$ validity$\uparrow$; \textbf{TB}/\textbf{TBT} $=$ tempo within $\pm$5\,bpm / $\pm$10\% of the caption tempo$\uparrow$; \textbf{CK} $=$ correct-key rate$\uparrow$ (music21 Krumhansl); \textbf{IF1} $=$ instrument-F1 vs.\ the caption's named instruments$\uparrow$; \textbf{CR$_z$} $=$ zlib compression ratio of the note stream$\downarrow$ (lower $=$ less repetitive); \textbf{Poly} $=$ max-polyphony. PMT, never trained on MidiCaps-style captions, matches text2midi's tempo adherence (TBT .44 vs.\ .44) at perfect validity while text2midi's 18.5 polyphony reproduces its at-home imprint (Table~\ref{tab:baselines}). The first two PMT rows are the 86.6k controlled-grid models; the third is the 6.25M one-epoch model, whose extra data nearly doubles key adherence (CK $.17$ vs.\ $.10$) at the same validity and tempo adherence. \emph{ChatMusician} (a leading ABC music LLM) leads key adherence ($.36$) but its ABC output is near-single-line (max-poly $3.8$, $2.2$ instruments), the ABC polyphony limit our controlled study predicts. \emph{MIDILM} and \emph{Amadeus} (2025) lead caption-attribute adherence (Amadeus TBT .90/CK .57/IF1 .62 vs.\ PMT .44/.10/.46) but over-texture heavily (max-poly 16.8 and 72.1 vs.\ a 6.1 reference); on identical-protocol FMD they sit $1.8$--$1.9\times$ farther from real than PMT (see text). \emph{Zero-shot general LLMs} (GPT-5.5, Qwen3.5-4B prompted to write ABC): GPT-5.5 copies the stated key (CK .66) yet is playable for only $58\%$ of captions, and Qwen3.5-4B compiles $100\%$ only via trivial single-line tunes (poly $2.5$). $^{\dagger}$text2midi's decoded MIDI is unparseable by the music21 key-finder. ``--'' marks metrics not extracted under an identical protocol to the trained rows.

\begin{table*}[t]
\centering
\small
\setlength{\tabcolsep}{2.2pt}
\begin{tabular}{lccccccc}
\toprule
System & Val.$\uparrow$ & Chord\%$\rightarrow$ & Poly$\rightarrow$ & Notes & Dur & FMD$\downarrow$ & C3$\uparrow$ \\
\midrule
\textit{Reference (real)} & -- & \textit{36.7} & \textit{6.1} & -- & -- & \textit{0} & -- \\
\textbf{PMT-0.8B} & \textbf{1.00} & \textbf{36.2$\pm$5.6} & 5.5 & 181 & 31 & \textbf{159$\pm$8} & \textbf{.146} \\
MIDI-LLM 1B & 0.96 & 72 & 11.9 & 313 & 18.1 & 507 & .085 \\
text2midi & 0.97 & 66 & 17.3 & 224 & 14.3 & 727 & .076 \\
\midrule
\textit{AMT} (uncond.) & 1.00 & 79 & 36.8 & 159 & 21.9 & 720 & -- \\
\bottomrule
\end{tabular}
\caption{Published text-to-MIDI systems under our harness: identical 100 captions and metrics. \textbf{Val.}\ $=$ validity$\uparrow$; \textbf{Chord\%}/\textbf{Poly} $=$ chord-time / max-polyphony ($\to$ref); \textbf{Notes/Dur} $=$ mean note count / duration (s); \textbf{FMD}$\downarrow$ (protocol of Table~\ref{tab:fid}); \textbf{C3} $=$ CLaMP-3 caption--music alignment$\uparrow$ (same MIDI$\to$MTF path). The last row adds the Anticipatory Music Transformer \citep{thickstun2023anticipatory} through the same harness ($240$ samples); it is \emph{unconditional}, so it takes no caption and no alignment score applies. NotaGen/MuPT are not free-text-conditioned and serve as as-reported context.}
\label{tab:baselines}
\end{table*}

MIDI-LLM imprints its Lakh-heavy training distribution regardless of the caption: double the reference chord density, at $17$ notes/s. PMT instead matches the reference within 1pp in expectation (per-seed std $\pm5.6$, Table~\ref{tab:main}), is $3$--$4\times$ closer in FMD, and scores $1.7$--$1.9\times$ higher caption alignment. The Anticipatory Music Transformer row makes the imprinting effect most visible, and it is the reason we include an unconditional model at all: it is the strongest published exemplar of the fine-arrival-time token family PMT also belongs to, yet at perfect validity it emits $14.8$-instrument arrangements with $79\%$ chord-time and max-polyphony $36.8$, roughly $6\times$ the reference's $6.1$, which places it farthest of all systems from the reference distribution (FMD $720$). Fine arrival-time tokens alone therefore do not buy distributional fidelity; the surrounding recipe, what the model is trained to reproduce, does.

Table~\ref{tab:baselines}: on identical captions, \emph{both} published systems imprint their training distributions, MIDI-LLM emits 72\% chord-time and text2midi 66\% with 17.3 max-polyphony against a 36.7\%/6.1 reference, which a more-is-better reading would reward and our distance-to-reference protocol exposes; both are short (14--18\,s vs.\ our $\sim$31\,s) and far from the reference (FMD 507 and 727 vs.\ PMT's 159$\pm$8 in the same protocol). PMT matches within 1pp in aggregate (per genre this is partly a mixture cancellation, Appendix~\ref{app:persource}; the timing-axis advantage does \emph{not} cancel), and the advantage is not specific to conditioning: the closest prior system, the Anticipatory Music Transformer~\citep{thickstun2023anticipatory}, is an LLM over fine-timing performance events like PMT but is \emph{unconditional} and piano-only, so it cannot enter our text-conditioned controlled swap; in a matched unconditional protocol a 26M from-scratch PMT beats it (128M) and MIDI-LLM (1B) on FMD (222 vs.\ 518 and 274), a contextual comparison since, as with MIDI-LLM, AMT is evaluated away from its training distribution. The asymmetry is real: the published systems are evaluated \emph{away} from their Lakh/MidiCaps training distribution against a folk-heavy reference, so the decisive, distribution-invariant evidence is the \emph{bidirectional} test (\S Public-Benchmark Transfer), where MIDI-LLM emits $71\%$ chord-time on MidiCaps versus $72\%$ on our captions, tracking its training data rather than the caption (Figure~\ref{fig:pipeline_compare} shows this qualitatively).

\begin{figure*}[tp]
\centering
\includegraphics[width=\textwidth]{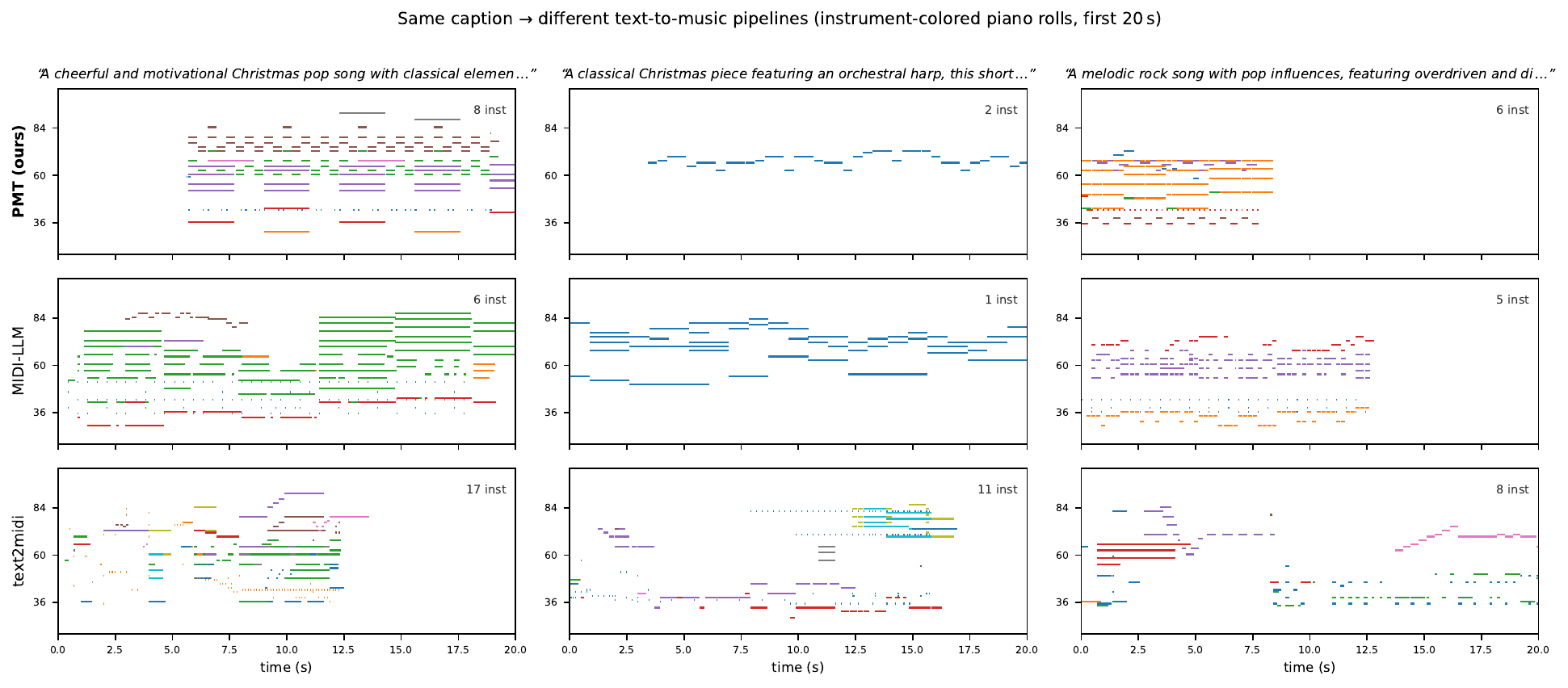}
\caption{\textbf{Full-pipeline qualitative comparison on MidiCaps.} The same caption (column title) drives three pipelines; instrument-colored piano rolls, first $20$\,s. PMT's instrument count \emph{tracks the caption}, rich for pop/rock prompts, sparse for the solo-harp prompt (center); MIDI-LLM collapses to a single track on the harp prompt yet thickens elsewhere, and text2midi consistently over-instruments ($11$--$17$ tracks), the imprinting Table~\ref{tab:baselines} quantifies as $66$--$72\%$ chord-time and FMD $507$--$727$ (vs.\ PMT's $36\%$/$159$).}
\label{fig:pipeline_compare}
\end{figure*}

\paragraph{Cross-modal (audio generators).}\label{apx:cross-modal-audio-generators} We render every system's output under one soundfont (audio models generate natively) and measure CLAP~\citep{wu2023large} caption--audio cosine on 200 captions against MusicGen~\citep{copet2023simple} and AudioLDM\,2~\citep{liu2024audioldm}. CLAP does not cleanly rank quality: caption-matched \emph{real} music scores lowest ($31.5$ vs.\ $35$--$50$ for generations), so CLAP rewards caption-keyword echoing over realism. Ranking: AudioLDM\,2 $50.2$, MIDILM $48.4$, MusicGen $47.5$, Amadeus $46.1$, PMT $42$ (mid-pack, consistent with weaker adherence); tellingly a \emph{symbolic} model (MIDILM) sits between the two audio generators despite their timbre advantage. We omit Fr\'echet Audio Distance: it would score FluidSynth-rendered symbolic output against natively-generated audio, conflating soundfont timbre with composition. PMT and audio generators serve different deliverables (an editable score vs.\ a fixed waveform), so the cross-modal number contextualizes rather than supersedes the symbolic result.

\paragraph{Zero-shot general LLMs.}\label{apx:zero-shot-general-llms} Prompting GPT-5.5 and Qwen3.5-4B to write ABC (Table~\ref{tab:midicaps}, bottom): GPT-5.5 leads key adherence ($.66$) by copying the stated key into the ABC header yet is playable for only $58\%$ of captions (the rest non-compiling, empty, or degenerate); Qwen3.5-4B compiles $100\%$ but only by writing trivial single-line tunes (poly $2.5$, one instrument). Neither delivers reliable \emph{and} rich generation, separating our contribution from ``just prompt a bigger model''; the vocabulary-extension and caption-masked SFT of \S Method are what convert an LLM into a dependable generator ($1.00$ validity, $10$+ polyphony).

\begin{figure*}[tp]
\centering
\includegraphics[width=\textwidth]{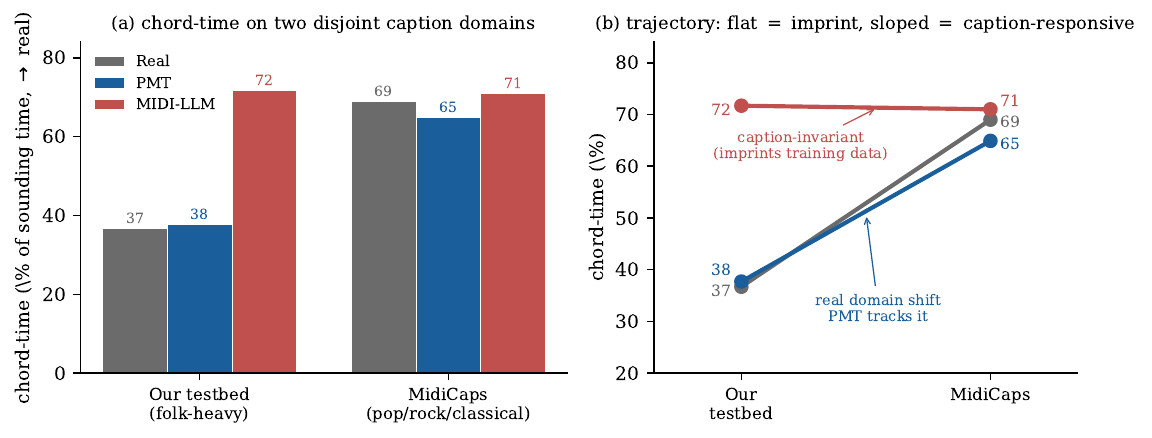}
\caption{\textbf{Imprinting diagnostic} (chord-time via \texttt{poly\_stats}, $n{=}200$ per domain). \emph{(a)}~Per-domain rates. \emph{(b)}~The same as trajectories: flat $=$ caption-invariant.}
\label{fig:imprint}
\end{figure*}

\paragraph{Audio-domain view of over-texturing.}\label{apx:audio-domain-view-of-over-texturing} Figure~\ref{fig:spectro} renders the same held-out caption (a solo-piano jazz ballad) through all three pipelines under one soundfont and shows the mel-spectrograms on a shared time axis. This is the audio-domain counterpart of the chord-time statistic, and it is a \emph{textural}, not a pitch-correctness, comparison: the real reference and PMT share an articulated character (discrete attacks, decay, inter-phrase silence, chord-time $56$ vs.\ $61\%$), whereas MIDI-LLM fills the low--mid register with a near-continuous energy wall ($80\%$ chord-time), the same over-texturing its Lakh-heavy training imprints (Figure~\ref{fig:imprint}). We show it as illustrative context, not evidence, the quantitative claims rest on the distance-to-reference tables.

\begin{figure}[tb]
\centering
\includegraphics[width=\columnwidth]{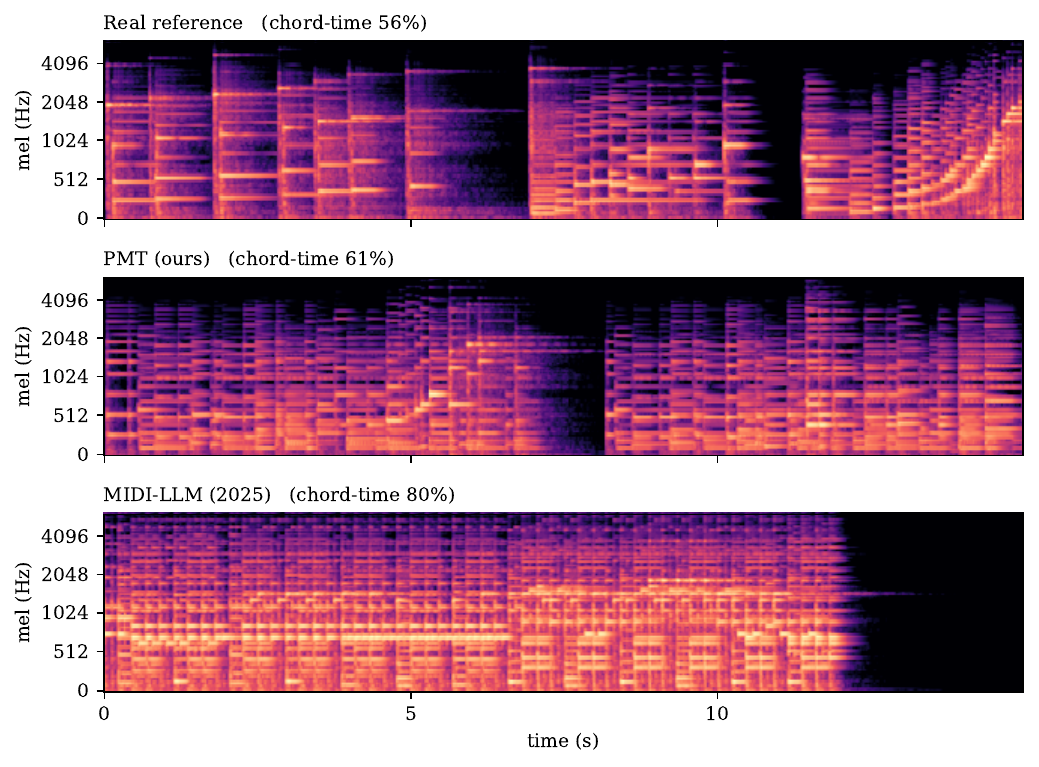}
\caption{\textbf{Audio-domain view (mel-spectrograms, one soundfont, shared time axis).} The same caption drives three pipelines; chord-time (paper metric) annotated per panel. Real and PMT share an articulated, silence-punctuated texture; MIDI-LLM emits a denser sustained energy wall, the audio-domain signature of the over-texturing quantified in Table~\ref{tab:baselines} and Figure~\ref{fig:imprint}. Illustrative context for a single caption, not a quantitative claim.}
\label{fig:spectro}
\end{figure}

\section{Additional Structural Metrics}
\label{app:structural}
\subsection{Efficiency and Ceilings}\label{apx:efficiency-and-ceilings}
\begin{table}[t]
\centering
\small
\begin{tabular}{lcccc}
\toprule
Repr. & Vocab & Tok/note$\downarrow$ & RT ceiling$\uparrow$ & Validity$\uparrow$ \\
\midrule
\textbf{PMT} & 609 & 4.03 & \textbf{39.1\%} & 1.00 \\
Beat-TSD & 475 & 4.79 & 35.5\% & 1.00 \\
REMI & 444 & 4.92 & 35.5\% & 1.00 \\
MIDI-Like & 562 & $\sim$4.0 & 35.7\% & 1.00 \\
PerTok & 473 & 6.8 & 35.0\% & 1.00 \\
ABC & (text) & \textbf{2.84} & 32.5\% & 1.00* \\
\bottomrule
\end{tabular}
\caption{Representation efficiency and expressiveness. \textbf{Vocab} $=$ music-vocabulary size; \textbf{Tok/note} $=$ tokens/note on generations (lower $=$ more compact); \textbf{RT ceiling} $=$ chord-time surviving a model-free encode--decode of real music (higher $=$ more expressive; raw reference 36.7\%); \textbf{Validity} $=$ fraction decoding to playable MIDI (*ABC: abc2midi compile rate).}
\label{tab:eff}
\end{table}

\begin{figure*}[tp]
\centering
\includegraphics[width=0.92\textwidth]{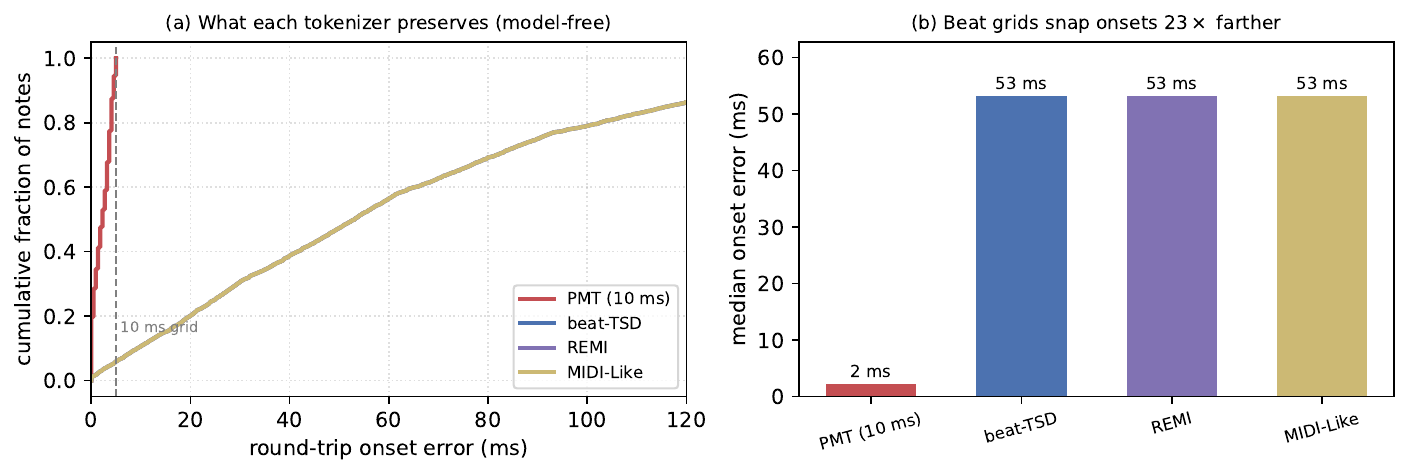}
\caption{\textbf{What each \emph{tokenizer} preserves, model-free.} Encode$\rightarrow$decode of $80$ real pieces per representation; no model is involved.}
\label{fig:tokenizer_compare}
\end{figure*}

Table~\ref{tab:eff} quantifies the trade-off. ABC is most compact ($2.84$ tokens/note) but compactness does not convert into quality; PMT pays a moderate token cost yet has the highest round-trip ceiling ($39.1\%$ vs.\ $35.5\%$ for beat grids) and, at $0.8$B, the highest attainment of it ($.92$). Inference wall-clock scales with tokens/note on a shared backbone, so PMT's compact $4.03$-token/note stream generates as fast as the leanest MIDI-event arm (MIDI-Like, ${\sim}4.0$) and faster than Beat-TSD ($1.19\times$), REMI ($1.22\times$), and PerTok ($1.69\times$) for the same music content (measured $23.8$ tok/s at 4B); PerTok's better onset timing thus also costs $1.7\times$ PMT's generation time. The mechanism is model-free: a round-trip preserves onsets to a $2.3$\,ms median under PMT but $53$\,ms under every beat grid (Figure~\ref{fig:tokenizer_compare}), timing lost \emph{before} any training. The lead holds at every scale (Table~\ref{tab:fid}: FMD $152$--$159$ for PMT vs.\ $272$--$286$ beat grids and $388$--$407$ ABC, beat-grid FMD flat from $0.8$B to $27$B); PerTok sits between ($185$--$188$), and a grouped-encoding control (Structured, FMD $287$) sits \emph{with} the beat grids, so the advantage is the $10$\,ms grid, not the flat layout. The metric decomposition (Table~\ref{tab:muspy}) localizes PMT's edge to timing (JSD$_{ioi}$ ${\approx}4\times$ lower than any beat grid) and duration while harmony ties. We surface the one axis that runs the other way: on symbolic \emph{caption alignment} (CLaMP-3, Table~\ref{tab:fid}) the beat grids slightly edge PMT ($.151$ vs.\ $.146$), consistent with PMT's weaker caption adherence. The advantage we claim is the marginal note and timing \emph{distribution}, not caption-following, and we hold both to the same standard: the CLaMP-3 gap ($.146$ vs.\ $.151$) is within noise, exactly as we treat similarly small differences elsewhere. A training-free round-trip on whole pieces further shows the advantage is overwhelmingly \emph{learnability}, not format ceiling: the model-free FMD gap is only ${\approx}8$ points, under $10\%$ of the ${\approx}120$-point trained gap, also ruling out truncation (full numbers in Appendix~\ref{app:extmetrics}).

\textbf{Robustness to piece complexity.} Public text-to-MIDI benchmarks skew short, so they cannot reveal whether a system degrades on longer, denser material. We build a \emph{tiered} benchmark: 240 real pieces stratified by PMT-token length into L1 ($<$1k), L2 (1--2k), L3 (2--4k), L4 (4--8k tokens), each caption paired with its source MIDI as a per-tier reference; PMT-4B (6.25M-corpus checkpoint) generates from all 240 (Table~\ref{tab:tiered}). Both PMT and MIDILM stay $100\%$ valid at every tier, but PMT is closer to the per-tier real distribution at \emph{every} level ($298/348/357/396$ vs.\ MIDILM's $422/378/377/436$), while MIDILM over-textures throughout (max-polyphony $15$--$30$, chord $64$--$73\%$ vs.\ a real ${\approx}6$/$37\%$; PMT holds $6$--$15$ and $26$--$38\%$). PMT's distributional advantage thus \emph{persists across complexity} and its FMD degrades gracefully with length rather than collapsing. (Per-tier FMD uses each tier's real subset as reference, a within-tier trend, not the MidiCaps-200 scale of Table~\ref{tab:midicaps}.)

\begin{table}[t]
\centering\small\setlength{\tabcolsep}{2.2pt}
\begin{tabular}{lcccccc}
\toprule
 & \multicolumn{2}{c}{FMD$\downarrow$} & \multicolumn{2}{c}{Max-poly ($\to$6)} & \multicolumn{2}{c}{Chord\% ($\to$37)} \\
\cmidrule(lr){2-3}\cmidrule(lr){4-5}\cmidrule(lr){6-7}
Tier & \textbf{PMT} & MIDILM & \textbf{PMT} & MIDILM & \textbf{PMT} & MIDILM \\
\midrule
L1 ($<$1k) & \textbf{298} & 422 & \phantom{0}6.2 & 28.4 & 25.6 & 64.3 \\
L2 (1--2k) & \textbf{348} & 378 & 14.6 & 15.2 & 37.6 & 67.3 \\
L3 (2--4k) & \textbf{357} & 377 & \phantom{0}6.8 & 16.6 & 36.1 & 68.1 \\
L4 (4--8k) & \textbf{396} & 436 & \phantom{0}6.7 & 30.5 & 37.8 & 73.0 \\
\bottomrule
\end{tabular}
\caption{\textbf{Complexity-tiered robustness} (240 real pieces, 60/tier; PMT-4B on the 6.25M corpus vs.\ MIDILM).}
\label{tab:tiered}
\end{table}

Beyond the timing/dynamics axes, we report a standard MGEval-style structural battery (Table~\ref{tab:struct}) on 0.8B generations. It reinforces the paper's central decomposition: PMT and the beat-grid arms are \emph{statistically indistinguishable} on every structural feature, note density, pitch range, used pitch-classes, mean note length, pitch-class entropy, because they encode the same note content; the representations diverge only on \emph{timing} (JSD$_{ioi}$, Table~\ref{tab:muspy}) and holistic distribution (FMD). One shared weakness is visible and honestly noted: all token models over-produce note density ($\sim$7.4/s vs.\ a 3.8/s reference), with ABC alone closer (4.4) but worst elsewhere.

\begin{table}[h]
\centering
\small
\begin{tabular}{lccccc}
\toprule
Arm (0.8B) & N/s & PRange & \#PC & NoteLen & PC-Ent \\
\midrule
Reference & 3.8 & 20 & 7.2 & .43 & 2.61 \\
\midrule
\textbf{PMT} & 7.5 & 29 & 7.9 & .44 & 2.64 \\
Beat-TSD & 7.4 & 29 & 7.9 & .44 & 2.63 \\
REMI & 6.7 & 29 & 8.0 & .46 & 2.68 \\
MIDI-Like & 7.1 & 28 & 7.9 & .47 & 2.65 \\
ABC & 4.4 & 26 & 7.8 & .39 & 2.63 \\
\bottomrule
\end{tabular}
\caption{Additional structural features (N/s note density, PRange pitch range, \#PC used pitch-classes, NoteLen mean note length in s, PC-Ent pitch-class entropy). PMT $\approx$ beat grids on \emph{every} structural axis, the representation advantage is specifically timing, not structure. All token models share a note-density overshoot.}
\label{tab:struct}
\end{table}

\subsection*{Comprehensive generative metrics: structure, diversity, novelty, harmony}\label{apx:comprehensive-generative-metrics-structure-diversity-novelty}
To evaluate as broadly as recent surveys recommend, we report four axes rarely included in text-to-MIDI papers (Table~\ref{tab:comprehensive}, 0.8B, all seeds): \textbf{CR} (compression ratio of the quantized note stream, a flat proxy for the COSIATEC compression used by text2midi, capturing long-term structure/repetition; closer to the real value is better), \textbf{Div} (intra-set diversity: mean pairwise pitch-class-histogram cosine distance among a system's generations; higher $=$ more varied output), \textbf{Nov} (novelty: mean nearest-neighbour distance from generations to the real set; higher $=$ less memorization), and \textbf{PCTM} (JSD of the $12{\times}12$ pitch-class transition matrix vs.\ real, a harmony/key-structure feature; lower is better). The picture is \emph{differentiated}: PMT is closest to real on structure (CR $3.25$ vs.\ $3.01$--$3.09$ for beat grids; real $3.52$) and leads on diversity ($.40$) and novelty ($.075$), whereas on harmonic transitions (PCTM) all representations are close and ABC/Beat-TSD are marginally nearer. This is exactly what the thesis predicts: PMT's advantage concentrates in timing, structure, and expressive variety, not harmony, which every representation encodes alike. (These agree with the note-density caveat of Table~\ref{tab:struct}: all token arms over-produce density; NPB is reported in the release but omitted here as high-variance.)

\begin{table}[t]
\centering\small
\begin{tabular}{lrrrr}
\toprule
Repr.\ (0.8B) & CR$_{\to3.52}$ & Div$\uparrow$ & Nov$\uparrow$ & PCTM$\downarrow$ \\
\midrule
\textbf{PMT}      & \textbf{3.25} & \textbf{.403} & \textbf{.075} & .085 \\
PerTok            & 3.27          & .383          & .067          & .074 \\
REMI              & 3.08          & .394          & .066          & .085 \\
Beat-TSD          & 3.09          & .396          & .070          & .068 \\
MIDI-Like         & 3.01          & .398 & -- & -- \\
ABC               & 3.36          & .355          & .067          & \textbf{.055} \\
\bottomrule
\end{tabular}
\caption{Comprehensive generative metrics at 0.8B. \textbf{CR}$=$compression ratio (structure; real $3.52$), \textbf{Div}$=$intra-set diversity$\uparrow$, \textbf{Nov}$=$novelty vs.\ real$\uparrow$, \textbf{PCTM}$=$pitch-class-transition JSD vs.\ real$\downarrow$. ``--'' marks axes not measured for that arm under this protocol. PMT leads structure, diversity and novelty; harmonic transitions are representation-agnostic.}
\label{tab:comprehensive}
\end{table}

\section{Extended Representation Grid}
\label{app:extgrid}

\paragraph{The controlled grid at every scale.}\label{apx:the-controlled-grid-at-every-scale} The main text shows the two extremes; the 2B, 4B and 9B blocks are here, on the same four metrics.

\begin{table}[t]
\centering
\small
\setlength{\tabcolsep}{3.5pt}
\begin{tabular}{llcccc}
\toprule
Size & Repr. & FMD$\downarrow$ & JSD$_{ioi}\downarrow$ & Chord\%$\rightarrow$ & Att.$\uparrow$ \\
\midrule
\multicolumn{6}{l}{\emph{Reference: FMD 0, JSD$_{ioi}$ 0, chord-time 36.7\%}}\\
\midrule
0.8B & \textbf{PMT} & \textbf{159$\pm$8} & \textbf{0.032$\pm$0.002} & \textbf{36.2$\pm$5.6} & \textbf{0.92} \\
0.8B & Beat-TSD & 286$\pm$1 & 0.152$\pm$0.008 & 30.8$\pm$0.9 & 0.87 \\
0.8B & REMI & 285$\pm$1 & 0.152$\pm$0.007 & 30.2$\pm$1.7 & 0.85 \\
0.8B & MIDI-Like & 285$\pm$2 & 0.160$\pm$0.006 & 26.6$\pm$2.0 & 0.75 \\
0.8B & PerTok & 188$\pm$7 & \textbf{0.011$\pm$0.001} & 30.2$\pm$2.6 & 0.86 \\
0.8B & Structured & 287 & 0.121 & 30.3 & 0.85 \\
0.8B & ABC & 406 & 0.079 & 16.2 & 0.50 \\
\midrule
2B & \textbf{PMT} & \textbf{152$\pm$9} & \textbf{0.037$\pm$0.003} & \textbf{32.6$\pm$2.7} & \textbf{0.83} \\
2B & Beat-TSD & 280$\pm$3 & 0.153$\pm$0.006 & 26.7$\pm$1.6 & 0.75 \\
2B & REMI & 281$\pm$2 & 0.161$\pm$0.001 & 30.5$\pm$0.4 & 0.86 \\
2B & MIDI-Like & 281$\pm$2 & 0.152$\pm$0.010 & 26.5$\pm$1.7 & 0.74 \\
2B & PerTok & 185 & \textbf{0.022} & 27.1$\pm$0.5 & 0.77 \\
2B & Structured & 287 & 0.112 & 27.4 & 0.77 \\
2B & ABC & 407 & 0.076 & 15.4 & 0.47 \\
\midrule
4B & \textbf{PMT} & \textbf{158$\pm$4} & \textbf{0.035$\pm$0.009} & \textbf{32.2$\pm$2.3} & \textbf{0.82} \\
4B & Beat-TSD & 271 & 0.149 & 24.5 & 0.69 \\
4B & REMI & 281$\pm$3 & 0.154$\pm$0.001 & 26.1$\pm$1.3 & 0.74 \\
4B & MIDI-Like & 280$\pm$2 & 0.150$\pm$0.006 & 24.9$\pm$0.9 & 0.70 \\
4B & Structured & 302 & 0.116 & 32.7 & 0.92 \\
4B & ABC & 394 & 0.085 & 16.5 & 0.51 \\
\midrule
9B & ABC & 388 & 0.078 & 15.1 & 0.46 \\
\midrule
27B & \textbf{PMT} & \textbf{156} & \textbf{0.048} & \textbf{35.8$\pm$0.5} & \textbf{0.92} \\
27B & REMI & 272 & 0.144$\pm$0.005 & 28.8$\pm$0.4 & 0.81 \\
27B & MIDI-Like & 285 & 0.148 & 27.9 & 0.78 \\
\bottomrule
\end{tabular}
\caption{Controlled comparison, 100 frozen-test captions per cell, mean$\pm$std over seeds (single value $=$ one seed). Identical backbone, data, budget and decoding; validity $0.99$--$1.00$.}
\label{tab:maingrid}
\end{table}

\begin{figure}[tb]
\centering
\includegraphics[width=\columnwidth]{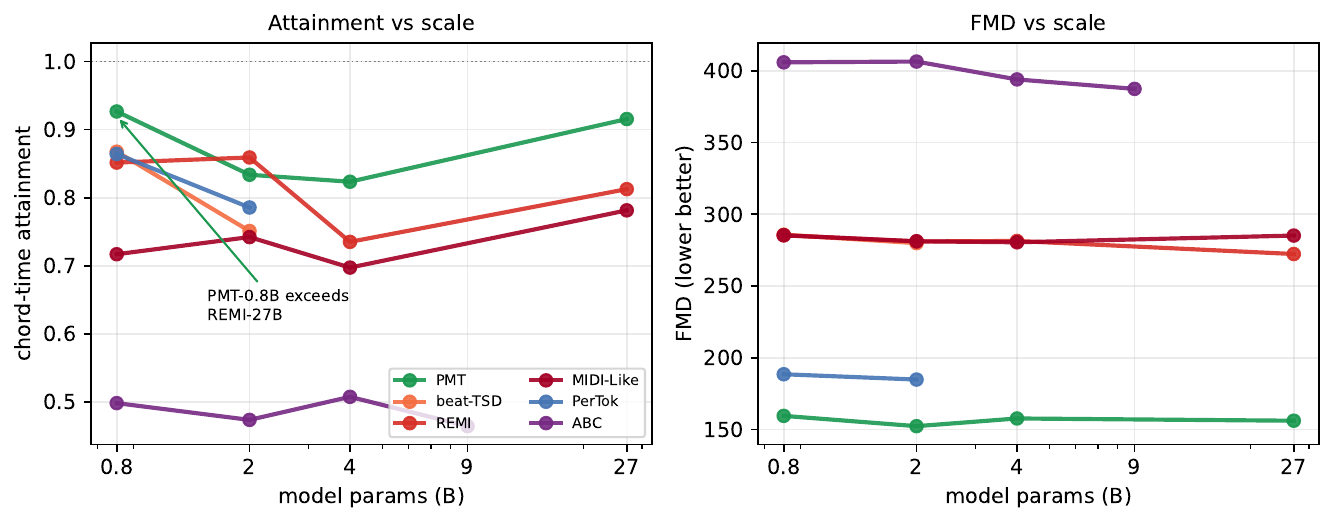}
\caption{Representation Pareto-dominates scale. Attainment (left, higher better) and FMD (right, lower better) vs.\ backbone size per representation.}
\label{fig:scaling}
\end{figure}

\paragraph{Reading the controlled grid in full.}\label{apx:reading-the-controlled-grid-in-full} The ordering, the PerTok comparison, and what separates the two performance-resolution arms.

Table~\ref{tab:main} orders the representations PMT $>$ Beat-TSD $\approx$ PerTok $\approx$ REMI $>$ MIDI-Like $\gg$ ABC. Performance resolution is not unique to PMT: PerTok, another performance-resolution encoding, joins PMT in the low-FMD/low-JSD$_{ioi}$ corner (Figure~\ref{fig:landscape}) and even edges it on onset divergence ($.011$ vs.\ $.032$), corroborating that performance resolution ($10$\,ms timing \emph{and} per-note velocity), not the beat grid, is the binding distinction. PMT's edge over PerTok is \emph{completeness}, not timing: PerTok spends $6.8$ tokens/note against PMT's $4.0$ (so it truncates more at the shared $1024$-token block) and carries no native velocity, which is why it reaches the best onset divergence yet a \emph{higher} FMD ($188$ vs.\ $159$); we thus credit the finding to performance-resolution encodings as a class. Crucially the gap does not close with scale: REMI's FMD plateaus at $272$ and JSD$_{ioi}$ at $.144$ even at 27B, far worse than PMT's \emph{0.8B} ($159$, $.032$), and MIDI-Like behaves identically, so \textbf{PMT-0.8B beats REMI-27B on both FMD ($159$ vs.\ $272$, $1.7\times$) and onset timing ($.032$ vs.\ $.144$, $4.5\times$)}, gaps a $34\times$ parameter increase does not close. ABC is the outlier, flat at $46$--$51\%$ attainment from 0.8B to 9B: to our knowledge the first \emph{controlled} evidence that, whatever ABC's merits for understanding and folk-tune modeling, its lack of native performance fields caps high-fidelity multi-instrument generation (model-free ceiling only $32.5\%$, Table~\ref{tab:eff}, before any training).

\paragraph{Scaling behaviour in full.}\label{apx:scaling-behaviour-in-full} The Pareto reading, the flat-slope fits, and what scale does buy per representation.

The clean reading of scale is cross-representation (Figure~\ref{fig:scaling}): PMT's FMD band ($152$--$159$ across $0.8$B--$27$B) lies entirely below every beat-grid arm's ($272$--$286$), so the efficiency claim is a Pareto statement, not a fragile single point. Fitting FMD against $\log_{10}$-parameters, every arm's slope is within noise of flat, leaving a near-constant ${\approx}125$-point PMT-vs-beat-grid offset that scale does not erode. We do not extrapolate a crossover from a fit this flat over $1.5$ decades; the defensible statement is that there is \emph{no} closing trend within range, so representation, not parameter count, is the binding variable at every scale we measure. Two conclusions follow: on distributional fidelity the representation outweighs the full $34\times$ parameter range we test, and how much scale buys \emph{depends on} the representation (ABC's chord-time stays flat at $15$--$16\%$ from $0.8$B to $9$B while PMT reaches reference parity by $27$B). Data scales where parameters do not: growing the corpus $1\text{k}\!\to\!12.7\text{k}$ improves FMD monotonically ($196\!\to\!167$), while a $35$B-A3B MoE point preserves timing (JSD$_{ioi}$ $.034$) but underfits the minority chordal mode ($.52$ attainment). Held-out NLL improves with scale within each arm but is incommensurable across vocabularies, so we rest scaling on the generation surface.
\begin{figure}[tb]
\centering
\includegraphics[width=0.98\columnwidth]{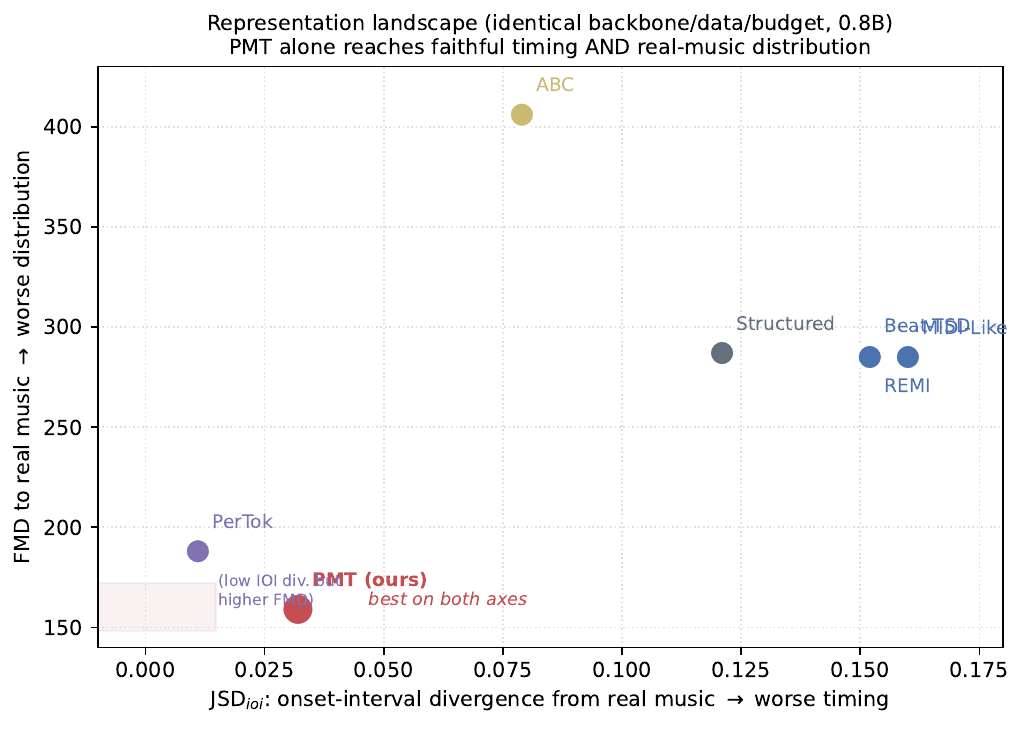}
\caption{\textbf{Representation landscape.} The seven controlled-grid representations at 0.8B under an identical backbone, data and budget.}
\label{fig:landscape}
\end{figure}

The main Table~\ref{tab:main} is drawn from a denser per-scale grid (all seeds, plus a grouped-encoding Structured arm and per-scale metric breakdowns) released with the harness; every cell uses the identical backbone, data, budget, and decoding.

\section{Training Details}
\label{app:training}

\paragraph{Implementation details in full.}\label{apx:implementation-details-in-full} Exact model coverage per representation, per-scale device splits, and the FMD-protocol caveats.

\textbf{Implementation details.} All backbones are from the \textbf{Qwen3.5} family (\texttt{model\_type: qwen3\_5}; 0.8B/2B/4B/27B and a 35B-A3B MoE, exact HuggingFace identifiers in the released harness, \emph{not} the same-sized Qwen3.6 variants), fine-tuned in a single SFT stage (0.8B/2B/4B across all representations; 27B for PMT, REMI, and MIDI-Like; a 9B point for ABC and a 35B-A3B MoE point for PMT). Music vocabularies extend the tokenizer with mean-initialized embeddings; captions are truncated to 640 tokens and masked with $-100$. Training: 10k steps ($\approx$1.9 epochs), effective batch 16, lr $8{\times}10^{-5}$ cosine, 3\% warmup, bf16, gradient checkpointing, block 1024 (per-scale device/accumulation splits, DeepSpeed ZeRO-3 beyond 4B, vocabulary sizes, and compute costs in this appendix). Evaluation: temperature 0.95, top-k 60, up to 900 music tokens, music-range-constrained logits (abc2midi compilation for ABC), $n{=}100$ frozen-test captions per cell, up to 3 seeds. FMD values are comparable only within one protocol: the \emph{same} PMT-0.8B model scores $159$ against the 500-piece stratified reference and $114$ against the 475-piece deduplicated reference, and the PMT-vs-beat-grid \emph{ratio} ranges $1.6$--$2.8\times$ across protocols, so the ``${\approx}2\times$'' headline is a representative midpoint and we report the ratio per protocol throughout. A real-vs-real FMD floor (two disjoint real subsets) is $43$ at $n{=}100$ and $23$ at $n{=}250$, so at $n{=}100$ absolute FMD carries ${\sim}40$-scale sampling noise; the PMT-vs-beat-grid gap ($113$: $159$ vs.\ $272$) exceeds this floor, so the arm \emph{ordering} is robust even where absolute magnitudes are noisy. We run no per-comparison significance tests and coverage thins at scale (27B two-seed, 35B-A3B single-seed): the \emph{large} PMT-vs-beat-grid gaps on FMD and JSD$_{ioi}$ far exceed the seed spread, but finer orderings should be read as trends (this appendix).
\paragraph{Per-scale sharding and cost.}\label{apx:per-scale-sharding-and-cost} All arms keep effective batch $16$ by splitting per-device batch and gradient accumulation by scale: batch $16{\times}1$ A800 at 0.8B, $8{\times}2$ at 2B, $4{\times}4$ at 4B, and $1{\times}8$ with gradient accumulation $2$ and DeepSpeed ZeRO-3 sharding at 9B/27B. Because the Qwen3.5 family shares its tokenizer, one encoded dataset serves every scale. One 0.8B cell costs $\approx$3.1 A800-hours; the 0.8B--4B grid $\approx$70 A800-hours; each 27B cell adds $\approx$36. The extended music vocabularies are PMT 609, Beat-TSD 475, REMI 444, MIDI-Like 562, and PerTok 473 symbols, all mean-initialized.

\paragraph{Trainability check.}\label{apx:trainability-check} Before the full runs, a 64-example overfit drove training loss to $0.005$ with 6/6 valid generations on the sanity suite, confirming the vocabulary-extension and caption-masking pipeline learns the music tokens.

\paragraph{Statistical rigor.}\label{apx:statistical-rigor} For the headline comparisons we report bootstrap $95\%$ confidence intervals ($B{=}1000$, resampling test pieces). On JSD$_{ioi}$ (Table~\ref{tab:controls}): PMT $[.027,.051]$ is separated from Beat-TSD $[.144,.184]$ and REMI $[.135,.165]$ but overlaps PerTok $[.011,.035]$, so we assert the PMT-vs-beat-grid ordering and explicitly \emph{decline} to claim PMT $>$ PerTok. On FMD the same resampling gives PMT $[144,204]$ non-overlapping with Beat-TSD $[268,339]$, so the headline FMD gap is not an $n{=}100$ sampling artifact. Beyond this we do not run per-comparison significance tests, and seed coverage thins at scale (27B is two-seed, 35B-A3B single-seed, some cells single-seed). The large PMT-vs-beat-grid gaps on FMD far exceed the seed spread we observe, but finer orderings, cross-scale monotonicity and the PMT-vs-PerTok ordering, sit within noise and are read as trends. The PMT-0.8B-vs-REMI-27B FMD gap ($159$ vs.\ $272$) is by contrast over $10\times$ the per-seed std, a magnitude difference we state as fact, not a trend.

\section{Limitations Details}
\label{app:limitations}

\paragraph{Limitations in full.}\label{apx:limitations-in-full} Each of the five limitations with its supporting numbers inline.

\subsection{Limitations}\label{apx:limitations}
Five limitations bound our claims (supporting numbers in this appendix). \emph{(1) Near-duplicate splits}: the folk source contains multiple community ``settings'' of the same tune and our split is setting-level, so near-duplicate melodies can straddle train/test, inflating held-out likelihood for very large models; we exclude $>$4B points from NLL analyses and will release a tune-level deduplicated split. Generation-side metrics survive a title-level dedup control, so Tables~\ref{tab:main}--\ref{tab:muspy} are unaffected. \emph{(2) No human listening study}: our evidence for the performance-timing thesis is objective (JSD$_{dur}$, groove, FMD) and our automatic MLLM forced-choice proxy is inconclusive; a \emph{human} pairwise preference study on rendered audio is the natural next validation, left to future work. \emph{(3) Seed coverage at scale}: 27B rows are two-seed and 35B-A3B is single-seed; smaller cells use up to three seeds. \emph{(4) Caption adherence}: PMT's key adherence is weak (CK .10 vs.\ MIDI-LLM's .32; Table~\ref{tab:midicaps}), but this is an absolute-key \emph{conditioning} gap, not tonal incompetence: PMT recovers major/minor \emph{mode} above chance and improves with data, and a lightweight decode-time attribute-conditioning constraint (restricting \texttt{PROG} tokens to the caption's requested instruments and down-weighting out-of-key \texttt{PITCH} tokens) more than doubles \emph{both} instrument-F1 ($.28\!\to\!.60$) and Correct-Key ($.16\!\to\!.35$) on $80$ captions, at unchanged validity ($.99\!\to\!1.00$) and no distributional cost (FMD $266\!\to\!271$, within the $n{=}80$ floor), with no retraining, so explicit attribute-conditioning is a \emph{demonstrated}, not merely hypothesized, direction (this appendix). Our captions are also machine-generated by a single captioner (Qwen3-Omni-30B) and audited for prompt \emph{leakage} but not factual \emph{accuracy}; a human-rated accuracy/hallucination subset is the clearest dataset-side next step. \emph{(5) Corpus scale in the controlled study}: the \emph{experiments} use the 86.6k four-modality set, trading raw count for per-case alignment and caption supervision. As a first step to scale, a PMT-4B model trained on the released 6.25M corpus already improves MidiCaps caption-key adherence over the 86.6k model and, on held-out in-distribution data, cuts FMD $2.0\times$ over its first pass ($291.7\!\to\!144.8$ on the tracking reference, $123.5$ on a source-stratified one; a steep early gain then a plateau, not monotone; this appendix). We report that one-epoch model and do not claim a converged one: the full multi-scale grid at 6.25M remains future work, and no claim in this paper depends on it.
This appendix expands three of the five limitations of \S Limitations with their supporting numbers.

\paragraph{Near-duplicate split control.}\label{apx:near-duplicate-split-control} We verified generation-side metrics survive title-level deduplication: a source-stratified reference rebuilt after title-level dedup ($n{=}475$) has chord-time $36.1\%$ vs.\ $36.2\%$ for the full reference, and PMT-0.8B's FMD moves $1\%$ ($112.7$ vs.\ $113.8$, size-matched control), so Tables~\ref{tab:main}--\ref{tab:muspy} are unaffected by the near-duplicate folk settings. The setting-level split also inflates held-out likelihood for very large models (a 27B model reaches implausibly low NLL by recalling variants), which is why we exclude $>$4B points from NLL analyses.

\paragraph{Key-adherence breakdown.}\label{apx:key-adherence-breakdown} A key-offset analysis restricted to captions that explicitly state a key (where the exact-key metric is well-defined, so figures run slightly below the full-set CK $.10$) locates the adherence gap: PMT reproduces the caption's \emph{exact} stated tonic only $12$--$26\%$ of the time (86.6k--6.25M models) versus $49$--$62\%$ for MIDILM/Amadeus, yet recovers the major/minor \emph{mode} above chance ($59\%$) and improves monotonically with data (exact-tonic $12\%\!\to\!26\%$, keyed-subset CK $.08\!\to\!.15$). The gap is thus absolute-key \emph{conditioning}, not tonal incompetence, which is why a lightweight decode-time constraint recovers it without retraining. Extending the in-key pilot to a full \emph{attribute-conditioning} decode on the 6.25M model (checkpoint at step $116$k), restricting \texttt{PROG} tokens to the caption's requested instruments and softly suppressing out-of-key pitches, more than doubles \emph{both} adherence axes \emph{simultaneously} ($n{=}80$: instrument-F1 $.28\!\to\!.60$ with the instrument constraint, Correct-Key $.16\!\to\!.35$ with the key constraint, validity $.99\!\to\!1.00$), at \emph{no} distributional cost (FMD $266\!\to\!271$, within the $n{=}80$ sampling floor). So caption adherence is a decode-time conditioning gap, closable without sacrificing fidelity, and not a deficit of the representation (Figure~\ref{fig:conditioning}).

\begin{figure}[tb]
\centering
\includegraphics[width=\columnwidth]{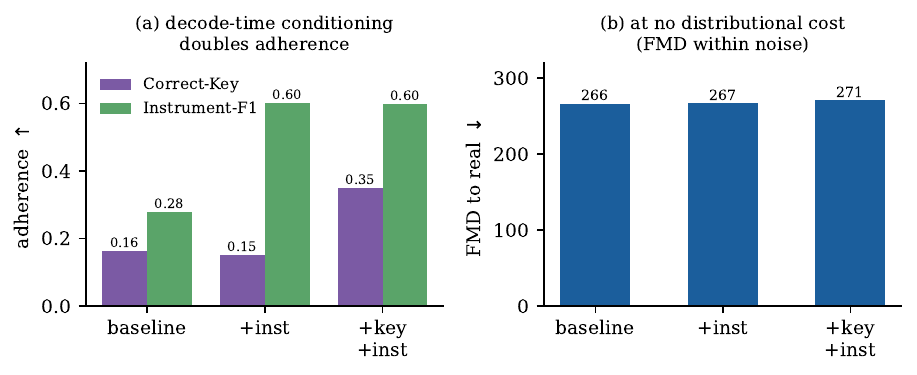}
\caption{\textbf{Caption adherence is a decode-time conditioning gap} ($n{=}80$, 6.25M model at step 116k).}
\label{fig:conditioning}
\end{figure} The captioner also sees score-derived metadata (instrument, key, tempo), so some attributes may be metadata paraphrases rather than heard, and the self-built portion inherits basic-pitch transcription noise.

\paragraph{Embedder robustness: the same gap under CLaMP-3.}\label{apx:embedder-robustness-the-same-gap-under} The headline distributional numbers are Fr\'echet distances in one embedding space, so a fair objection is that they measure CLaMP-2 rather than the music. We therefore recompute the Fr\'echet distance over \textbf{CLaMP-3} music embeddings~\citep{wu2025clamp}: a different model family, trained with text/audio/symbolic contrastive alignment rather than CLaMP-2's symbolic objective, and reading MIDI through the MTF serialization, which is a lossless tick-level dump of the MIDI event stream (so the swap changes the embedder, not the timing information reaching it). Everything else is held fixed: the same $100$ generations per arm from the same 0.8B controlled-swap cells, the same $500$-piece reference, the same Gaussian estimator. To confirm the protocols are comparable we also recompute CLaMP-2 on those exact directories and reproduce the main-table cells to within $0.1$ (PMT $161.6$ vs.\ $161.55$; Beat-TSD $285.2$; REMI $285.8$). Table~\ref{tab:clamp3} reports both spaces.

\begin{table}[t]
\centering\small
\begin{tabular}{lcc}
\toprule
0.8B arm & CLaMP-2 FMD$\downarrow$ & CLaMP-3 FD$\downarrow$ \\
\midrule
PMT (perf.-timed)     & \textbf{161.6} & \textbf{110.4} \\
PerTok (perf.-timed)  & 195.1 & 141.0 \\
Beat-TSD                    & 285.2 & 208.3 \\
REMI                        & 285.8 & 206.3 \\
MIDI-Like                   & 288.8 & 216.6 \\
\bottomrule
\end{tabular}
\caption{\textbf{The gap survives an embedder swap.} Same generations and reference per 0.8B arm; only the embedding model changes.}
\label{tab:clamp3}
\end{table}

\begin{figure}[tb]
\centering
\includegraphics[width=0.86\linewidth]{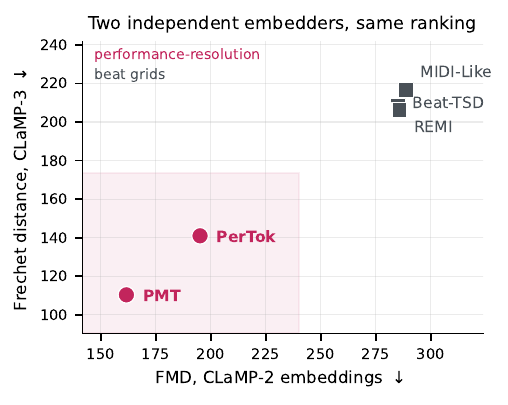}
\caption{\textbf{The gap is not an artifact of one embedding space.} Each 0.8B arm scored twice on the same generations against the same $500$-piece reference: $x$ over CLaMP-2, $y$ over CLaMP-3.}
\label{fig:embedder}
\end{figure}

\noindent Both the ordering and the magnitude survive (Figure~\ref{fig:embedder}). Under CLaMP-3 the three beat grids average $1.9\times$ PMT's distance, slightly \emph{more} than the $1.8\times$ under CLaMP-2, and the class structure is identical: both performance-resolution encodings rank ahead of all three beat grids, with PerTok intermediate. The beat grids permute among themselves (REMI and Beat-TSD swap), but they lie within $1\%$ of each other in both spaces, which is consistent with our position that we do not rank beat grids against each other. We conclude that the representation effect is a property of the generated music, not of one embedding space.

\paragraph{Reference robustness: the same gap against a source-stratified reference.}\label{apx:reference-robustness-the-same-gap-against} The second thing a single number can hide is the choice of reference. Our fixed $500$-piece reference comes from the folk source, which dominates the frozen split but is monophonic: measured with the same extractor as every other row, it has max-polyphony $2.0$, chord-time $9.0\%$ and $1.00$ instruments per piece, against $6.1$, $35.0\%$ and $2.20$ for a reference sampled proportionally across all eight sources of that split. The captions are stratified either way, so this is purely a reference swap and needs no regeneration. It moves the absolute scale and, importantly, moves it \emph{against} the reported numbers (Table~\ref{tab:refswap}).

\begin{table}[t]
\centering\small
\begin{tabular}{lcc}
\toprule
0.8B arm & folk reference & stratified reference \\
\midrule
PMT       & \textbf{161.6} & \textbf{110.5} \\
PerTok    & 195.1 & 141.0 \\
REMI      & 285.8 & 266.3 \\
Beat-TSD  & 285.2 & 270.2 \\
MIDI-Like & 288.8 & 275.7 \\
\bottomrule
\end{tabular}
\caption{\textbf{The gap survives a reference swap.} Same generations per 0.8B arm; folk vs.\ source-stratified $500$-piece reference.}
\label{tab:refswap}
\end{table}

\noindent The ordering is identical, and PMT moves much closer to the stratified reference ($161.6\!\to\!110.5$) while the beat grids barely move ($285\!\to\!266$--$276$), so the class separation grows from $1.8\times$ to $2.4\times$. A second stratified sample of $475$ pieces reproduces this ($113.8$ vs.\ $271$--$281$). The tables therefore report the \emph{conservative} reference: the effect we claim is smaller than the effect a corpus-representative reference would show. This also rules out the mirror-image objection, that PMT wins by matching a sparse single-instrument reference; if that were the mechanism, PMT would degrade against the richer reference rather than improve.

\paragraph{6.25M-corpus scaling trajectory.}\label{apx:6-25m-corpus-scaling-trajectory} A PMT-4B model trained on the 6.25M corpus (built by applying the Qwen3-Omni pipeline to the deduplicated raw union, filtered against the frozen test split and the C1--C5 caption gates) improves MidiCaps caption-key adherence over the 86.6k model even at a \emph{very} early checkpoint. We train the model for one pass over the corpus, ${\approx}166$k optimizer steps on $2\times8$ A800s (the model trains on the ${\approx}5.3$M pieces that fit the 8192-token block, at 32 pieces/step), and every 6.25M number in this paper is that one-epoch model; the checkpoints below are its intermediate steps. At 0.1 epoch (step 16.6k): key adherence including relative/parallel keys (CKD) $.17\!\to\!.27$, exact CK $.10\!\to\!.11$, validity $1.00$. At 0.4 epoch (step 66.4k): CKD $.17\!\to\!.30$, CK $.10\!\to\!.14$. At the full epoch (step 166k), the model we report: CKD $.17\!\to\!.32$, CK $.10\!\to\!.17$, validity $1.00$, and same-protocol FMD $211$ against $314$--$421$ for the published systems. On \emph{held-out in-distribution} data ($200$ frozen-test captions against a self-built reference of the same distribution) FMD drops steeply and then flattens with fluctuation over the first epoch: $291.7$ (0.1\,ep), $168.0$ (0.4), $163.4$ (0.5), $194.1$ (0.7), $140.7$ (0.9), $144.8$ (1.0). The first-epoch reduction is $2.0\times$ ($291.7\!\to\!144.8$), so the corpus is learnable and the representation transfers at this scale. Two honest qualifications. The path is \emph{not} monotone: the 0.7-epoch checkpoint regresses by $31$ points before 0.9 recovers. And the last three points ($163.4$, $140.7$, $144.8$) sit within a ${\approx}20$-point band, so what the epoch shows is a steep early gain followed by a plateau, not steady improvement to the end; the $0.9$ and $1.0$ checkpoints differ by $4$ points, which we read as noise. Each point is a single checkpoint and seed at $n{=}200$, so we assert the direction and the plateau, not the ordering of adjacent points.

One caveat bounds how far this trajectory can be read: the $200$ tracking captions and their reference come from the self-built source, the narrowest slice of the corpus (median $12$ distinct pitches and $164$ notes, against $25$ and $254$ for a corpus-wide sample), so the curve measures distance to one source rather than to the corpus mixture. We therefore also score against a source-stratified reference, on which the 1-epoch model reaches $123.5$.

\textbf{What we report: one epoch.} All 6.25M numbers in this paper are the \emph{one-epoch} model, and we fix them there deliberately rather than training longer. The claim this corpus has to support is that it trains, and one full pass settles that: FMD falls $2.0\times$ on the tracking reference ($291.7\!\to\!144.8$) and reaches $123.5$ against the source-stratified reference, caption-key adherence nearly doubles over the 86.6k model (CKD $.17\!\to\!.32$, CK $.10\!\to\!.17$), and on the MidiCaps public test its FMD of $211$ stays clear of every published system ($314$--$421$; Table~\ref{tab:midicaps}). None of the paper's claims scale with further passes: the headline representation result is the 86.6k controlled grid, and the corpus and harness are contributions independent of how long any one model trains. A converged 6.25M model, and the multi-scale grid at that corpus size, are future work; the one-epoch result is what we have measured end to end, and it is what we report.

\section{Scalable Audio-to-Symbolic Data Curation}
\label{app:curation}

\paragraph{Corpus construction in full.}\label{apx:corpus-construction-in-full} The four-modality corpus and the 6.25M scaled corpus, with per-source provenance.

\textbf{Dataset.} We build and release a four-modality corpus of 86{,}598 real-music pieces, each aligned across an English caption, sliced MIDI, ABC, and rendered audio, drawn from four sources (The Session folk, GiantMIDI classical piano~\citep{kong2022giantmidi}, a genre-balanced Lakh~\citep{raffel2016learning} subset, and jazz; per-source counts, caption pipeline, and source texture profiles in this appendix). Splits are frozen and source-stratified (84{,}576 train / 2{,}022 test, seed 42), and PMT-encoded lengths are heavy-tailed (median 835 tokens), motivating the 1024-token training block that covers 58\% of pieces unclipped. Table~\ref{tab:datasets} places the corpus among existing symbolic collections: the large raw-MIDI datasets supply scale but no language supervision, multi-track audio--MIDI sets such as Slakh~\citep{manilow2019cutting} pair synthesized audio but no captions, and the one captioned dataset (MidiCaps) is single-modality, whereas ours is the only corpus aligning \emph{four} modalities per piece. Because the largest raw collections (GigaMIDI, Aria-MIDI, Discover-MIDI) ship un-captioned, we apply the same pipeline to their deduplicated union and build a \emph{scaled} captioned corpus of \textbf{6.25M} MIDI--caption pairs (Table~\ref{tab:datasets}, ``Ours (scaled)''), the largest captioned symbolic dataset, $37\times$ MidiCaps and $6.5\times$ MetaScore, deduplicated against the frozen test split. Redistribution follows each source's terms (captions, derived features, and source pointers where raw redistribution is precluded). The controlled experiments below use the aligned 86.6k set for its per-case four-modality alignment; this appendix details the scaling pipeline, its caption quality-control filter, and why we do not draw on the task-mixed MusicPile~\citep{yuan2024chatmusician} instruction corpus.
\begin{table*}[t]
\centering\small
\setlength{\tabcolsep}{4.5pt}
\begin{tabular}{lrccccc}
\toprule
Dataset & \#Pieces & T & M & B & A & Multi-src \\
\midrule
Lakh~\citep{raffel2016learning}          & 174.5k & & \checkmark & & & \\
MidiCaps~\citep{melechovsky2024midicaps}    & 168.4k & \checkmark & \checkmark & & & \\
MetaScore~\citep{xu2024generating}$^{\dagger}$ & 326k & \checkmark & \checkmark & & & \checkmark \\
MetaMIDI~\citep{ens2021building} & 436.6k & & \checkmark & & & \checkmark \\
Aria-MIDI~\citep{bradshaw2025aria}       & 1.19M  & & \checkmark & & & \\
GigaMIDI~\citep{lee2025gigamidi}    & 1.43M  & & \checkmark & & & \checkmark \\
\midrule
\textbf{Ours (aligned)} & 86.6k & \checkmark & \checkmark & \checkmark & \checkmark & \checkmark \\
\textbf{Ours (scaled)} & \textbf{6.25M} & \checkmark & \checkmark & & & \checkmark \\
\multicolumn{7}{l}{\quad\emph{scaled-corpus sources (all captioned by our pipeline):}}\\
\quad Discover-MIDI$^{\ddagger}$          & 3.56M & & & & & \\
\quad GigaMIDI$+$Aria~\citep{lee2025gigamidi,bradshaw2025aria} & 1.88M & & & & & \\
\quad Purchased packs                     & \phantom{0}623k & & & & & \\
\quad Self-built (audio$\to$MIDI)$^{\S}$   & \phantom{0}192k & & & & & \\
\bottomrule
\end{tabular}
\caption{Symbolic music corpora. \textbf{T}/\textbf{M}/\textbf{B}/\textbf{A} mark the modalities each corpus provides \emph{per case} (text/MIDI/ABC/audio); \textbf{Cap.} marks text captions.}
\label{tab:datasets}
\end{table*}

The 86.6k four-modality corpus of \S Experiments is deliberately small, curated for per-piece alignment across caption, MIDI, ABC, and audio. To scale caption--music supervision toward the volume that raw collections offer but ship \emph{without} language (GigaMIDI 1.43M, Aria-MIDI 1.19M; \S Limitations), we run a separate, continuously-operating pipeline that turns \emph{raw audio}, for which no score exists, into curated symbolic training pairs. This appendix documents that pipeline and the 6.25M-pair captioned corpus it produces.

\paragraph{Aligned 86.6k corpus.}\label{apx:aligned-86-6k-corpus} The four sources of the controlled corpus are The Session folk archive ($52{,}794$, CC BY-SA, native ABC), GiantMIDI classical piano~\citep{kong2022giantmidi} ($10{,}847$), a genre-balanced multi-instrument subset derived from Lakh~\citep{raffel2016learning} ($22{,}347$), and jazz ($610$). Each piece carries an English caption generated by Qwen3-Omni-30B from deterministically rendered audio (real scraped titles prepended; invented names forbidden; mean $2{,}277$ characters; $0\%$ prompt leakage on a $500$-sample audit), sliced MIDI, ABC (native or converted, flagged), and audio. Source texture profiles span both regimes: GiantMIDI $100\%$ chordal (mean max-polyphony $11.8$), the balanced subset $98\%$ chordal ($88\%$ multi-instrument), and folk $15\%$ chordal (single-line). PMT-encoded lengths are heavy-tailed (median $835$ tokens, mean $1{,}442$, p90 $3{,}174$), so the $1024$-token block covers $58\%$ of pieces unclipped. We do not draw on MusicPile~\citep{yuan2024chatmusician}, the other large LLM-music corpus, despite its scale ($5.17$M samples): it is a task-mixed \emph{instruction}-tuning set (its music portion, ${\sim}0.9$M ABC score-chat plus ${\sim}0.75$M GPT-4 music-QA/summary items, sits alongside math, code, and general chat) whose music is stored as \emph{ABC notation}, score, not performance, and targets music \emph{understanding} rather than generation. Our corpus is instead single-purpose (caption$\rightarrow$performance-MIDI pairs) and preserves performance parameters ABC discards, so MusicPile is neither a substitute data source nor a like-for-like comparison for our generation corpus.

\paragraph{Acquisition, all sources.}\label{apx:acquisition-all-sources} The corpus draws from four source families spanning Western and Chinese repertoire; Table~\ref{tab:sources} gives per-source counts. \textbf{(1)~Native symbolic collections} (pre-scored, no transcription): GigaMIDI~\citep{lee2025gigamidi} ($1.43$M) and Aria-MIDI~\citep{bradshaw2025aria} ($1.19$M raw; we keep a $371$k content-deduplicated subset), plus \textbf{Discover-MIDI}, the Los-Angeles-MIDI web aggregate ($6.74$M raw $\rightarrow$ $3.59$M high-quality net-new after content-hash dedup and the gates below); the genre-balanced Lakh~\citep{raffel2016learning} subset and GiantMIDI~\citep{kong2022giantmidi} classical piano additionally feed the aligned 86.6k set. \textbf{(2)~Licensed packs}: multiple purchased collections ($627$k after dedup). \textbf{(3)~Self-built audio$\rightarrow$MIDI transcriptions} of permissively-licensed audio: Free Music Archive ($137.6$k), AudioSet-music ($26.2$k), MTG-Jamendo/Jamendo CC crawls, and CC0-music ($8.3$k, ships native captions). \textbf{(4)~Chinese-repertoire sources}, added to counter the Western-centric public sets: ACE-Opencpop ($77.1$k human-annotated singing MIDI$+$lyrics), M4Singer and crawled singing ($42.7$k), the CCMusic folk/vocal suite ($4.0$k), POP909 ($549$; audio$+$MIDI$+$chord annotations), MusicCaps-zh ($5.1$k), and MIR-1K. MAESTRO~\citep{hawthorne2018enabling} and POP909~\citep{wang2020pop909} are held out as classical-piano and C-pop \emph{evaluation} benchmarks (never trained on). Web crawls are queried per world region (Mandarin/Cantonese and Western/world tags) to broaden coverage. Native and annotation sources enter the gates directly; audio-only sources carry no score and are transcribed next.

\paragraph{Streaming transcription.}\label{apx:streaming-transcription} Audio is transcribed to MIDI with a CNN onset-and-pitch model (\textsc{basic-pitch}, ONNX backend). Because raw audio is $10$--$50\times$ larger than its transcribed MIDI and working storage is bounded, transcription is \emph{streaming}: each clip is downloaded, transcribed, and its audio deleted immediately, so peak disk is set by the in-flight working set rather than corpus size. Download and transcription run as self-healing daemons (relaunch on failure, skip completed shards) across machines against a shared store; completed source archives are offloaded to object storage.

\paragraph{Quality curation.}\label{apx:quality-curation} Transcription of complex or low-fidelity audio is noisy, spurious onsets, over-segmented sustains, or degenerate streams. Because streaming deletes the reference audio, we cannot score true audio--MIDI \emph{alignment}; instead we apply five interpretable well-formedness gates that proxy transcription cleanliness (a clip is retained only if \emph{all} pass): \textbf{(G1)}~$\ge 8$ notes and duration in $[5,900]$\,s; \textbf{(G2)}~note rate in $[0.3,25]$\,notes/s (rejects sparse artifacts and onset-noise storms); \textbf{(G3)}~fewer than $60\%$ of notes shorter than $60$\,ms (rejects over-segmentation); \textbf{(G4)}~$\ge 4$ distinct pitches and pitch range $\ge 3$ semitones (rejects degenerate streams); \textbf{(G5)}~median note duration $\ge 45$\,ms (rejects click-train transcriptions). Each retained clip also receives a continuous quality score $q\in[0,1]$, the geometric mean of soft sub-scores on density-centredness (log-centred at $4$\,notes/s), low fragmentation, pitch diversity, and sustain, for optional stricter thresholding. On the current transcribed corpus ($137.7$k clips), $93.4\%$ pass all gates ($128{,}498$ retained; median quality $q{=}0.89$, mean $0.85$, and $90.3\%$ at the stricter $q\!\ge\!0.5$), with median density $4.2$\,notes/s and median short-note fraction $0.00$, clean output, not fragmented noise. Rejections are dominated by near-empty ($3.8\%$, $<\!8$ notes) and pitch-degenerate ($2.4\%$, $<\!4$ distinct pitches) clips; density, duration, and click-train gates remove a further $<\!0.4\%$ combined. Thresholds and per-clip scores ship with the release for reproducibility and re-filtering.

\paragraph{Deduplication and captioning.}\label{apx:deduplication-and-captioning} Retained MIDIs are deduplicated by content hash against both the aggregated pool and the frozen evaluation split (preventing test leakage), then captioned with the same Qwen3-Omni pipeline as the main corpus (deterministic render $\rightarrow$ MLLM caption), yielding (caption, MIDI) pairs directly usable by \S Method. The pipeline is representation-agnostic: curated MIDIs feed the PMT serializer and every baseline tokenizer identically, so corpus growth benefits all arms of the controlled comparison equally.

\paragraph{Captioning at scale.}\label{apx:captioning-at-scale} We caption the full deduplicated pool with Qwen3-Omni-30B, a large audio--language model, run as a self-healing sharded fleet, atomic per-shard claims with stale-claim stealing on preemption, that sustains ${\approx}0.75$\,s/piece per A800 (throughput is decode-bound by caption length, not render or batch). Each MIDI is rendered to audio deterministically (FluidSynth, fixed SoundFont) and the model is prompted with the rendered audio \emph{plus} score-derived metadata, instrument list, time signature, tempo, key, duration, and any scraped catalog title/artist, and asked to describe genre, mood, instrumentation, tempo, rhythm, harmony, and texture in $3$--$5$ sentences (mean ${\approx}365$ tokens). Grounding on the render (not on tags) yields captions of what is actually heard, while the metadata keeps instrument, tempo, and key factual and forbids invented titles. Rendering is bounded by a short wall-clock cap: a small fraction of web-scraped MIDIs carry corrupt duration metadata that drives the synthesizer into a multi-hour runaway tail, so we retain only the leading seconds of real audio and discard the tail, bounding working storage and keeping the pathological files captionable rather than skipped. The pass yields \textbf{6{,}250{,}143} (caption, MIDI) pairs ($99.2\%$ of the target pool; the remainder are un-renderable files).

\paragraph{Caption quality control.}\label{apx:caption-quality-control} Machine captions carry a small failure tail, refusals, empty stubs, off-topic text, degenerate repetition, or non-English fragments, so, mirroring the MIDI gates, we apply five interpretable caption gates (a caption is kept only if \emph{all} pass): \textbf{(C1)}~length $8$--$400$ words; \textbf{(C2)}~non-repetition, distinct-$4$-gram ratio $\ge 0.55$ and no $3$-gram repeated more than $4\times$; \textbf{(C3)}~$\ge 3$ distinct musical terms from a curated genre/instrument/tempo/mood/theory lexicon; \textbf{(C4)}~English, ASCII-letter ratio $\ge 0.90$; \textbf{(C5)}~no refusal or ``no-audio'' artifact. Each caption also receives a continuous quality score $q\in[0,1]$ (geometric mean of soft sub-scores on length-centredness, low repetition, term density, script, and artifact-freeness) for optional stricter thresholding. On the $6.25$M captions, \textbf{$99.97\%$ pass all gates} ($6{,}248{,}214$ retained; only $1{,}929$ removed), rejections dominated by refusal artifacts ($0.009\%$) and length stubs ($0.016\%$), with off-topic ($0.005\%$), repetition ($0.001\%$), and non-English ($0.0002\%$) negligible, direct evidence that audio-grounded MLLM captioning is well-formed at scale. Gates, per-caption $q$, and rejects ship with the release for reproducibility and re-filtering.

\paragraph{Positioning among caption datasets.}\label{apx:positioning-among-caption-datasets} The result is, to our knowledge, the largest captioned symbolic-music dataset: $37\times$ MidiCaps~\citep{melechovsky2024midicaps} ($168$k) and $6.5\times$ MetaScore~\citep{xu2024generating} ($963$k), the two prior symbolic caption sets, and larger than every music-specific caption resource, including the audio-domain LP-MusicCaps~\citep{doh2023lp} ($2.2$M tag-derived pseudo-captions); only \emph{general}-audio caption corpora reach comparable counts. Beyond scale, our captions are more strongly grounded and richer than prior symbolic sets: MidiCaps composes audio-feature extraction with LLM templating and MetaScore derives text from score \emph{metadata} (tags/genre), whereas ours are produced by an audio--language model that \emph{listens} to the rendered performance. Two caveats bound the claim honestly: the captions are grounded in a deterministic \emph{rendering} of the same MIDI rather than a real recording (consistent self-annotation, not human ground truth), and quality control is automated well-formedness rather than a human hallucination audit, a human-rated subset with inter-rater agreement and a measured hallucination rate is the natural next step, left to future work. Table~\ref{tab:capdata} places our corpus among \emph{music} caption datasets of both domains, audio-domain (captions of rendered/real music audio) and symbolic-domain (captions paired with scores). Scoped correctly to \emph{music} (not general audio), it is the largest: $2.8\times$ the largest prior music caption set of any domain (the audio-domain LP-MusicCaps~\citep{doh2023lp}, 2.2M) and $6.5$--$37\times$ the prior \emph{symbolic} sets (MetaScore~\citep{xu2024generating}/MidiCaps~\citep{melechovsky2024midicaps}). We are careful to distinguish this from \emph{general-audio} caption corpora, AudioSetCaps~\citep{audiosetcaps2024} (${\sim}$6M) and WavCaps caption speech and environmental sound alongside music, a different task, so we do not count them as music datasets. Human-annotated music sets (MusicCaps~\citep{agostinelli2023musiclm}, Song Describer~\citep{manco2023song}) remain the quality gold standard but are three orders of magnitude smaller.

\begin{table*}[t]
\centering\small
\setlength{\tabcolsep}{3pt}
\begin{tabular}{lrccc}
\toprule
Caption dataset & \#Items & Domain & Caption source & Grnd. \\
\midrule
MusicCaps~\citep{agostinelli2023musiclm}            & 5.5k  & audio & human & \checkmark \\
Song Describer~\citep{manco2023song} & 1.1k  & audio & human & \checkmark \\
MusicBench~\citep{melechovsky2024mustango}        & 53k   & audio & MusicCaps{+}MIR & $\sim$ \\
JamendoMaxCaps~\citep{roy2025jamendomaxcaps}& 362k  & audio & audio-LM & \checkmark \\
LP-MusicCaps~\citep{doh2023lp}     & 2.2M  & audio & LLM${\leftarrow}$tags & $\times$ \\
\midrule
MidiCaps~\citep{melechovsky2024midicaps}            & 168k  & symbolic & feat.{+}LLM tmpl. & $\sim$ \\
MetaScore~\citep{xu2024generating}          & 963k$^{\ddagger}$ & symbolic & LLM${\leftarrow}$meta & $\times$ \\
\textbf{Ours}                            & \textbf{6.25M} & symbolic & audio-LM${\leftarrow}$render & \checkmark \\
\bottomrule
\end{tabular}
\caption{\textbf{Music} caption datasets, audio-domain and symbolic-domain. \textbf{Grnd.}\ $=$ captions grounded in actual audio content (\checkmark), partial ($\sim$), or from tags/metadata only ($\times$). Ours is the largest music caption dataset, $2.8\times$ the next (LP-MusicCaps, 2.2M) and $6.5$--$37\times$ prior \emph{symbolic} sets. $^{\ddagger}$MetaScore's full size; only the 326k publicly captioned subset is listed in Table~\ref{tab:datasets}. We deliberately \emph{exclude} general-audio caption corpora (AudioSetCaps~\citep{audiosetcaps2024} ${\sim}$6M, WavCaps), these caption speech and environmental sound, not music, and are not comparable.}
\label{tab:capdata}
\end{table*}

\paragraph{Per-source acquisition--curation funnel.}\label{apx:per-source-acquisition-curation-funnel} Table~\ref{tab:funnel} accounts for every source from raw acquisition through content-hash deduplication and the quality gates to the high-quality MIDI that enters the pool. Native pre-curated collections (GigaMIDI, Aria-MIDI) pass the gates near-unchanged; the heavy filtering falls on the licensed packs, where near-empty loops and cross-pool duplicates cut the count to a quarter, and on the largest web aggregate, Discover-MIDI, where $27\%$ duplicates content we already hold and a further $20\%$ fails the well-formedness gates, leaving $53\%$ ($3.59$M) high-quality net-new. The audio path adds a pre-transcription audio gate (loudness, silence fraction, spectral flatness) that removes $13{,}974$ of $132{,}975$ crawled clips before any transcription compute is spent. Across all sources, $11.3$M raw candidates reduce to $6.30$M curated MIDI ($56\%$); this is the raw symbolic pool that the captioning pipeline scales language supervision onto, and is $73\times$ the $86.6$k four-modality set used for the controlled experiments.

\begin{table}[t]
\centering\small
\setlength{\tabcolsep}{2.6pt}
\begin{tabular}{lrrr}
\toprule
Source & Raw & Removed & Retained \\
\midrule
GigaMIDI $+$ Aria (native)      & 1{,}884{,}262 & --$^{a}$      & 1{,}884{,}262 \\
Purchased packs (native)        & 2{,}518{,}754 & 1{,}891{,}595 & \phantom{0}627{,}159 \\
Self-built (audio$+$annot.)$^{b}$ & \phantom{0,0}272{,}021 & \phantom{00,0}77{,}655 & \phantom{0}194{,}366 \\
Discover-MIDI (native)          & 6{,}747{,}346 & 3{,}153{,}119 & 3{,}594{,}227 \\
\midrule
\textbf{Total curated pool}     & \textbf{11.4M} & \textbf{5.1M} & \textbf{6{,}300{,}014} \\
\bottomrule
\end{tabular}
\caption{Per-source acquisition$\rightarrow$curation funnel. \emph{Removed} $=$ content-hash duplicates $+$ clips failing the quality gates. $^{a}$Native, pre-deduplicated at source; gates near-pass-through. $^{b}$Raw $=132{,}975$ crawled audio clips ($\rightarrow$$119{,}001$ past the audio gate $\rightarrow$ transcribed) plus $139{,}046$ score-annotation records (ACE-Opencpop, POP909); strict G1--G5 subset $=168{,}667$, i.e.\ $87.7\%$ of the 192k captioned self-built pool (Table~\ref{tab:sources}).}
\label{tab:funnel}
\end{table}

\begin{table*}[t]
\centering\footnotesize
\setlength{\tabcolsep}{5pt}
\begin{tabular}{llllr}
\toprule
Source & Sourced from & Category & Caption/label & \#Pieces \\
\midrule
GigaMIDI$+$Aria~\citep{lee2025gigamidi,bradshaw2025aria} & public & W. multi-instr. & Omni & 1{,}884k \\
Discover-MIDI                      & web (LA-MIDI) & W. multi-instr. & Omni & 3{,}557k \\
Purchased packs                    & licensed & mixed multi-instr. & Omni & \phantom{0,}623k \\
\midrule
\multicolumn{5}{l}{\emph{Self-built (192k captioned after well-formedness gating $+$ content-hash dedup):}}\\
FMA                                & audio$\to$MIDI & W. mixed & Omni & \phantom{0,}138k \\
ACE-Opencpop                       & annot. singing MIDI & Mandarin sing & native+Omni & \phantom{0,0}77k \\
Crawled $+$ M4Singer               & audio$\to$MIDI & Mandarin sing & Omni & \phantom{0,0}43k \\
AudioSet-Music                     & audio$\to$MIDI & general & Omni & \phantom{0,0}26k \\
CC0-Music                          & CC0 audio & mixed & native+Omni & \phantom{0,00}8k \\
MusicCaps-zh                       & MusicCaps-zh & mixed & native & \phantom{0,00}5k \\
CCMusic                            & Ch. folk/vocal & Chinese trad. & Omni & \phantom{0,00}4k \\
POP909~\citep{wang2020pop909}      & audio+MIDI+chords & C-pop & Omni & \phantom{0,00}0.5k \\
\midrule
\textbf{Total (captioned)} & & & & \textbf{6{,}248k} \\
\bottomrule
\end{tabular}
\caption{\textbf{Constituent sources of the 6.25M scaled corpus.} Counts are curated-and-captioned pieces; per-source provenance and gating are in the text.}
\label{tab:sources}
\end{table*}

\paragraph{Corpus statistics: duration, density, token length.}\label{apx:corpus-statistics-duration-density-token-length} Table~\ref{tab:corpusstats} and Figure~\ref{fig:datadist} report the per-pool and overall distributions (30k-per-pool sample). The pool totals $164{,}064$ hours ($\approx$18.7 years; mean $94$\,s/piece). \emph{Duration is bimodal} (Fig.~\ref{fig:datadist}a): licensed loops and short transcriptions concentrate below $30$\,s ($46\%$ of pieces), while the public and Discover collections contribute a long-form mode at $2$--$10$ minutes ($28\%$), so the corpus spans both loop-scale and full-arrangement material (median $36$\,s, p99 $626$\,s). Note density is healthy everywhere, per-pool medians $3.7$--$8.2$ notes/s, confirming the gates leave dense musical content rather than sparse transcription artifacts (a tightened sparse/silence gate, density $<0.5$\,notes/s or max onset gap $>20$\,s, removes a further $<\!1\%$ long-dead-air tail). \emph{PMT-token length inherits this structure and is strongly heavy-tailed} (Fig.~\ref{fig:datadist}b,c): median $940$ but mean $3{,}930$ and p99 $33{,}633$, with $14\%$ of pieces exceeding $8$k tokens. Coverage is pool-dependent, the licensed short-loop pool is $95\%$ captured by a $1024$-token block, whereas the long-form public and Discover pools sit below $50\%$, so training-block economics depend on the source mixture, not just the piece count. \emph{Handling the tail.} Pieces up to $4\times$ the block are trained head-truncated; longer pieces are dropped by a residual-explosion guard: at block $8192$ we train $86\%$ of pieces whole, head-truncate $13\%$, and skip only the $1\%$ pathological tail ($>\!32$k tokens), versus $74\%/20\%/6\%$ at block $4096$.

\begin{table}[t]
\centering\small
\setlength{\tabcolsep}{3pt}
\begin{tabular}{lrrrrr}
\toprule
Pool & Pieces & Hrs & Dur$_{\text{m}}$ & Dens$_{\text{m}}$ & Tok$_{\text{m}}$/$\le$1k \\
\midrule
Public       & 1{,}884{,}262 & 59{,}102 & 63\,s & 7.5 & 1412 / 47\% \\
Self-built   & \phantom{0}194{,}366 & \phantom{0}2{,}747 & 30\,s & 3.7 & \phantom{0}367 / 78\% \\
Purchased    & \phantom{0}627{,}159 & \phantom{0}2{,}924 & \phantom{0}8\,s & 8.2 & \phantom{0}210 / 95\% \\
Discover     & 3{,}594{,}227 & 99{,}291 & 43\,s & 7.5 & 1247 / 46\% \\
\midrule
\textbf{Total} & \textbf{6{,}300{,}014} & \textbf{164{,}064} & 36\,s & 7.4 & 940 / 52\% \\
\bottomrule
\end{tabular}
\caption{Curated-pool statistics by source. \textbf{Hrs} $=$ total hours; \textbf{Dur$_{\text{med}}$} $=$ median piece duration; \textbf{Dens$_{\text{med}}$} $=$ median note density (notes/s); \textbf{Tok$_{\text{med}}$/${\le}1024$} $=$ median PMT-token length and the fraction of pieces a $1024$-token block covers unclipped. Duration/token statistics from a $30$k-per-pool random sample (density from an earlier $12$k sample); counts and hours exact.}
\label{tab:corpusstats}
\end{table}

\begin{figure*}[tp]
\centering
\includegraphics[width=\textwidth]{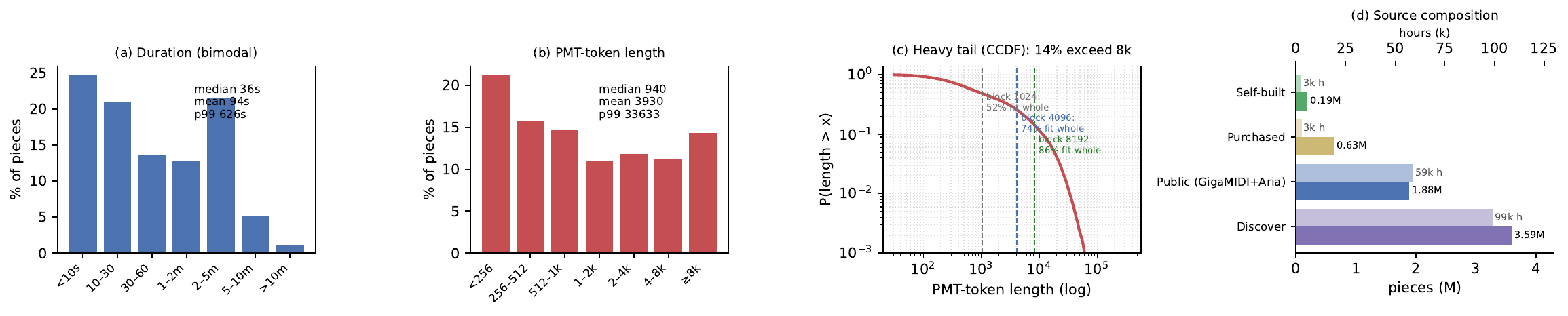}
\caption{\textbf{Statistics of the 6.25M captioned corpus} (30k-per-pool random sample; counts and hours exact). \emph{(a)}~Piece duration. \emph{(b)}~PMT-token length. \emph{(c)}~Block coverage. \emph{(d)}~Source composition.}
\label{fig:datadist}
\end{figure*}

\section{Generation Gallery: Caption, Piano Roll and Engraved Score}
\label{app:gallery}
Figures~\ref{fig:gala}--\ref{fig:galk} put thirty-one generations on the page in the three views that let
a reader judge them without listening: the caption that produced the piece, its piano roll, and its
opening bars engraved as notation. All three come from the same MIDI file in the same pass, so the views
cannot disagree with each other. Figures~\ref{fig:cmpa}--\ref{fig:cmpb} then give the same caption to our
model and to five published systems, with the real reference alongside.

Two things about the gallery need saying plainly. \textbf{The examples are curated.} Every one was kept by
a human listener working through a pool of several hundred, so this is the strong end of the distribution,
not a sample from it; the printed note count, instrument count, maximum polyphony and chord-time fraction
are there so a reader can see how far from typical each one is (aggregate distributions are in
Appendix~\ref{app:extmetrics}). \textbf{The staff labels are what the model emitted, not what the caption
asked for.} They are the General MIDI programs in the generated file, so where a caption requests one
instrumentation and the notation shows another, that is a real caption-adherence miss of the kind
Table~\ref{tab:midicapsfull} quantifies, visible here rather than averaged away.

Blocks are labelled with the model that produced them: most are the 4B model trained on the 6.25M corpus,
some come from the 86.6k controlled-grid models at 4B/27B/35B-A3B, and two decode with the budget raised
to 7000 tokens (Appendix~\ref{app:declen}), which is where the longest pieces in the gallery come from.
The engraved excerpt is the opening bars after quantization to a sixteenth grid, a legibility choice for
print and not how the model writes: the notation is quantized, the audio and the roll are not. The number
of bars shown adapts to the number of staves. The system comparison uses rolls only and a common $24$\,s
window, so bar widths mean the same thing on every row.

\begin{figure*}[tp]
\centering
\includegraphics[width=\textwidth]{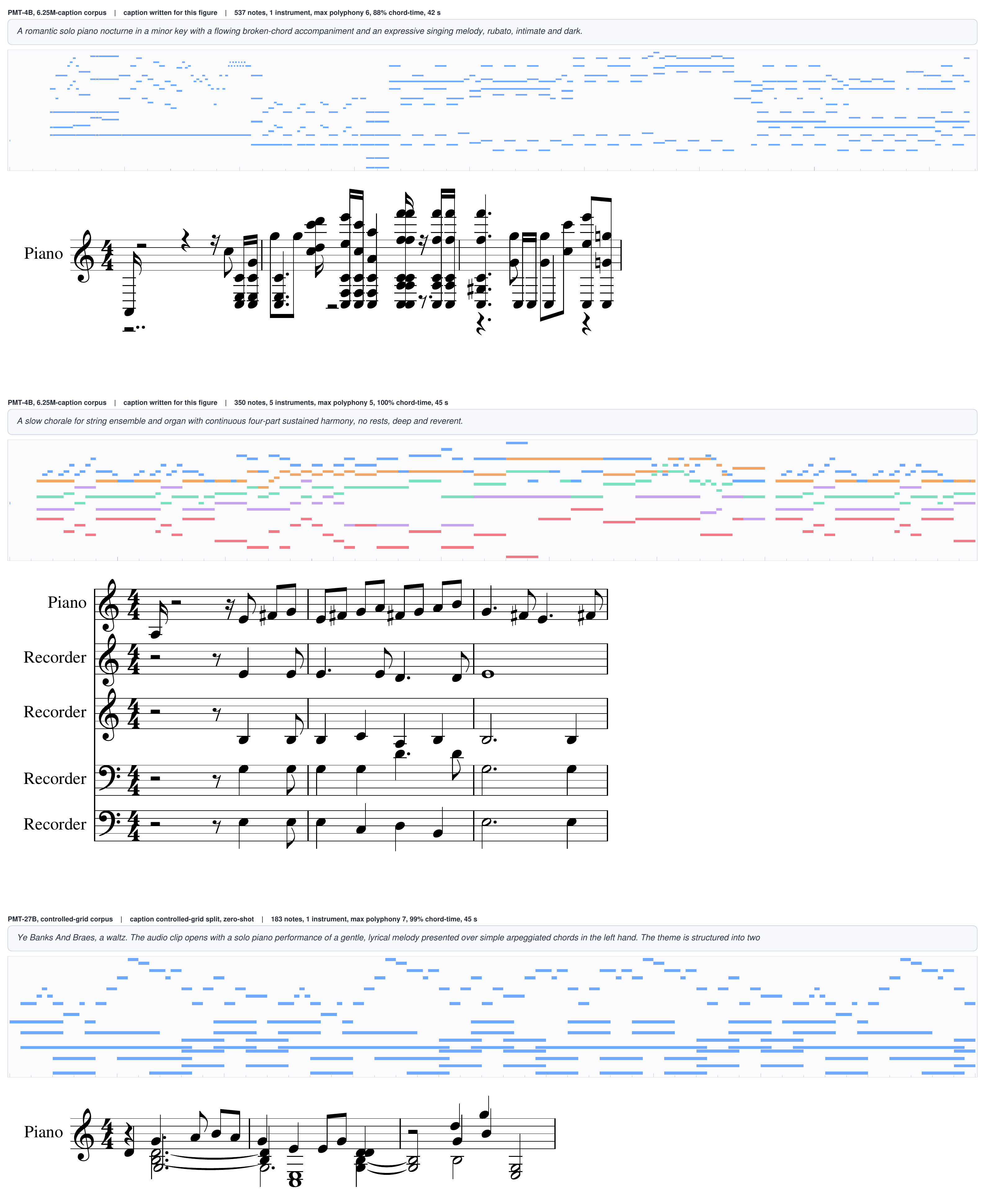}
\caption{\textbf{Generation gallery (1 of 11).} Per example: model and caption provenance, measured
texture, the caption, the piano roll over the whole clip (colour $=$ track, vertical axis $=$ pitch), and
the opening bars engraved. All curated; see Appendix~\ref{app:gallery}.}
\label{fig:gala}
\end{figure*}

\begin{figure*}[tp]
\centering
\includegraphics[width=\textwidth]{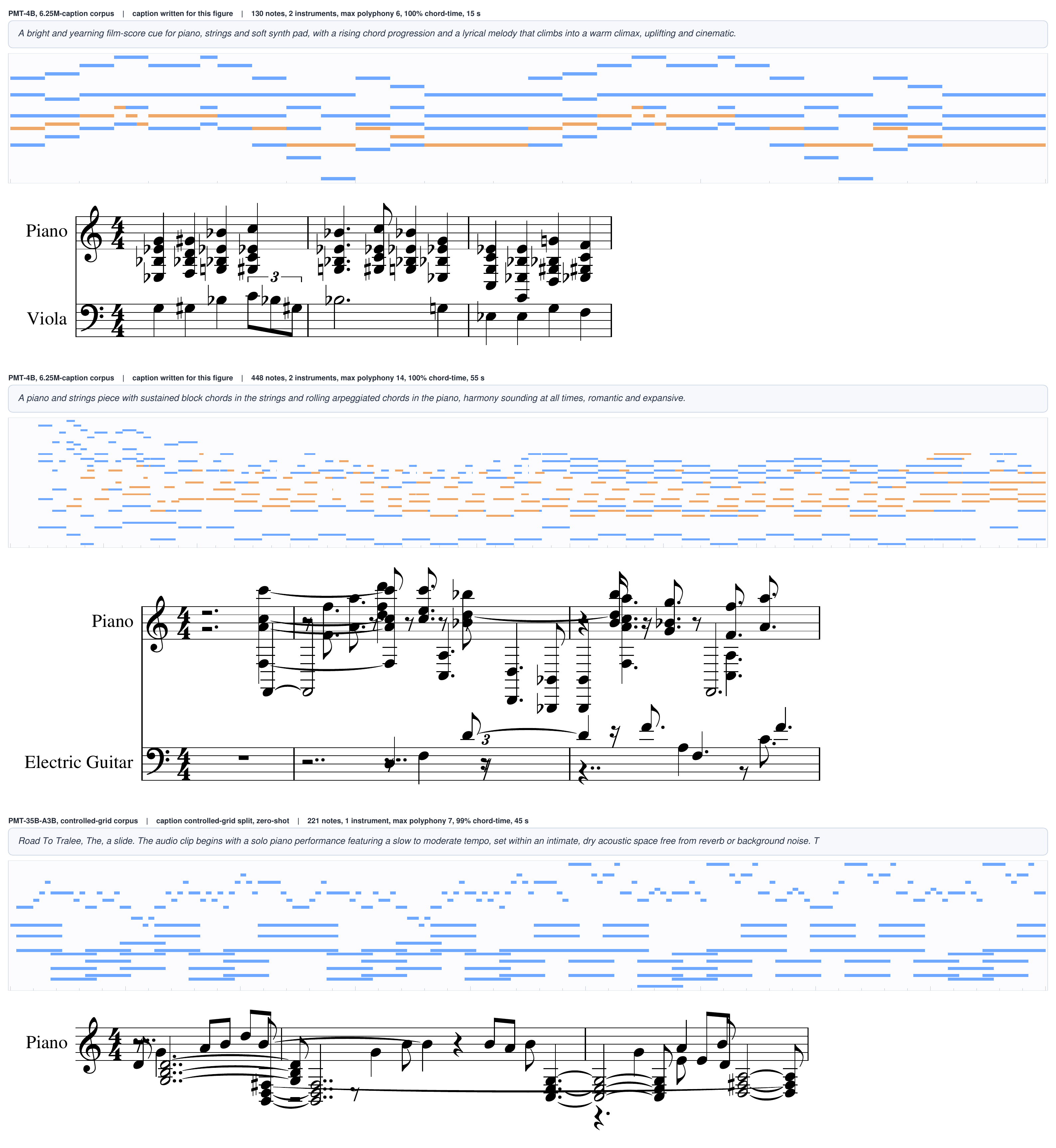}
\caption{\textbf{Generation gallery (2 of 11).} Per example: model and caption provenance, measured
texture, the caption, the piano roll over the whole clip (colour $=$ track, vertical axis $=$ pitch), and
the opening bars engraved. All curated; see Appendix~\ref{app:gallery}.}
\label{fig:galb}
\end{figure*}

\begin{figure*}[tp]
\centering
\includegraphics[width=\textwidth]{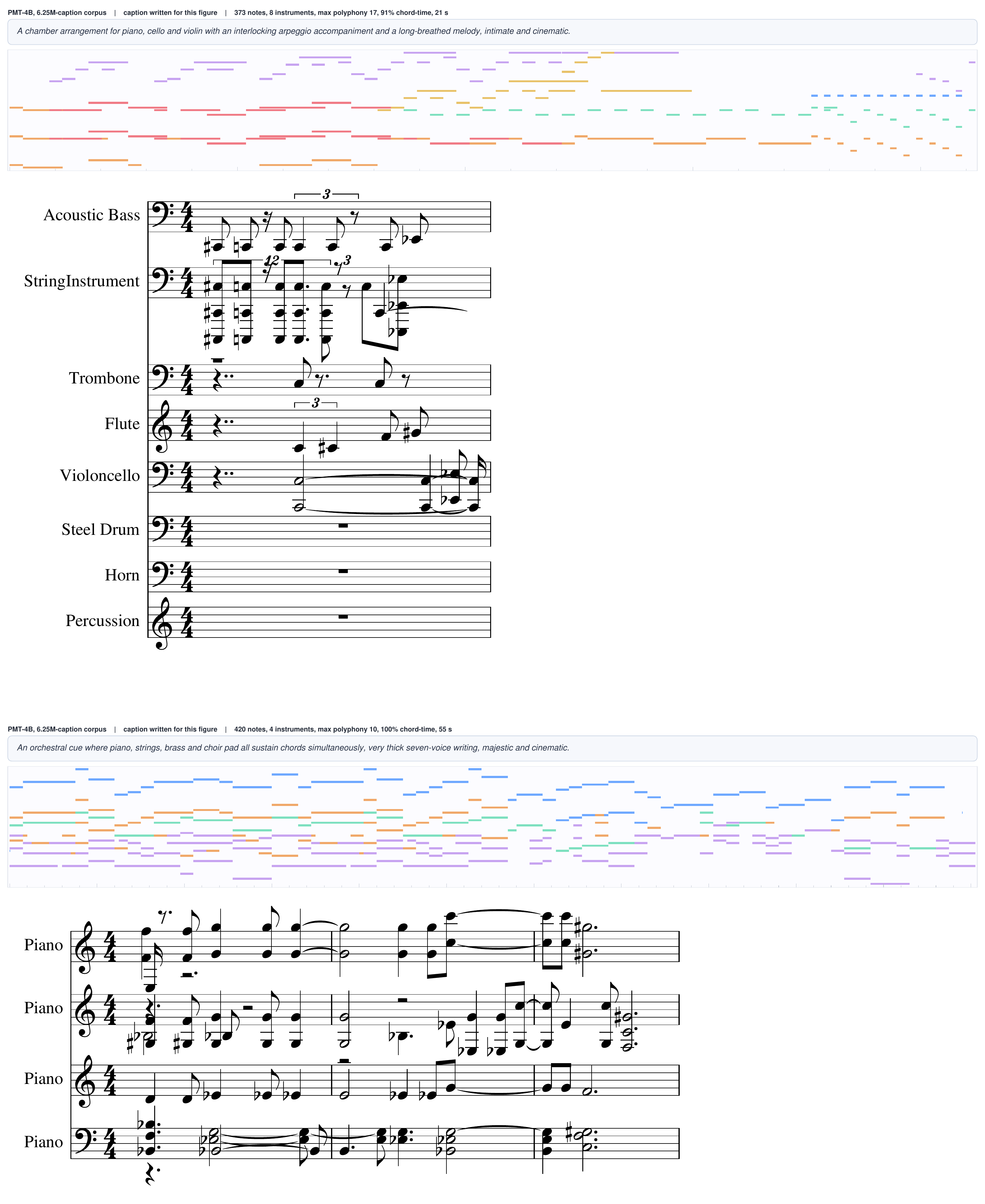}
\caption{\textbf{Generation gallery (3 of 11).} Per example: model and caption provenance, measured
texture, the caption, the piano roll over the whole clip (colour $=$ track, vertical axis $=$ pitch), and
the opening bars engraved. All curated; see Appendix~\ref{app:gallery}.}
\label{fig:galc}
\end{figure*}

\begin{figure*}[tp]
\centering
\includegraphics[width=\textwidth]{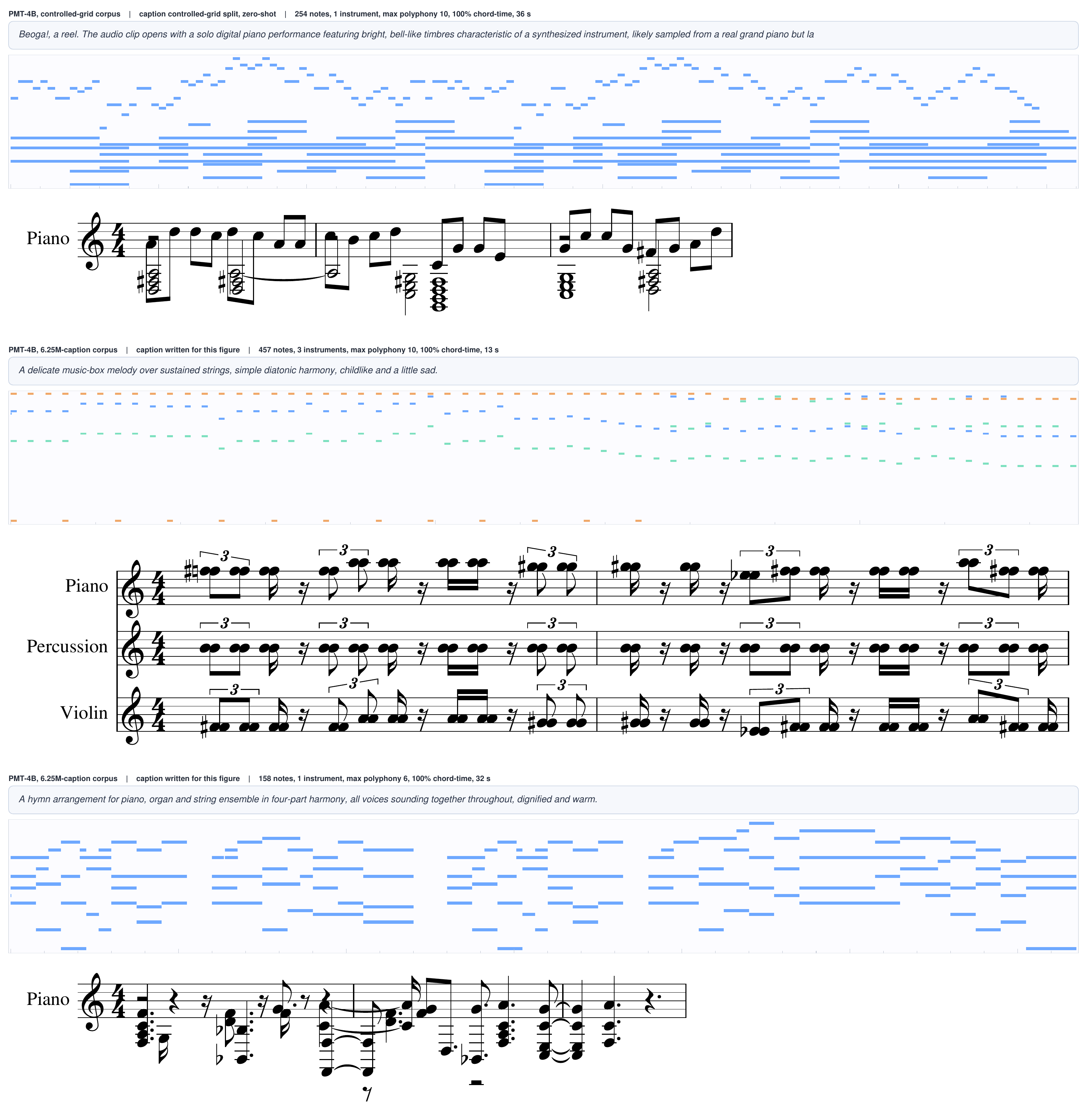}
\caption{\textbf{Generation gallery (4 of 11).} Per example: model and caption provenance, measured
texture, the caption, the piano roll over the whole clip (colour $=$ track, vertical axis $=$ pitch), and
the opening bars engraved. All curated; see Appendix~\ref{app:gallery}.}
\label{fig:gald}
\end{figure*}

\begin{figure*}[tp]
\centering
\includegraphics[width=\textwidth]{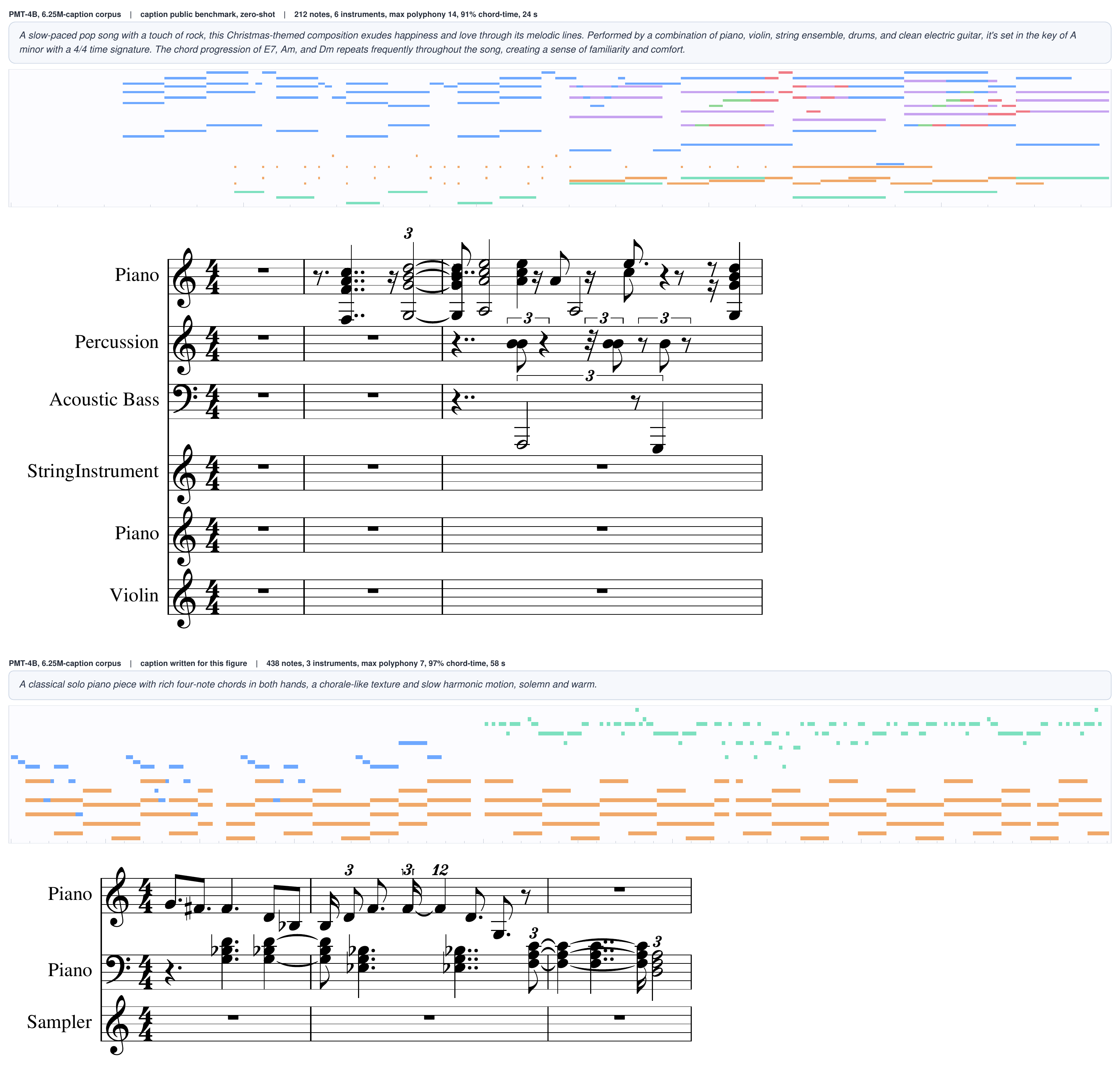}
\caption{\textbf{Generation gallery (5 of 11).} Per example: model and caption provenance, measured
texture, the caption, the piano roll over the whole clip (colour $=$ track, vertical axis $=$ pitch), and
the opening bars engraved. All curated; see Appendix~\ref{app:gallery}.}
\label{fig:gale}
\end{figure*}

\begin{figure*}[tp]
\centering
\includegraphics[width=\textwidth]{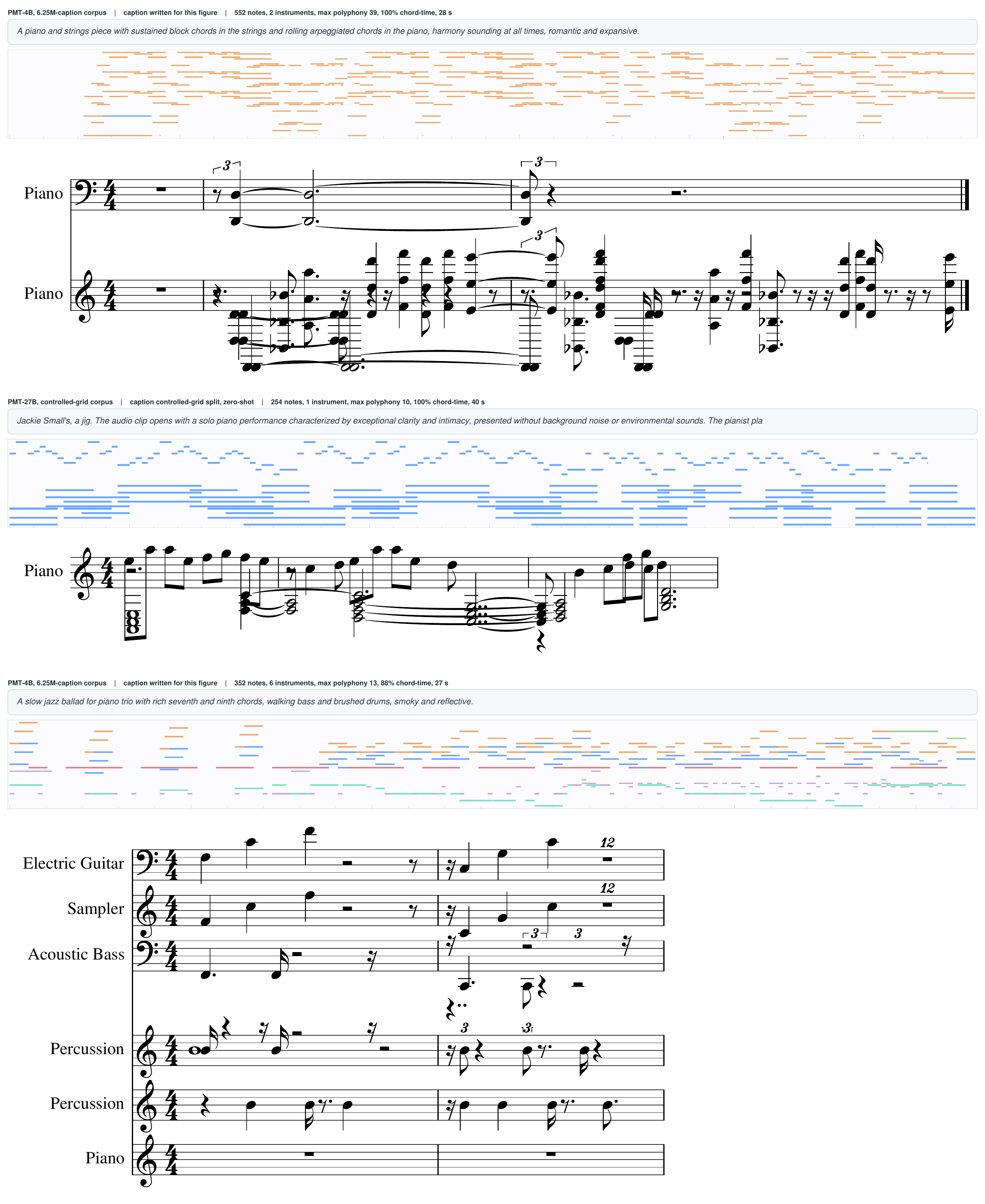}
\caption{\textbf{Generation gallery (6 of 11).} Per example: model and caption provenance, measured
texture, the caption, the piano roll over the whole clip (colour $=$ track, vertical axis $=$ pitch), and
the opening bars engraved. All curated; see Appendix~\ref{app:gallery}.}
\label{fig:galf}
\end{figure*}

\begin{figure*}[tp]
\centering
\includegraphics[width=\textwidth]{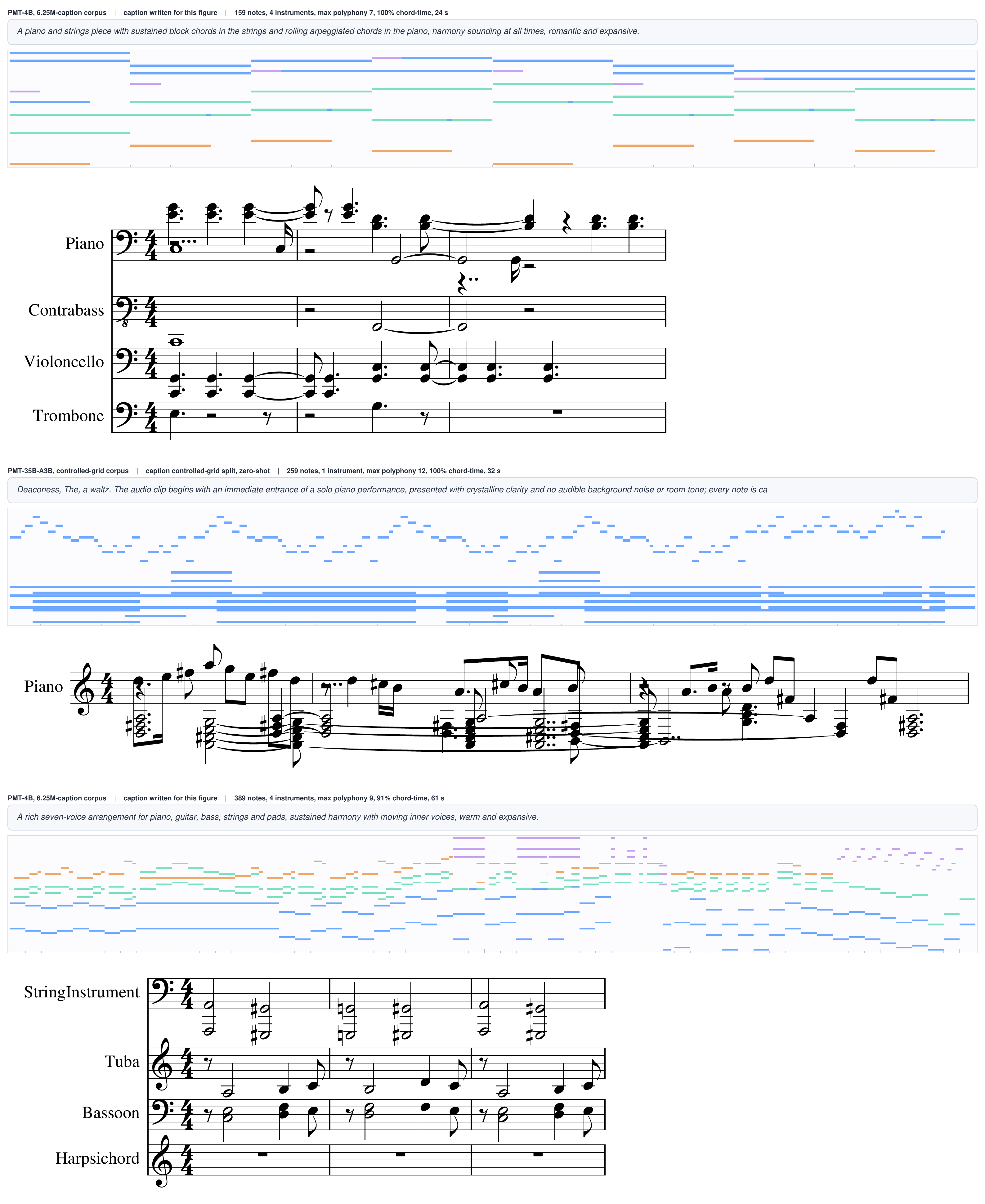}
\caption{\textbf{Generation gallery (7 of 11).} Per example: model and caption provenance, measured
texture, the caption, the piano roll over the whole clip (colour $=$ track, vertical axis $=$ pitch), and
the opening bars engraved. All curated; see Appendix~\ref{app:gallery}.}
\label{fig:galg}
\end{figure*}

\begin{figure*}[tp]
\centering
\includegraphics[width=\textwidth]{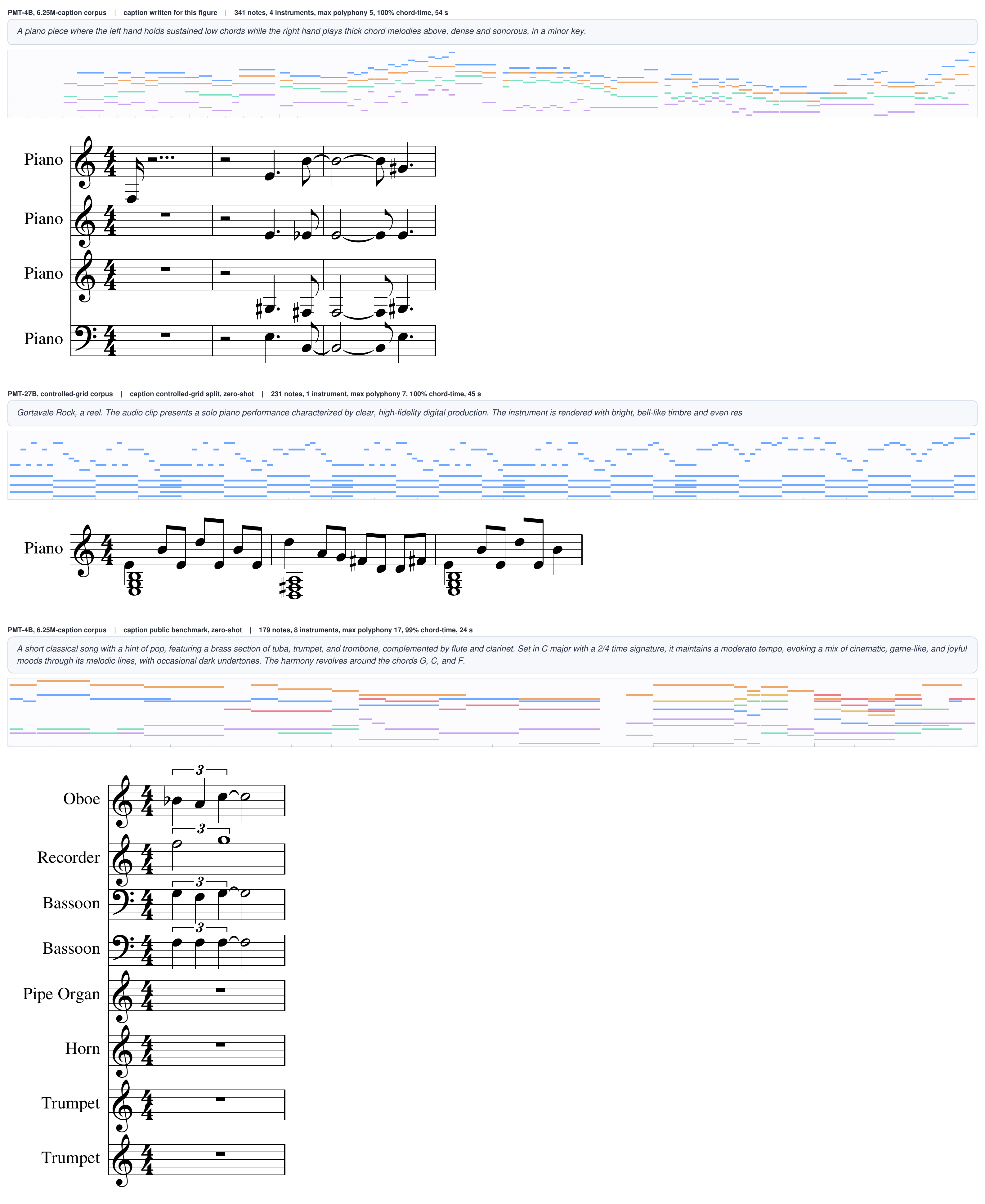}
\caption{\textbf{Generation gallery (8 of 11).} Per example: model and caption provenance, measured
texture, the caption, the piano roll over the whole clip (colour $=$ track, vertical axis $=$ pitch), and
the opening bars engraved. All curated; see Appendix~\ref{app:gallery}.}
\label{fig:galh}
\end{figure*}

\begin{figure*}[tp]
\centering
\includegraphics[width=\textwidth]{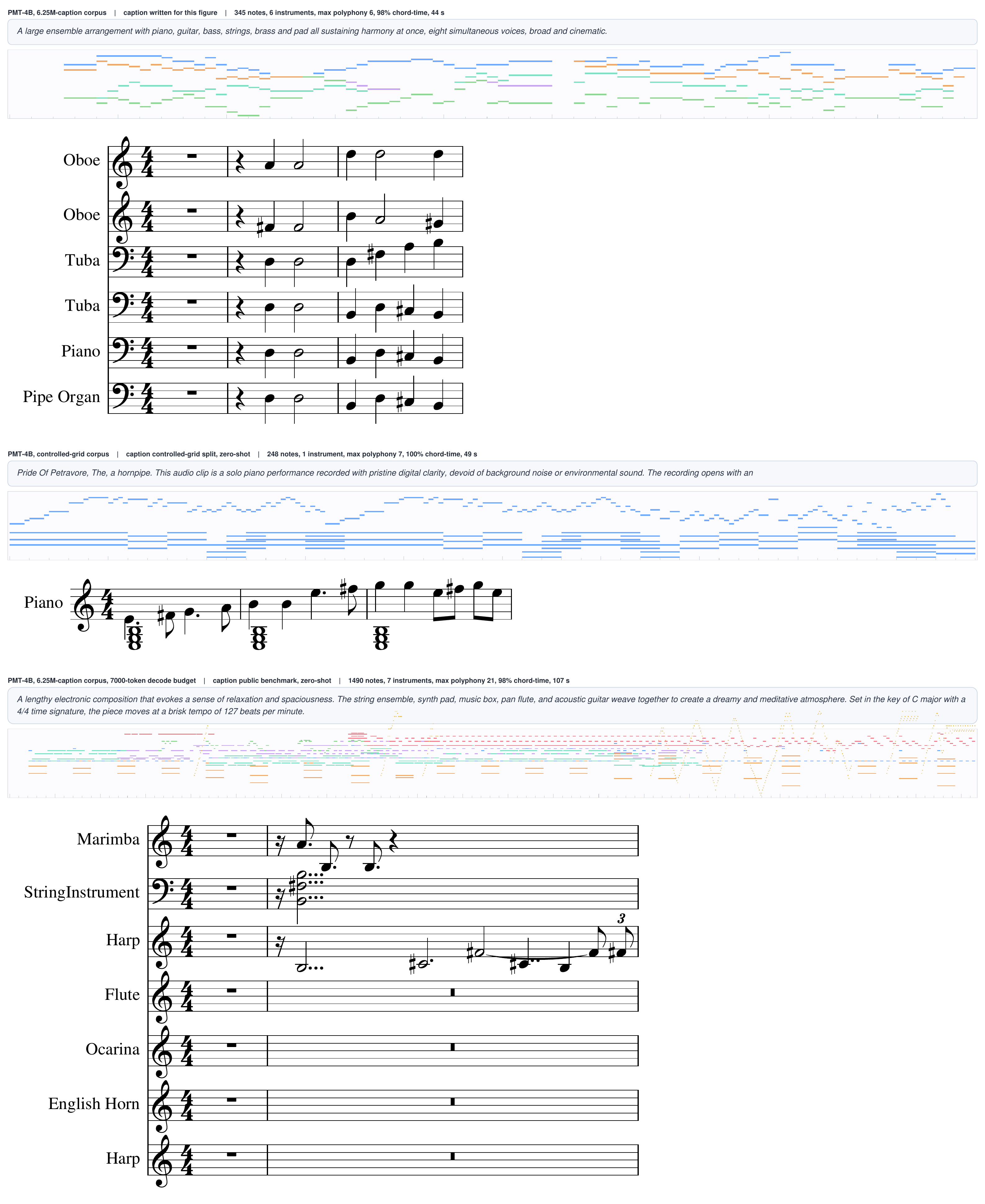}
\caption{\textbf{Generation gallery (9 of 11).} Per example: model and caption provenance, measured
texture, the caption, the piano roll over the whole clip (colour $=$ track, vertical axis $=$ pitch), and
the opening bars engraved. All curated; see Appendix~\ref{app:gallery}.}
\label{fig:gali}
\end{figure*}

\begin{figure*}[tp]
\centering
\includegraphics[width=\textwidth]{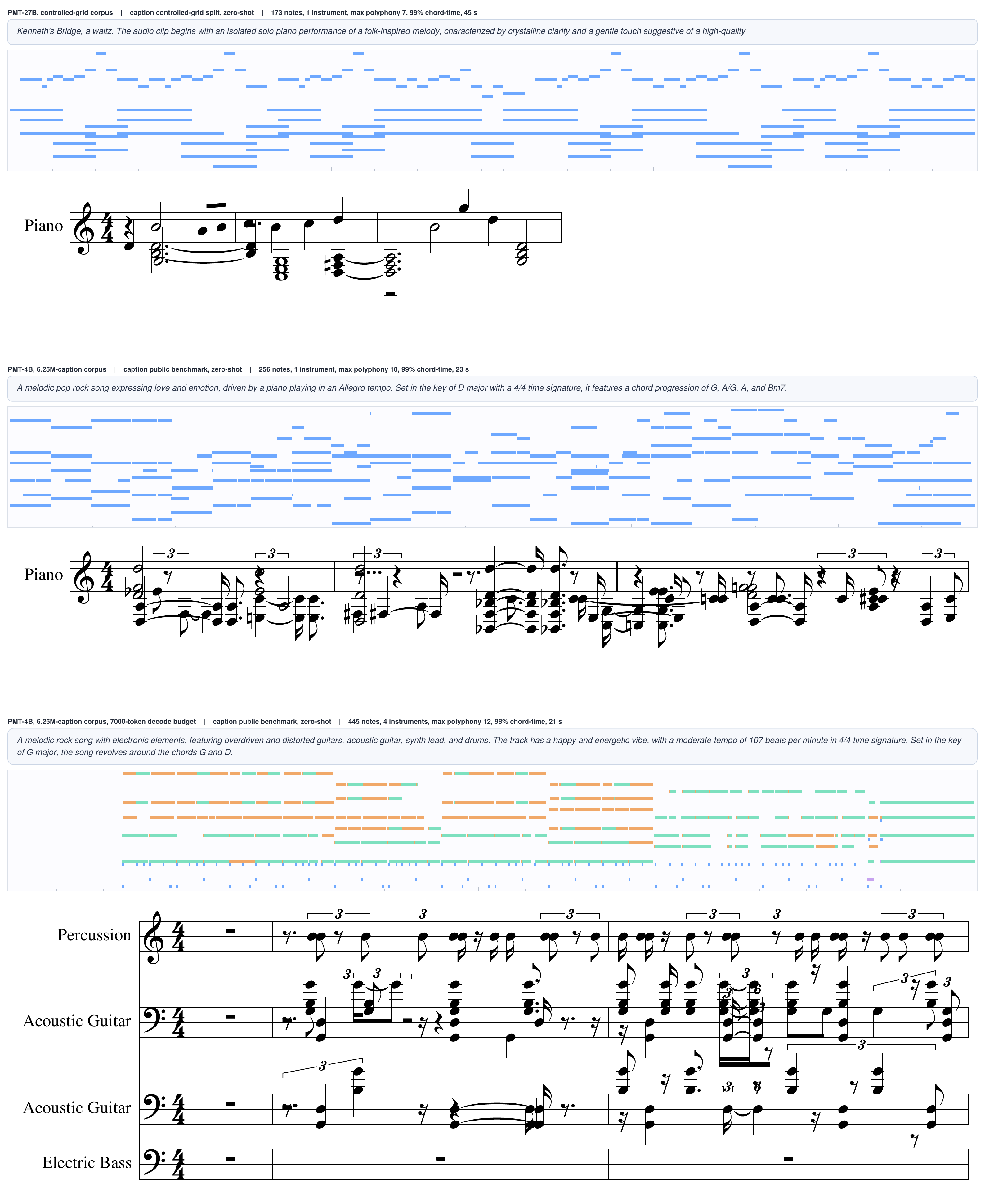}
\caption{\textbf{Generation gallery (10 of 11).} Per example: model and caption provenance, measured
texture, the caption, the piano roll over the whole clip (colour $=$ track, vertical axis $=$ pitch), and
the opening bars engraved. All curated; see Appendix~\ref{app:gallery}.}
\label{fig:galj}
\end{figure*}

\begin{figure*}[tp]
\centering
\includegraphics[width=\textwidth]{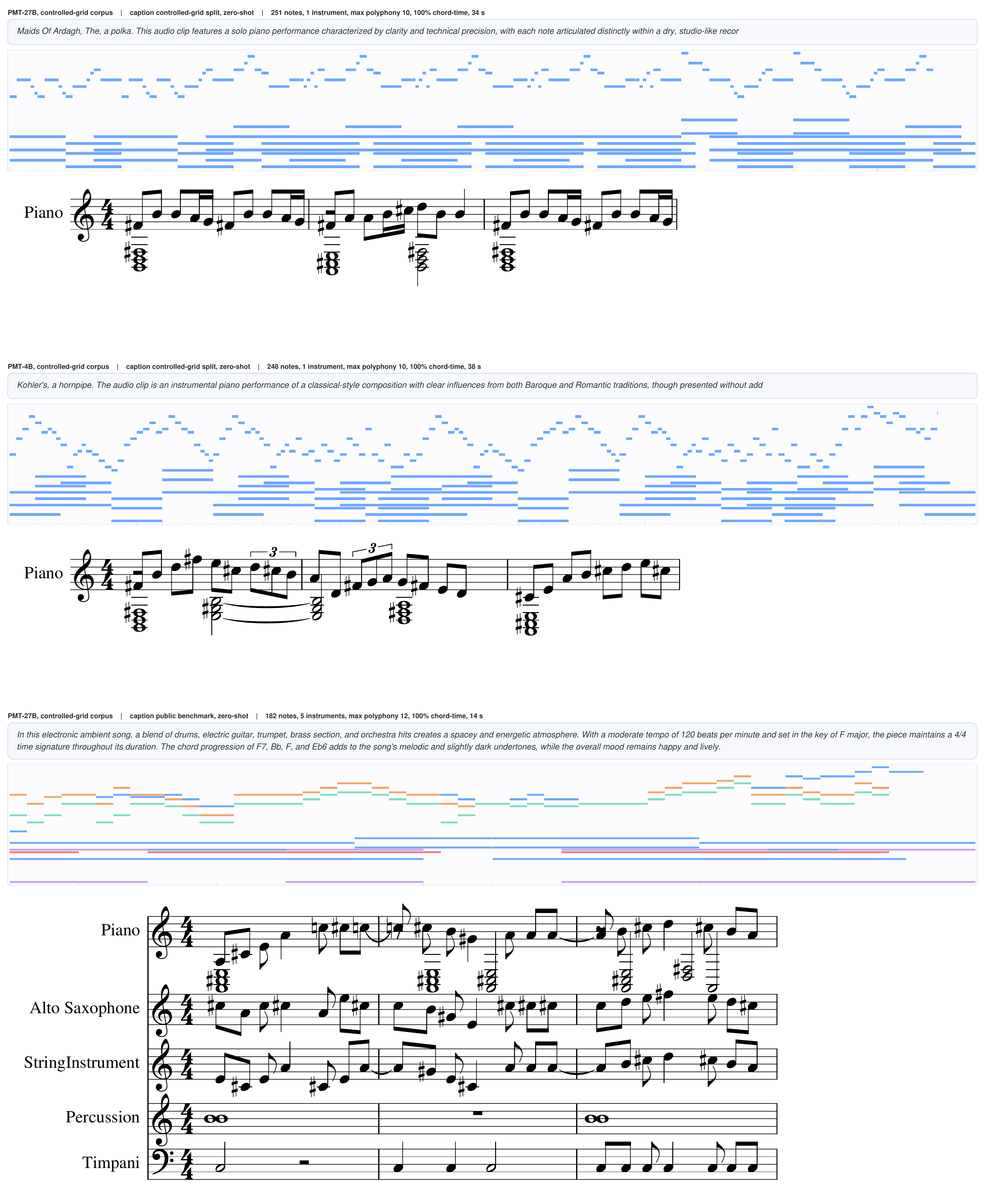}
\caption{\textbf{Generation gallery (11 of 11).} Per example: model and caption provenance, measured
texture, the caption, the piano roll over the whole clip (colour $=$ track, vertical axis $=$ pitch), and
the opening bars engraved. All curated; see Appendix~\ref{app:gallery}.}
\label{fig:galk}
\end{figure*}

\begin{figure*}[tp]
\centering
\includegraphics[width=\textwidth]{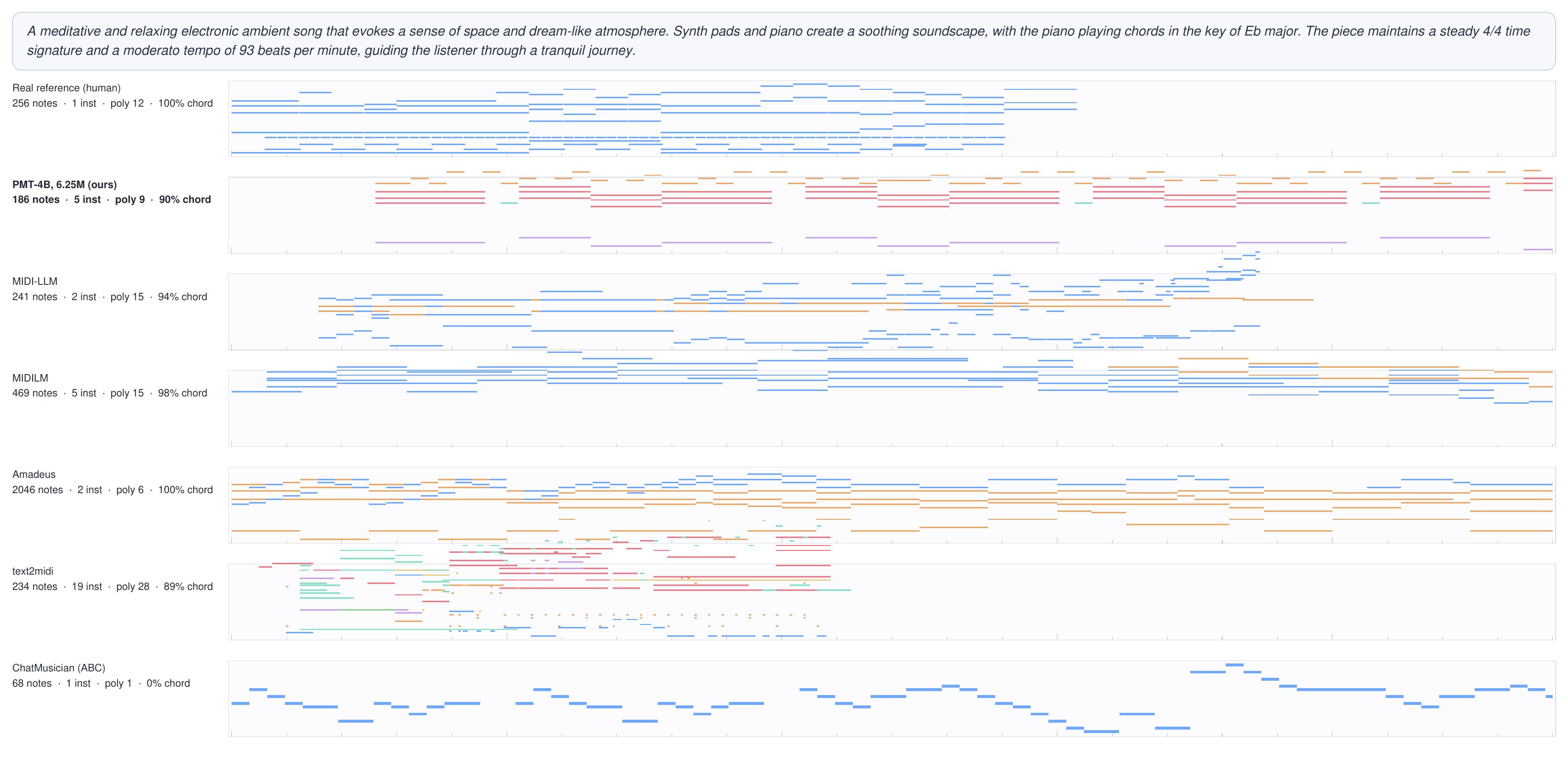}
\caption{\textbf{One caption, seven systems.} Rolls over a common $24$\,s window; row labels carry the
measured texture.}
\label{fig:cmpa}
\end{figure*}

\begin{figure*}[tp]
\centering
\includegraphics[width=\textwidth]{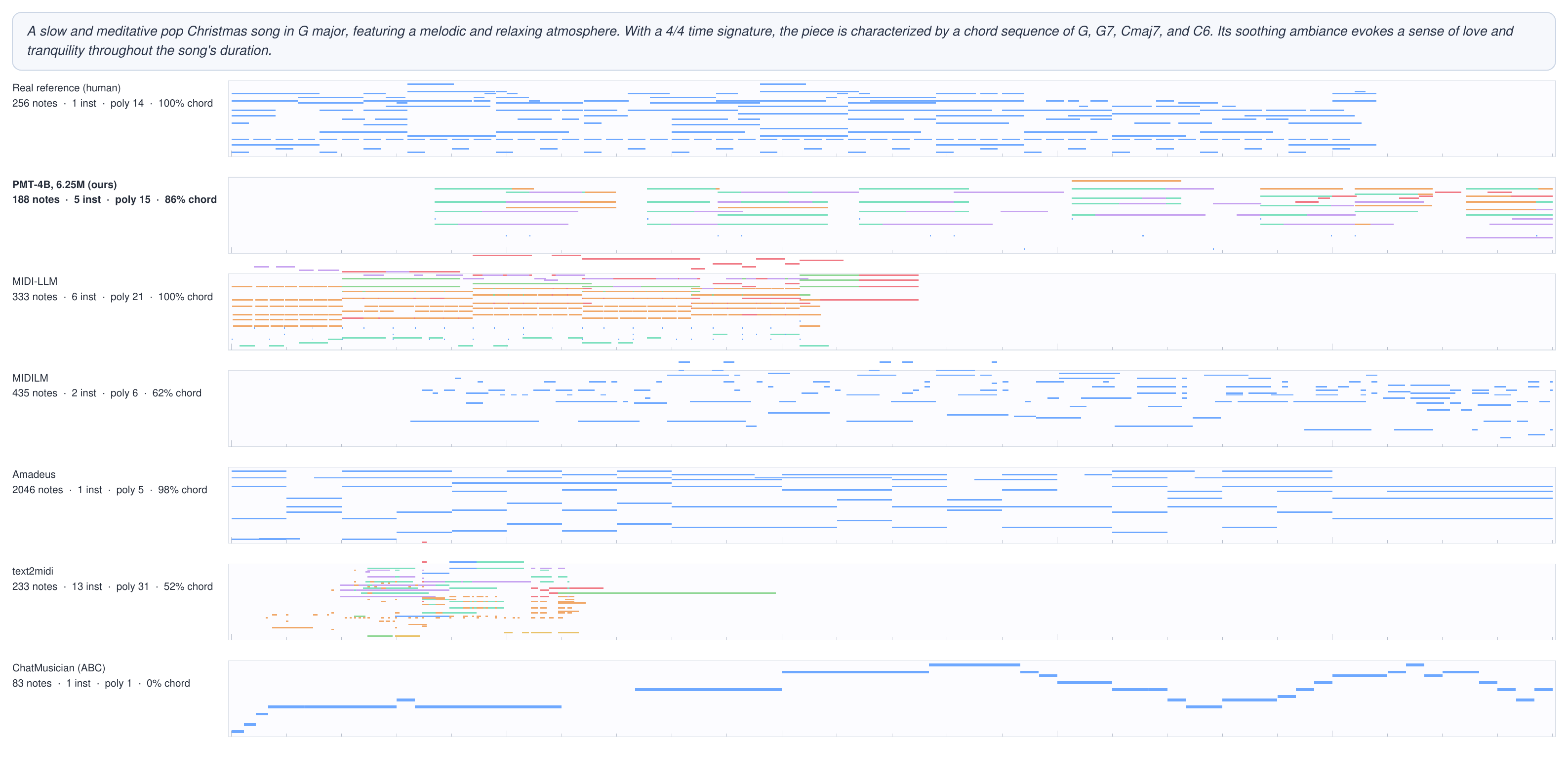}
\caption{\textbf{One caption, seven systems (second caption).} Same protocol as
Figure~\ref{fig:cmpa}. The reference row is clipped at $256$ notes by the benchmark's own construction,
which is the length bias discussed in Appendix~\ref{app:declen}.}
\label{fig:cmpb}
\end{figure*}

\end{document}